\documentclass[11pt]{article}
\usepackage{geometry} 

\usepackage{lineno}

\usepackage[T1]{fontenc}

\usepackage{amsmath}
\usepackage{amsthm}
\usepackage{amsfonts}
\usepackage[colorinlistoftodos,prependcaption,textsize=small]{todonotes}
\usepackage{hyperref}
\usepackage{comment}
\usepackage{algorithm}
\usepackage{algpseudocode}
\usepackage[most]{tcolorbox}
\usepackage{xspace}
\usepackage{appendix}
\usepackage{thm-restate}

\algblock{ParFor}{EndParFor}
\algnewcommand\algorithmicparfor{\textbf{parfor}}
\algnewcommand\algorithmicpardo{\textbf{do}}
\algnewcommand\algorithmicendparfor{\textbf{end\ parfor}}
\algrenewtext{ParFor}[1]{\algorithmicparfor\ #1\ \algorithmicpardo}
\algrenewtext{EndParFor}{\algorithmicendparfor}

\DeclareMathAlphabet{\mathpzc}{OT1}{pzc}{m}{it}

\newtheorem{theorem}{Theorem}[section]
\newtheorem{corollary}[theorem]{Corollary}
\newtheorem{lemma}[theorem]{Lemma}
\newtheorem{question}[theorem]{Question}

\newcommand{\jdef}{ deg_i \cdot c  {\ln n}}

\newcommand{\izo}{\lfloor {\log (\kappa\rho)}\rfloor}
\newcommand{\mch}{\mathcal{H}}

\newcommand{\aismai}{A_i\setminus A_{i+1}}
\newcommand{\hopell}{ \lceil (1/\rho)\rceil }
\newcommand\nfrac{n^{1+\frac{1}{\kappa}}}
\newcommand\epsi{\left({1}/{\epsilon}\right)^i}
\newcommand{\degi}{n^\frac{2^i}{\kappa}}

\newcommand{\apdi}{(1+\epsilon_H)\delta_i}

\newcommand{\eps}[1]{\left({1}/{\epsilon}\right)^{#1}}

\newcommand{\krange}{[k_0,\lambda]}
\newcommand{\emuellb}{{\log (\kappa\rho)} + 1/\rho}
\newcommand{\emuellbk}{{\log (\rho{\log n})} + 1/\rho}
\newcommand{\deltai}{\alpha\cdot \epsi+4R_i}
\newcommand{\betaemu}{ 28\left(\frac{28(\emuellb)}{\epsilon_{M}}\right)^{\emuellb-1}}
\newcommand{\betaemuk}{ 28\left(\frac{28(\emuellbk)}{\epsilon_{M}}\right)^{\emuellbk}}
\newcommand{\betaemuko}{ 28\left(\frac{28(\emuellbk)}{\epsilon_{M}}\right)^{\emuellbk-1}}

\newcommand{\dlimits}{\frac{t}{(1+\epsilon_\mathcal{H})^2(k+1)}}
\newcommand{\addedmemoryparam}{n^\rho\cdot c\cdot {\ln n}}
\newcommand{\addedmemory}{\mu}

\newcommand{\dgmhb}[1]{d^{({#1}\beta_\mathcal{H})}_{G\cup \mathcal{H}}}
\newcommand{\dghBH}{d^{(\beta_{\mathcal{H}})}_{G\cup \mathcal{H}}}

\newcommand{\tval}{\left(  \frac{28(\emuellbk)}{\epsilon_{M}}\right)^{\emuellbk+1}}

\newcommand{\dghb}{d^{(4\beta)}_{G\cup H}}
\newcommand{\emuell}{\lfloor{\log {\kappa \rho}}\rfloor + \lceil \frac{\kappa+1}{\kappa\rho}\rceil-1}

\newcommand{\poly}{\mathop\mathrm{poly}}
\newcommand{\clique}{\textsc{Congested Clique}\xspace}

\title{Massively Parallel Algorithms for Approximate Shortest Paths}

\makeatletter
\renewcommand\@date{{%
  \vspace{-\baselineskip}%
  \large\centering
  \begin{tabular}{@{}c@{}}
    Michal Dory\textsuperscript{1} \\
    \normalsize mdory@ds.haifa.ac.il
  \end{tabular}%
  \qquad \quad
  \begin{tabular}{@{}c@{}}
    Shaked Matar\textsuperscript{2} \\
    \normalsize matars@post.bgu.ac.il
  \end{tabular}

  \bigskip

  \textsuperscript{1}Department of Computer Science, University of Haifa \par
  \textsuperscript{2}Department of Computer Science, Ben-Gurion University of the Negev

  \bigskip

  \today
}}
\makeatother

\begin{document}

 


\maketitle
\begin{abstract}
   We present fast algorithms for approximate shortest paths in the massively parallel computation (MPC) model. 
We provide randomized algorithms that take $\poly(\log{\log{n}})$ rounds in the near-linear memory MPC model. Our results are for unweighted undirected graphs with $n$ vertices and $m$ edges.

Our first contribution is a $(1+\epsilon)$-approximation algorithm for Single-Source Shortest Paths (SSSP) that takes $\poly(\log{\log{n}})$ rounds in the near-linear MPC model, where the memory per machine is  $\tilde{O}(n)$ and the total memory is $\tilde{O}(mn^{\rho})$, where $\rho$ is a small constant. 

Our second contribution is a distance oracle that allows to approximate the distance between any pair of vertices. The distance oracle is constructed in $\poly(\log{\log{n}})$ rounds and allows to query a $(1+\epsilon)(2k-1)$-approximate distance between any pair of vertices $u$ and $v$ in $O(1)$ additional rounds. The algorithm is for the near-linear memory MPC model with total memory of size $\tilde{O}((m+n^{1+\rho})n^{1/k})$, where $\rho$ is a small constant.

While our algorithms are for the near-linear MPC model, in fact they only use one machine with $\tilde{O}(n)$ memory, where the rest of machines can have sublinear memory of size $O(n^{\gamma})$ for a small constant $\gamma < 1$.
All previous algorithms for approximate shortest paths in the near-linear MPC model either required $\Omega(\log{n})$ rounds or had an $\Omega(\log{n})$ approximation.

Our approach is based on fast construction of near-additive emulators, limited-scale hopsets and limited-scale distance sketches that are tailored for the MPC model. While our end-results are for the near-linear MPC model, many of the tools we construct such as hopsets and emulators are constructed in the more restricted sublinear MPC model.
\end{abstract}

\thispagestyle{empty}

\newpage

\tableofcontents

\thispagestyle{empty}

\newpage

\setcounter{page}{1}

\section{Introduction}
\label{sec intro}

Processing massive data is an important algorithmic challenge, which has received a lot of attention in recent years. 
The massively parallel computation (MPC) model \cite{karloff2010model,beame2017communication,GoodrichSZ11} 
is a modern parallel model developed to model large-scale parallel processing settings such as MapReduce \cite{dean2008mapreduce}, Hadoop \cite{white2012hadoop}, Spark \cite{zaharia2010spark}, and Dryad \cite{isard2007dryad}, 
that deal with massive data.
In this model, the input is distributed between a set of machines with a limited local memory of size $s$, that communicate with each other in synchronous rounds. In each round each machine can communicate with other machines, but is limited to send and receive $s$ information in total. The goal is to minimize the number of communication rounds, as in each round a massive amount of data is communicated between the machines.  A central line of research focuses on obtaining fast algorithms for graph problems in MPC, with the goal of obtaining faster algorithms compared to algorithms in the traditional parallel settings. There are several variants of the MPC model that depend on the size of memory per machine. In the \emph{super-linear} MPC model each machine has $s=n^{1+\epsilon}$ memory for a constant $\epsilon$, where $n$ is the number of vertices in the input graph, in the \emph{near-linear} MPC model each machine has $s=\tilde{O}(n)$ memory, and in the \emph{sublinear} MPC model each machine has $s=n^{\gamma}$ memory for a constant $\gamma < 1$. 

\paragraph{Algorithms for the MPC Model.}

The MPC model has received a lot of attention in recent years. A rich line of work led to fast algorithms for various graph problems  
such as minimum spanning tree \cite{lattanzi2011filtering,andoni2014parallel,nowicki2021deterministic}, coloring \cite{chang2019complexity,czumaj2021simple}, matching and maximal independent set \cite{lattanzi2011filtering,czumaj2018round,ghaffari2018improved, behnezhad2019massively,behnezhad2019exponentially,ghaffari2019sparsifying,assadi2019coresets}, minimum cut \cite{lattanzi2011filtering,ghaffari2020massively}, shortest paths and spanners \cite{li2020faster, andoni2020parallel,DinitzN19,DBLP:conf/podc/DoryFKL21,DBLP:conf/spaa/BiswasDGMN21,fischer2022massively} and more. 
The main goal is to obtain very fast algorithms that ideally take sub-logarithmic  or even constant number of rounds. While in the super-linear and linear memory regimes of MPC many problems indeed have constant or $\poly(\log{\log{n}})$ round algorithms, in the sublinear memory regime many important problems such as computing minimum spanning tree or shortest paths are conjectured to require $\Omega(\log{n})$ rounds, based on the 1-cycle vs 2-cycles conjecture.
In this work, we are interested in solving distance related problems in the MPC model.

\paragraph{Distance Computation in MPC.} Despite the long line of work studying graph problems in the MPC model, less is known about distance problems such as approximate shortest paths.
The current known algorithms for approximate shortest paths can be divided into 2 categories:   algorithms that require $\Omega(\log{n})$ rounds or algorithms with $\Omega(\log{n})$-approximation. The first category includes $\poly(\log{n})$-round algorithms for $(1+\epsilon)$-approximate Single-Source Shortest Paths (SSSP) based on PRAM algorithms \cite{li2020faster, andoni2020parallel}, as well as poly-logarithmic algorithms for distance sketches \cite{DinitzN19} and approximate All-Pairs Shortest Paths (APSP) \cite{hajiaghayi2019mapreduce}.
All these algorithms work in the more restricted sublinear MPC model, and work in weighted graphs.
As mentioned above, in sublinear MPC all these problems are conjectured to require $\Omega(\log{n})$ rounds. The second category includes $\Omega(\log{n})$-approximation algorithms based on building graph spanners in near-linear MPC \cite{DBLP:conf/podc/DoryFKL21,DBLP:conf/spaa/BiswasDGMN21,fischer2022massively}. This approach leads to $O(\log{n})$-approximation in $O(1)$ rounds for weighted APSP in the near-linear MPC model. A major open question is to obtain sub-logarithmic approximation in sub-logarithmic number of rounds. This question is already open for unweighted undirected graphs.

\begin{question}\label{question_fastSP}
    Can we obtain $o(\log{n})$ approximation for approximate shortest paths in $o(\log{n})$ rounds in the near-linear MPC model?
\end{question}

\subsection{Our Contribution}\label{sec contribution}

In this work, we answer Question \ref{question_fastSP} in the affirmative, by providing randomized $\poly(\log{\log{n}})$-round algorithms for $O(1)$-approximate shortest paths in the near-linear MPC model for unweighted undirected graphs.

\paragraph{Single-Source Shortest Paths.}
We first study the single-source shortest paths (SSSP) problem where the goal is to compute the distances from a single source. As standard, we use the notation $\tilde{O}(x)$ to hide poly-logarithmic factors in $x$. Our algorithm has the following guarantees.



\begin{theorem} \label{thm_SSSP}
    Given an unweighted, undirected graph $G=(V,E)$ on $n$ vertices, a parameter $\epsilon <1$ and a constant $\rho \in[ 1/{\log {\log n}}, 1/2]$, there is a randomized algorithm that computes $(1+\epsilon)$-approximation for SSSP with high probability (w.h.p.). The algorithm works
    in the near-linear MPC model using $\tilde{O}((|E|+n^{1+\rho})\cdot n^{\rho})$ total memory.
     The round complexity of the algorithm is $T(n) =  \tilde{O}\left( \frac 
        {\log{\log{n}} }
        {\epsilon\rho}\right)^{\frac{1}{\rho} +2}.$ 
\end{theorem}

As mentioned above, all the previous algorithms for SSSP, even in unweighted undriected graphs, require $\Omega(\log{n})$ rounds or $\Omega(\log{n})$ approximation. In particular, our algorithm is exponentially faster compared to previous algorithms for $(1+\epsilon)$-approximate SSSP \cite{li2020faster, andoni2020parallel} that worked also for weighted graphs. We remark that the total memory used by our algorithm is slightly super-linear, where previous algorithms used $\tilde{O}(|E|)$ total memory. However, this slight increase in total memory allows us to obtain significantly better running time or approximation compared to previous algorithms.

While our algorithm works in the near-linear memory regime, in fact it only uses one machine with near-linear memory, where the rest of machines have sublinear memory of size $O(n^{\gamma})$ for a small constant parameter $\gamma < 1$. Hence it works in the heterogeneous MPC regime defined in \cite{fischer2022massively}. Having one machine with near-linear memory is necessary, as in the sublinear MPC model the problem is conjectured to require $\Omega(\log{n})$ rounds.

Our algorithm can be easily extended to compute $(1+\epsilon)$-approximate distances from a set $S$ of sources of size $O(n^{\rho})$ in the same time and with the same total memory.  In this case we need $|S|$ machines with near-linear memory, where the rest of machines have subliner memory.

\paragraph{All-Pairs Shortest Paths.}
Our next goal is to approximate the distances between all pairs of vertices. Note that storing all distances explicitly requires $\Omega(n^2)$ total memory, and our goal is to find a sparser implicit representation that is more suitable for the MPC model. In particular, our goal is to construct a \emph{distance oracle}, a sparse data structure that allows to query the distance between any pair of vertices efficiently.  Implicit representations of the distances in MPC were also studied before. For example, the previous works \cite{DBLP:conf/podc/DoryFKL21,DBLP:conf/spaa/BiswasDGMN21,fischer2022massively} construct a multiplicative weighted $\Omega(\log{n})$-spanner of $\tilde{O}(n)$ size that can be stored in one machine that can then compute approximate distance between any pair of vertices when needed. Another work \cite{DinitzN19} studies computation of distance sketches that allow to query approximate distances in $O(1)$ rounds in weighted graphs. Here the goal is to compute a small sketch per vertex, such that given the sketches of a pair of vertices one can compute their approximate distance.

Note that Theorem \ref{thm_SSSP} already allows to query $(1+\epsilon)$-approximate distance between any pair of vertices $u,v$ in $\poly(\log{\log{n}})$ rounds by just computing SSSP from $u$ or $v$. 
We next design algorithms that have smaller $O(1)$ query time, with the following guarantees. 

\begin{theorem} \label{thm_APSP_intro}
    Let $G=(V,E)$ be an unweighted, undirected graph on $n$ vertices, and let $\epsilon <1/2$, $\rho \in[ 1/{\log {\log n}}, 1/2]$ and $2 \leq k \leq 1/\rho$ be parameters. 
    There is a randomized algorithm that w.h.p.  computes a distance oracle of size $\tilde{O}(kn^{1+1/k})$ that provides $(1+\epsilon)(2k-1)$-approximation for all the distances.
    The preprocessing time for constructing the oracle is $T(n) =  \tilde{O}\left( \frac 
        {\log{\log{n}} }
        {\epsilon\rho}\right)^{\frac{1}{\rho} +2}$ rounds. Given the oracle, we can query each distance in $O(1)$ rounds. The algorithm works in the near-linear MPC model using $\tilde{O}((|E|+n^{1+\rho})n^{1/k})$ total memory. 
\end{theorem}

Compared to the distance sketches of \cite{DinitzN19}, our algorithm is exponentially faster running in $\poly(\log{\log{n}})$ rounds, where the algorithm in \cite{DinitzN19} takes at least poly-logarithmic number of rounds. Compared to the $O(\log{n})$-approximations in \cite{DBLP:conf/podc/DoryFKL21,DBLP:conf/spaa/BiswasDGMN21,fischer2022massively}, our algorithm gives a better $O(1)$-approximation. We remark that the previous works also handle weighted graphs.
This algorithm also uses only one machine with near-linear memory, where the rest of machines can have sublinear memory. 

\paragraph{Optimizing the total memory.}
As common in the construction of distance oracles, our approach offers a tradeoff between the approximation, round complexity and total memory requirements of the algorithm. We can further improve the total memory requirements by running our algorithms on a spanner, which leads to $O(1/c)$-approximation for SSSP in $\tilde{O}(|E|+n^{1+c})$ total memory, or $O(1/c^2)$-approximation in our distance oracle in the same $\tilde{O}(|E|+n^{1+c})$ total memory, where $c$ can be an arbitrarily small constant. In particular, if $|E|=\Omega(n^{1+c})$ for a constant $c>0$, the total memory becomes $\tilde{O}(|E|)$.

To do so, we start by constructing an $O(1/c)$-spanner of size $O(n^{1+c/2})$, denote it by $G'=(V,E')$. $G'$ is a subgraph of $G$ that preserves all distances up to an $O(1/c)$-approximation, and it is known that such spanners can be constructed in $O(1)$ rounds in the MPC model, see \cite{DBLP:conf/podc/DoryFKL21,fischer2022massively}. Then, we can just run the algorithm from Theorem \ref{thm_APSP_intro} on the graph $G'$ choosing $\rho=c/2, k = 2/c$. This leads to an $O(1/c)$-approximation of the distances in $G'$, which is an $O(1/c^2)$-approximation of the distances in $G$, where the total memory requirements are $\tilde{O}(|E|+(|E'|+n^{1+\rho})n^{1/k})= \tilde{O}(|E|+(n^{1+c/2})n^{c/2}) = \tilde{O}(|E|+n^{1+c})$. 
Similarly, if we run the algorithm from Theorem \ref{thm_SSSP} on the spanner $G'$, we get a $(1+\epsilon)$-approximation of the distances in $G'$, which is an $O(1/c)$-approximation of the distances in $G$ in $O(|E|+n^{1+c})$ total memory.


\section{Technical Overview} 

Our approach is based on efficient constructions of hopsets, emulators and distance sketches that we show that are tailored for the MPC model. While these structures were studied before in many computational models, implementing previous algorithms in MPC would be too expensive in terms of memory or running time, as we discuss next. To overcome it, we show a new framework that allows us to obtain fast and low-memory algorithms. 
We next describe our approach in more detail.

\paragraph{Dealing with Short Paths via Hopsets.} A hopset is a set of weighted edges that is added to a graph, such that after adding these edges, there are low-hop paths between any pair of vertices. Specifically, a set $H$ is considered a $(1+\epsilon_H,\beta_H)$-hopset for a graph $G$, if for every pair of vertices $u,v\in V$, the graph $G' = (V,E\cup H)$  contains a path of at most $\beta_H$ hops approximating $d_G(u,v)$ by a factor of $(1+\epsilon_H)$. Hopsets are key in approximate shortest path algorithms across various computational models, enabling focus on paths with a small number of hops. In particular, explorations from a source vertex $s\in V$ can be terminated after $\beta_H$ hops, providing near-exact distance approximations.

Previous hopset algorithms have hopbound $\beta_H$ and a running time polylogarithmic in $n$, rooted in the need to explore the graph $G$ to depth $\Omega(D)$, where $D$ is the diameter of the graph. To avoid $\Omega(D)$ running time, hopset computation can be decomposed into ${\log D} \leq {\log n}$ tasks of computing \textit{limited-scale} hopsets. Specifically, for every $k\in [0,{\log n}]$, one can compute a hopset $H_k$ for the scale $(2^k,2^{k+1}]$ by exploring the graph to a depth of approximately $O(2^{k+1})$. Leveraging lower-scale hopsets allows constructing $H_k$ using only explorations to $O(\beta_H)$ hops. The union of all hopsets $\{ H_k \ | \ k\in [0,{\log n}]$ is a \textit{full-scale} hopset.

Dividing distances into ${\log n}$ scales and addressing them individually results in $\Omega({\log n})$ running time, which is too expensive for our needs as we aim for a sub-logarithmic round complexity. 
To overcome it, we use hopsets only to deal with short paths. More concretely, we show that we can build in $\poly({\log t})$ time a limited-scale hopset that provides $(1+\epsilon_H)$-approximations for pairs of vertices at distance up to $t$. To do so, we build hopsets for the first ${\log t}$ scales, each having hopbound $\beta_H = \poly({\log t})$. We will choose $t$ such that $\poly({\log t}) = \poly({\log {\log n}})$, which results in limited-scale hopsets that can be constructed in $\poly({\log {\log n}})$ time.
While limited-scale hopsets solve part of the problem, they lack guarantees for distant pairs of vertices. To address this, we use near-additive emulators.

\paragraph{Dealing with Long Paths via Near-Additive Emulators.}

The problem of computing near-exact approximations for distant pairs of vertices can be solved using near-additive emulators. Given an unweighted undirected graph $G=(V, E)$, a near-additive emulator $M=(V,E')$ is a sparse graph with the same vertex set that provides $(1+\epsilon_M,\beta_M)$-approximation for the distances. Concretely, for any pair of vertices $u,v$, we have $$d_G(u,v) \leq d_H(u,v) \leq (1+\epsilon_M) d_G(u,v) + \beta_M.$$ If $H$ is a subgraph of $G$, then it is called a spanner, but in general, $H$ may have edges that are not part of $G$, and it can be weighted even though $G$ is unweighted.

Intuitively, the reason that near-additive emulators are useful for near-exact approximation of shortest paths is that, for pairs of vertices that are far away (of distance $\Omega(\beta_M/\epsilon_M)$ from each other), the additive term becomes negligible, and the approximate distance in the emulator is already a $(1+O(\epsilon_M))$-approximation. 
    
Our goal is to construct an emulator of size $\tilde{O}(n)$, since we can store it in one machine that can then locally compute distances in the emulator. Existing emulators of this size have $\beta_M=(\log{\log{n}})^{O(\log{\log{n}})}$, see, e.g., \cite{ElkinP01,ThorupZ06,Pettie10,DBLP:journals/talg/ElkinN19}. This is also close to optimal by existential lower bounds \cite{abboud2018hierarchy}. Constructing emulators requires exploring the graph to depth $\beta_M$, which can lead to an  $\Omega(\beta_M)$ round complexity, i.e., super-logarithmic and too expensive for our needs. While emulators have been widely studied in various computational models (see, e.g., \cite{ElkinP01,ThorupZ06,Pettie10,DBLP:journals/talg/ElkinN19,DBLP:conf/podc/ElkinM19,DBLP:conf/podc/DoryP20,DBLP:journals/eatcs/ElkinN20,DBLP:conf/podc/ElkinM21,bergamaschi2021new,DBLP:conf/approx/ElkinT22}), most existing constructions of emulators with $\tilde{O}(n)$ size require $\Omega(\beta_M)$ time.\footnote{The only exception we are aware of is the \clique algorithm from \cite{DBLP:conf/podc/DoryP20} that also takes $\poly({\log {\log n}})$ time as we discuss in detail later. However, this algorithm requires a lot of memory and hence does not lead to an efficient MPC algorithm.} 

To achieve a goal of $\poly({\log {\log n}})$ MPC time, we suggest a different approach. Instead of constructing the emulator directly from the graph, we leverage previously created limited-scale hopsets designed for approximating short distances. These hopsets offer near-exact approximations for distances up to $t = {({\log{\log n}}/\epsilon_M)}^{O({\log {\log n}})}$ with a hopbound $\beta_H = \poly(\log{t}) = \poly({\log {\log n}})$. By carefully adjusting the parameters of both the hopset $H$ and the emulator $M$, we can utilize the hopset $H$ to construct $M$ in $\poly({\log {\log n}})$ time.

To summarize, we combine the strengths of near-exact hopsets and near-additive emulators as complementary products. Utilizing the hopsets tailored for short distances, constructed in $\poly({\log {\log n}})$ time, we efficiently build sparse emulators approximating long distances in the same round complexity. Then, in order to compute shortest paths from a source $s$, we aggregate the edges of the emulator into a single machine $M^*$ with $\tilde{O}(n)$ memory. The machine $M^*$ locally computes distances from $s$ in the emulator. This gives $(1+O(\epsilon_M))$-approximation for  pairs of vertices with distance at least $\beta_M/\epsilon_M$.  Additionally, in $O(\beta_H)$ time, $(\beta_H)$-hops limited distances from $s$ in $G\cup H$ are computed, which gives a $(1+\epsilon_H)$-approximation for pairs of vertices with distance at most $\beta_M/\epsilon_M$. For each vertex $v\in V$, the minimal distance estimate is selected. 
In fact, in the same round complexity and total memory we can compute $(1+\epsilon)$-approximate distances from a set of sources of size $O(n^{\rho})$. In this case, we require additional $O(n^\rho)$ machines with $\tilde{O}(n)$ memory.


\paragraph{Constructing Hopsets and Emulators.}

Traditionally, hopsets and emulators have been treated as distinct structures, despite a recognized strong connection between the two (see \cite{DBLP:journals/eatcs/ElkinN20} for a survey). This connection is evident in the shared methodologies for constructing near-exact hopsets and near-additive emulators,  leading to similar characteristics in size and even the parameter $\beta$, although its meaning differs in the two contexts. Moreover, the analyses of algorithms for constructing these structures exhibit notable similarities. While it is well known that there is a connection between near-exact hopsets and near-additive emulators, its nature is not yet fully understood. We provide a generic algorithm, inspired by the construction of Thorup and Zwick \cite{ThorupZ06}, which produces a generic structure.  This algorithm can be used in a black-box manner to produce both near-exact hopsets and near-additive emulators, by simply adjusting its input parameters.

In particular, our generic algorithm is composed of two building blocks. The first building block is a procedure that samples a hierarchy of vertices $\mathcal{A} = A_0 \supseteq\dots \supseteq A_{\ell+1}$ where $A_0 = V$ and $A_{\ell+1} = \emptyset$, for a suitable parameter $\ell$. Given a degree sequence $\{deg_i \ | \ i\in [0,\ell-1]\}$, each vertex $v\in A_i$ is sampled to $A_{i+1}$ with probability $1/deg_i$, for all $i\in [0,\ell-1]$. 

The second building block is a procedure that, given a hierarchy $\mathcal{A}$, selects edges to form a set $Q$. This edge selection procedure draws inspiration from \cite{ThorupZ06}. 
In essence, for every $i\in [0,\ell-1]$ and for each vertex $v\in A_i\setminus A_{i+1}$, if $v$ has an $A_{i+1}$-vertex $w$ in its vicinity, then $v$ selects $w$ to be its \textit{pivot} and adds the edge $(v,w)$ to $Q$. Otherwise, the vertex $v$ adds to $Q$ edges to all vertices $u\in A_i$ that are in close proximity to $v$. This scenario only occurs if the vertex $v$ does not have any proximal vertex $w\in A_{i+1}$, indicting that w.h.p., $v$ has a small number of proximal $A_{i}$ vertices.

\paragraph{MPC Implementation.}
Implementing the edge selection procedure requires executing several Bellman-Ford explorations to a certain depth. To expedite computation, we provide the edge selection procedure with a suitable limited-scale $(1+\epsilon_H,\beta_H)$-hopset $H$. 
In particular, for every $i\in [1,\ell]$ our algorithm first executes $(\beta_H)$-rounds of a Bellman-Ford exploration in the graph $G\cup H$ from all vertices in $A_{i+1}$. Consequently, each vertex $v\in A_i$ that has a proximal $A_{i+1}$-vertex can select such a vertex to be its pivot $w\in A_{i+1}$, and add the edge $(v,w)$ to the set $Q$. We remark that this vertex does not have to be the closest $A_{i+1}$ vertex to $v$. 
Subsequently, we are left with $A_i$-vertices that did not chose a pivot, indicating that w.h.p. they do not have many nearby $A_i$-vertices. Connecting these vertices with their $A_i$-neighborhood requires executing $(\beta_H)$-hops restricted Bellman-Ford explorations in $G\cup H$ from them in parallel. 
While we may execute $\Omega(n)$ explorations simultaneously, we show that w.h.p., each vertex $v\in V$ is traversed by at most $\tilde{O}(n^\rho)$ explorations in parallel. 
Thus, we are able to provide an efficient sublinear MPC algorithm that computes these explorations in parallel, utilizing $\tilde{O}(n^\rho)$ memory per edge.

To construct hopsets, we employ the sampling procedure once and, for $O({\log t})$ consecutive phases, execute the edge selection procedure. Each execution $k$ computes a hopset $H_k$ for the scale $(2^k,2^{k+1}]$. Maintaining the same hierarchy across all executions enables the provision of sparse hopsets. The hopset $H = \bigcup_{k'\in [0,k]}H_{k'}$ is the limited-scale hopset provided to execution $k+1$, allowing all explorations to be performed in $\poly({\log {\log n}})$ time in sublinear MPC. Crucially, we maintain a $\poly({\log {\log n}})$ hopbound by customizing the sampling probability to ensure a hierarchy with only a constant number of levels. 
In particular, given a small constant $\rho \in [1/{\log {\log n}},1/2]$, we compute a $(1+\epsilon,\beta)$-hopset of size $\tilde{O}(n^{1+\rho})$, and hopbound $\beta = O\left( {{\log t}}/{\epsilon_H}\right)^{1/\rho} $, for distances up to $t$. 

To construct our emulator, we use the sampling and edge selection procedure once, incorporating the limited-scale hopset $H$ to expedite computations. Notably, the emulator must be aggregated into a single machine with $\tilde{O}(n)$ memory. To achieve a sparse emulator, even at the cost of super-logarithmic additive stretch, we employ a degree sequence that guarantees sparsity. This emulators construction requires $\poly({\log {\log n}})$ sublinear MPC rounds. 

While our end-results are for the near-linear MPC model because we collect the emulator in one machine, our hopsets and emulators are constructed in the more restricted sublinear MPC model. Since our algorithm is efficient both in terms of running time and memory usage we believe that our approach can be useful also for other computational settings with limited resources.

\paragraph{Comparison to Previous Approaches.} 

While previous constructions view hopsets and emulators as different structures, 
we provide a unified algorithm that produced a generic object, that can be viewed as either a hopset or an emulator. 
Analyzing the intrinsic properties of this generic structure independently of its application for hopsets or emulators provides valuable insights, shedding light on the underlying connection between near-exact hopsets and near-additive emulators.
We remark that understanding the connection between hopsets and emulators receives a lot of interest in recent years (see \cite{DBLP:journals/eatcs/ElkinN20} for a survey). While it is well-known that similar approaches can be used to construct hopsets and emulators, the full connection between these objects is still not fully understood. A recent work \cite{DBLP:conf/focs/KoganP22} shows a reduction in one direction: one can construct emulators in a black-box manner using algorithms for full-scale hopsets and multiplicative spanners. Note that constructing full-scale hopsets is a global problem, and all current constructions of full-scale hopsets require $\Omega(\log{n})$ rounds, so using the reduction in this direction does not lead to an efficient MPC algorithm. Here we use the more local nature of emulators to construct them faster. We provide a different perspective on the connection between hopsets and emulators, by providing a unified object that can be viewed both as a hopset and an emulator.

As we mentioned above, near-additive spanners and emulators have been widely studied in various computational models, leading to centralized, distributed, PRAM, streaming and dynamic algorithms \cite{ElkinP01,ThorupZ06,Pettie10,DBLP:journals/talg/ElkinN19,DBLP:conf/podc/ElkinM19,DBLP:conf/podc/DoryP20,DBLP:journals/eatcs/ElkinN20,DBLP:conf/podc/ElkinM21,bergamaschi2021new,DBLP:conf/approx/ElkinT22}. 
Almost all of the existing constructions of near-additive emulators need time that is at least linear in $\beta_M$  (which is super-logarithmic and too expensive for our needs). One exception is the \clique algorithm from \cite{DBLP:conf/podc/DoryP20} that builds emulators and computes approximate shortest paths in $\poly(\log{\log{n}})$ rounds. While the linear-memory MPC model and the \clique model are closely related, a major difference between the models is that there are no explicit memory requirements in the \clique model. Indeed, the algorithm from \cite{DBLP:conf/podc/DoryP20} computes near-additive emulators by using graph exponentiation techniques which allow vertices to learn $t$-neighborhoods around them in $\poly(\log{t})$ rounds. However, this approach uses a lot of memory, and a direct application of it in the MPC model would require $\Omega(n^2)$ total memory, which is too expensive. To overcome it, we show a different construction of near-additive emulators that is both efficient and uses low memory. We believe that our approach can found future applications in other computational models with limited resources.

\paragraph{All Pairs Approximate Shortest Paths (APSP)}

As discussed in Section \ref{sec contribution}, we aim at constructing a sparse data set in $\poly({\log {\log n}})$ time, such that given a query $(u,v)$, one can retrieve a $(1+\epsilon)(2k-1)$ approximation for the distance $(u,v)$ in $O(1)$ time. We remark that our emulators provide a satisfactory solution for pairs of vertices $u,v$ with distance $d(u,v) = \Omega(t/k)$. However, computing distances using the hopset requires $O(\beta_H)$ time. Thus, the hopsets cannot be used directly to answer distance queries in $O(1)$ time, and close pairs of vertices require a different approach. 

To provide distance approximations for close pairs of vertices, we provide an algorithm designed to construct limited-scale distance sketches. These sketches offer $(2k-1)(1+\epsilon)$-distance estimates for pairs of vertices with distances up to $O(t/k)$. Essentially, our algorithm can be seen as a modification of the Thorup and Zwick \cite{ThorupZ05} algorithm for distance sketches. We note that a previous  MPC implementation of the Thorup-Zwick algorithm was presented by Dinitz and Nazari \cite{DinitzN19}. However, their algorithm has polylogarithmic round complexity since they build a full-scale hopset in polylogarithmic time and then use this hopset to build distance sketches. Given our objective of achieving a time complexity of $\poly({\log {\log n}})$, building a full-scale hopset becomes impractical. Nevertheless, we demonstrate that a limited-scale hopset is sufficient for constructing limited-scale distance sketches for distances up to $O(t/k)$.  

Using hopsets to construct distance sketches, and in particular using limited-scale hopsets, presents certain challenges, for which we have developed algorithmic solutions. First, the stretch analysis in the algorithm of Thorup and Zwick relies on the triangle inequality which holds for distances in $G$. However, the $\beta_H$-hops limited distances in the graph $G\cup H$ do not adhere to the triangle inequality. Consequently, using a hopset implies that we can no longer directly apply the stretch analysis of Thorup and Zwick. To address this, we customize the number of hops in each exploration to some characteristics of the origin of the exploration. This modification enables us to get $(1+\epsilon_H)(2k-1)$ approximations. 

In addition, there is a spacial challenge in using only a limited-scale hopset. The correctness of the algorithm of Thorup and Zwick hinges on the existence of a set of vertices $A$ such that all vertices in $V$ know their distances to all vertices in $A$. This requires  explorations of the graph to a depth of $\Omega(D)$ (where $D$ is the diameter of the graph $G$), which is impractical in our scenario. 
To overcome this challenge, we provide a new stretch analysis, and show that in order to answer distance queries regarding vertices with distance up to $O(t/k)$ from one another, it suffices to explore only the $t$-neighborhood of each vertex $v\in V$ in the preprocessing step. Thus, these explorations can be conducted in $\poly({\log {\log n}})$ time using the limited-scale hopset $H$. 

To the best of our knowledge, previous MPC algorithms for the APSP problem provide a single data structure, that deals with all distance queries \cite{DinitzN19,DBLP:conf/podc/DoryFKL21,DBLP:conf/spaa/BiswasDGMN21,fischer2022massively}. Our approach involves two distinct structures: a limited-scale distance sketch and an emulator. While neither structure handles all queries independently, their combined construction is exponentially faster compared to the distance sketches in \cite{DinitzN19}, and the approximation provided is constant which improves over the $\Omega(\log{n})$-approximation provided by computing a multiplicative spanner of size $\tilde{O}(n)$ \cite{DinitzN19,DBLP:conf/podc/DoryFKL21,DBLP:conf/spaa/BiswasDGMN21,fischer2022massively}.
We remark that our limited-scale distance sketches are constructed in the sublinear MPC model, and the only reason we need one machine with $\tilde{O}(n)$ memory is to collect the emulator and compute distance estimates in the emulator.

\subsection{Outline}
Section \ref{sec prel} contains definitions and details of the computation model. 
In Section \ref{sec general}, we devise the general framework for producing both near-exact hopsets and near-additive emulators. Sections \ref{sec hop} and \ref{sec emu} show how the general framework is used to compute near-exact hopsets and near-additive emulators, respectively. In Section \ref{sec applications} we discuss the applications of the general framework to shortest paths problems. Our distance oracles construction is given in Section \ref{sec APSP}.

\subsection{Preliminaries}\label{sec prel}

\paragraph{The Model}
We consider the Massively Parallel Computation, or
MPC model. The MPC model was introduced in \cite{karloff2010model} and refined by \cite{beame2017communication,GoodrichSZ11}. It models MapReduce and other realistic distributed settings. In the model, there is an input of size $N$. The input is arbitrarily distributed over 
a set of machines, each with $S = N^\epsilon$ memory for some $\epsilon \in (0,1)$.  
The number of machines is usually close to $\tilde{O}(N/S)$ which means the total size of memory is $\tilde{O}(N)$, but having an additional extra memory is also considered. For graph problems, the size of the input is $O(m)$ words, where $m=|E|$ is the number of edges in the input graph. 
There are several prevalent settings for MPC, which differ by the amount of memory each machine is allowed to have. In the $MPC(n^\gamma)$ low memory setting, each machine is allowed to have at most $O(n^\gamma)$ memory, for $\gamma\in (0,1)$ and $n = |V|$. 

We use three variants of MPC. 
First, we consider a variant of the low memory MPC model with extra total space. The memory per machine is still $O(n^\gamma)$ and the number of machines is $\tilde{O}(\frac{mn^\rho}{n^\gamma})$. Consequently, the total memory used by all machines is $\tilde{O}(mn^\rho)$, where $\rho\in [1/{\log {\log n}},1/2]$ is a small constant. We require $\gamma \in (0,1)$.

The second variant is the \textit{heterogeneous MPC setting}, defined by \cite{fischer2022massively}. The heterogeneous setting is created by adding a single near-linear machine to the sublinear MPC regime. We use this model in conjunction with the extra space MPC, i.e., we have one machine with $\tilde{O}(n)$ memory, and $\tilde{O}(\frac{mn^\rho}{n^\gamma})$ machines with $O(n^\gamma)$ memory each.


\section{The General Framework}\label{sec general}

In this section, we present an algorithm that selects a set of weighted edges $Q$.  In Sections \ref{sec hop} and \ref{sec emu}, we demonstrate that this algorithm can be used to effectively generate the desired near-exact hopset and near-additive emulator, respectively.

The input for our algorithm is an unweighted, undirected graph $G=(V, E)$  on $n$-vertices, and parameters $\epsilon <1$, $\rho \in [1/{\log {\log n}},1/2]$, $\kappa \geq 1/\rho$, and $\alpha >0$. The output of the algorithm is a set of edges $Q$. The parameter $\epsilon$ determines the multiplicative stretch of $Q$, while $\rho$ governs the memory usage of the algorithm. The parameter $\kappa$ dictates the size of the set $Q$. To carry out explorations efficiently, we assume access to a $(1+\epsilon_H, \beta)$-hopset $H$ for distances up to 
 $\frac{1}{2}\alpha\eps{\ell}$ in the graph $G$ for a suitable parameter $\ell$.
We remark that if one chooses not to use a hopset, then the empty set can trivially serve as a  $(1,\frac{1}{2}\alpha\eps{\ell})$-hopset for distances up to $\frac{1}{2}\alpha\eps{\ell}$.

We now provide a high-level overview of the framework, drawing inspiration from the construction of near-exact hopsets and near-additive spanners and emulators by Thorup and Zwick \cite{ThorupZ06}.
The underlying idea is to add to $Q$ edges from each vertex to all others in its proximity. However, creating edges between all vertices in such a manner may result in an overly dense set. Since we aim at a sparse set $Q$, we do not allow all vertices to connect with all vertices in their vicinity. Instead, we first compute a hierarchy of vertices. Then, vertices with a higher-sampled counterpart in their proximity connect exclusively with a single such higher-level vertex. The remaining vertices establish connections with proximal vertices in their own level. The rationale behind this approach lies in the expectation that vertices lacking a proximal higher-sampled vertex are likely to belong to a sparser neighborhood.

Formally, the algorithm is composed of two building blocks. The first one is a procedure that samples a hierarchy $\mathcal{A} = A_0\supseteq A_1 \supseteq\dots\supseteq A_\ell \supseteq A_{\ell+1}$ where $A_0 = V$ and $A_{\ell+1} = \emptyset$, for a suitable parameter $\ell$ that will be specified in the following section. The second building block is a procedure that, given a hierarchy $\mathcal{A}$, computes a set of edges $Q$. 

Given a hierarchy of vertices  $\mathcal{A}$, the edge selection procedure operates as follows. For a vertex $u \in \aismai$, if there exists a vertex $v \in A_{i+1}$ in proximity to $u$, then $u$ adds an edge only to such a specific vertex $v$. In this case, the responsibility for managing the vertices in the vicinity of $u$ is entrusted to $v$. Essentially, $v$ considers the nearby vertices of $u$ as part of its own vicinity.

On the other hand, if there is no vertex $v \in A_{i+1}$ in the vicinity of $u$, then $u$ must independently manage its connections. In this scenario, $u$ adds to $Q$ edges to all vertices from $A_i$ that are within its proximity. This encapsulates the fundamental concept of the general framework


\subsection{Sampling}\label{sec samp gen}
Let $\ell$ be a parameter, and let $\{deg_i \ |\ i\in [0,\ell]\}$ be a degree sequence. We specify their values shortly. 
First, we define $A_0 = V$. Then, for every $i\in [0,\ell-1]$, each vertex of $A_{i}$ is sampled into $A_{i+1}$ with probability $1/deg_i$.  
In addition, we define $ A_{\ell+1}= \emptyset$. The pseudo-code of the sampling procedure is given in Algorithm \ref{alg sample}.

For every $i\in [0,\ell]$, setting $deg_i = \degi$ ensures that the overall size of the set $Q$ is roughly $O(\nfrac)$. However, during our algorithm we will work in phases where the execution of each phase $i$ requires using $\tilde{O}(deg_i)$ memory per edge in $E\cup H$. Therefore, we restrict $deg_i$ to be at most $n^\rho$. To summarize, for every $i\in [0,i_0 =\izo]$, we set $deg_i = \degi$, and for every $i\in [i_0+1,\ell]$ we set $deg_i = n^\rho$.

For the parameter $\ell$, we set $\ell = \emuell$. The intuition behind this assignment is as follows. Firstly, vertices of $A_\ell$ cannot connect with higher-level vertices. Consequently, all vertices of $A_\ell$ must take care of themselves, meaning that each vertex of $A_\ell$ could potentially connect with all other vertices of $A_\ell$. In addition, memory constraints limit us from allowing $A_\ell$ to exceed $n^\rho$. Setting $\ell = \emuell$ ensures that, w.h.p., $|A_{\ell}|\leq n^\rho$, adhering to our memory constraints
This completes the description of the hierarchy selection procedure. The pseudo-code of this procedure is given in Algorithm \ref{alg sample}.

\begin{algorithm}
   \caption{Sample Hierarchy}
    \begin{algorithmic}[1]
      \Function{Sample}{$V,\kappa,\rho$} 
      \State  $\ell =\emuell$
        \For {$i\in [0,\ell]$} $deg_i = {\min \{ \degi, n^\rho \} }$
        \EndFor
        \State $A_0 = V$
        \For{$i = 0$ to $\ell-1$}
            \State sample each vertex of $A_i$  to $A_{i+1}$ with probability $1/deg_i$.
        \EndFor
        \State $A_{\ell+1} = \emptyset$\\
        \Return $\{A_i \ | \ i\in [0,\ell+1] \}$
    \EndFunction
\end{algorithmic}
       \label{alg sample}
\end{algorithm}

\subsection{Selecting Edges} \label{sec edge selec gen}
In this section, we discuss the process of selecting edges for the set $Q$. Recall that we have access to a $(1+\epsilon_H,\beta)$-hopset $H$ for distances up to $\frac{1}{2}\alpha\eps{\ell}$. 
The algorithm begins by initializing $Q = \emptyset$. Then, for every level $i\in [0,\ell]$, we connect vertices of $\aismai$ with vertices in their vicinity. Specifically, for a vertex $u\in \aismai$, we define the \textit{bunch} of $u$ to be the set $B(u) = \{v\in A_i \ | \ \dghb(u,v)\leq \apdi\}$, where $\delta_i$ is a parameter specified later.
A vertex $u \in \aismai$ is labeled as \textit{$i$-dense} if $B(u)$ contains a vertex from $A_{i+1}$; otherwise, it is labeled as \textit{$i$-sparse}.

Consider an $i$-dense vertex $u$. The vertex $u$ selects a vertex $p(u) \in B(u)\cap A_{i+1}$ as its pivot. We remark that the pivot of $u$ does not have to be the closest $A_{i+1}$ vertex to $u$ in $B(u)$.
The edge $\{u, p(u)\}$ is then added to $Q$ with weight $\dghb(p(u),u)$. Now, the responsibility of \textit{taking care} of the vertex $u$ and all vertices $v \in A_{0} \setminus A_i$ for which $u$ is responsible lies with the pivot $p(u)$. 

Next, consider an $i$-sparse vertex $u$. One can add to the set $Q$ edges from $u$ to all vertices in its bunch without violating the restriction on the size of the set $Q$. However, this may violate memory constraints. Thus, we define the \textit{close-bunch} of $u$ to be $CB(u) = \{v\in A_i \ | \ \dghb(u,v)\leq \frac{1}{2}\apdi\}$. For every vertex $v\in CB(u)$, the vertex $u$ adds to the set $Q$ the edge $\{u,v\}$ with weight $\dghb(u,v)$.

Next, we discuss the parameter $\delta_i$, which defines the radii of the bunches of vertices in $A_i$. Observe that vertices of $A_0$ are in charge of taking care only of themselves, while vertices of $A_1$ may be in charge of taking care of some vertices in $A_0$, etc. Intuitively, we want to allow the radii of bunches of vertices from $A_i$ to be large enough to take care of all vertices of lower scales that depend on them. For this aim, we introduce the a variable $R_i$, representing an upper bound on the on the distance between a vertex $v \in \aismai$ and all vertices $u$ that depend on $v$. Define $R_0 = 0$ and  $R_{i} = (1+\epsilon_H)\delta_{i-1}+ R_{i-1}$ for $i\in [1,\ell]$.
For the sequence of distance thresholds, we set $\delta_i = \deltai$ for all $i\in [0,\ell]$.

Note that the maximal distance  to which we explore the graph is $\delta_\ell$. Intuitively, in Lemma \ref{lemma zeta} from Section \ref{sec gen stretch}, we show that explorations that are limited to $4\beta$-hops in $G\cup H$ provide $(1+\epsilon_H)$-approximations for distances up to $\delta_\ell$ in $G$. For this reason, all of our explorations are executed to $4\beta$-hops in $G\cup H$.

This completes the description of the edge selection procedure. The pseudo-code of this procedure is given in Algorithm \ref{alg connect}.

\begin{algorithm}
   \caption{Connect Vertices}
    \begin{algorithmic}[1]
      \Function{Connect}{$G, H, \epsilon_H,\beta,\kappa,\rho,\alpha, \{ A_i \ | \ i\in [0,\ell+1] \}$} 
      \Comment{$H$ is a $(1+\epsilon_H,\beta)$-hopset for distances up to $\frac{1}{2}\alpha\eps{\ell}$}
       \State  $\ell =\emuell$
      \State $R_0 = 0$, for $i\in [1,\ell]$ define $R_{i} = (1+\epsilon_H)\delta_{i-1}+ R_{i-1}$
      \State $\delta_i = \deltai$ for all $i\in [0,\ell]$ 
      \State $Q = \emptyset$
      \For{$i=0 $ to $\ell$}
        \For{ $u\in A_{i}\setminus A_{i+1}$ in parallel}
            \If{\label{step bf1}{$\exists \ v\in A_{i+1}$ with $\dghb(u,v)\leq \apdi$}}
               \State add  $\{u,v\}$ with weight $\dghb(u,v)$ to $Q$
            \Else
            \For{\label{step bf2}{ $v\in A_i$ with $\dghb(u,v)\leq \apdi/2$} in parallel}
              \State add  $\{u,v\}$ with weight $\dghb(u,v)$ to $Q$
            \EndFor
            \EndIf
        \EndFor
        \EndFor\\
        \Return $Q$
    \EndFunction
\end{algorithmic}
       \label{alg connect}
\end{algorithm}

\subsection{Implementation and Round Complexity}\label{sec gen cc}

In this section, we provide an MPC implementation for Algorithms \ref{alg sample} and  \ref{alg connect}, and analyze their round  complexity in MPC when we use machines with $O(n^\gamma)$ memory, and a total memory of $\tilde{O}((|E|+|H|)n^\rho)$. 

We assume that the machines are indexed $M_0,M_1,\dots$, and each machine $M_j$ knows its index $j$.
Recall that the input for the algorithm is the graph $G$, a $(1+\epsilon_H,\beta)$-hopset $H$, parameters $\epsilon <1$, $\rho \in [1/{\log {\log n}},1/2]$, $\kappa \geq 1/\rho$, and $\alpha >0$.
In addition, we assume that for each edge $(u,v)\in E \cup H$, both $(u,v)$ and $(v,u)$ are saved to the memory of some machine. For every vertex $u\in V$, let $M(u)$  be the set of machines that store edges with $u$ as their left endpoint. Denote by $r_u$ the minimal index of a machine in $M(u)$. 

\subsubsection{Preprocessing}

In a preliminary step, we guarantee that for each vertex $u \in V$, the machines in $M(u)$ form a contiguous set, and all machines within $M(u)$ are aware of $r_u$. To achieve this, the algorithm begins by sorting all edges based on their left endpoint. Consequently, for any vertex $u \in V$, the set $M(u)$ becomes contiguous. It has been demonstrated by Goodrich et al. \cite{GoodrichSZ11} that this edge sorting process on the machines can be accomplished in $O(1/\gamma)$ rounds of MPC using machines with $O(n^\gamma)$ memory, and a total memory of ${O}(|E|+|H|)$.

Once the set $M(u)$ is a contiguous set, we can compute $r_u$. For this aim, we utilize a subroutine provided by Dinitz and Nazari  \cite{DinitzN19}. This subroutine, computes for each edge $(u,v)$ a tuple $((u,v), deg(u),$ $ deg(v), r_u,$ $ r_v, i_u(v), i_v(u))$, where $deg(u), deg(v)$ represent the degrees of $u$ and $v$ in the graph $G \cup H$, and $i_u(v),i_v(u)$ are the respective indices (according to lexicographical order) of the edge  $(u,v)$ among all edges incident to $u$ and $v$. These tuples play a crucial role in facilitating seamless communication between vertices. Once these tuples are computed, for every $u\in V$, all machines in $M(u)$ know  $r_u$. 
The following lemma from \cite{DinitzN19} provides a concise summary of the round complexity of computing these tuples for every edge.

\begin{lemma}\label{lemma dn tuple}\cite{DinitzN19}
    Let $G' = (V',E')$ be an input graph. 
    Let $M_{(u,v)}$ be the machine that stores a given edge $(u,v)$. We can create tuples of the form $((u,v), deg(u),deg(v), r_u, r_v, i_u(v),i_v(u))$, stored at $M_{(u,v)}$ for all edges in $O(1/\gamma)$ rounds of MPC using machines with $O(n^\gamma)$ memory, and a total memory of ${O}(|E'|)$.  
\end{lemma}
We use the procedure devised by Dinitz and Nazari to create tuples for every (directed) edge $(u,v)\in E\cup H$. By Lemma \ref{lemma dn tuple}, this can  be done in $O(1/\gamma)$ rounds of MPC using machines with $O(n^\gamma)$ memory, and a total memory of ${O}(|E|+|H|)$.   

\subsubsection{Sampling Hierarchy}\label{sec sampling hierarchy gf}
Here we explain how Algorithm \ref{alg sample} ensures that all machines in $M(u)$  agree on the index $i$ such that $u\in A_i\setminus A_{i+1}$,  for every vertex $u\in V$. 

Recall that every copy of the edge $(u,v)$ is associated with a tuple that contains the index ${r}_u$. The machine  $M_{{r}_u}$ samples $u$ to the hierarchy. Then, the machine $M_{{r}_u}$ broadcasts the message $\langle u, i\rangle $ such that $u\in A_i\setminus A_{i+1}$ to all other machines in $M(u)$. 
Note that $M_{r_u}$ knows $deg(u)$ and the number of edges stored on each machine $n^\gamma$. Thus, it can infer which machines are in $M(u)$.  
Observe that some machines may be responsible for sampling more than one vertex. However, since the edges are sorted by the left endpoint,  there can be at most one vertex $u\in V$ such that $M$ is the first machine in $M(u)$ and not the last machine in $M(u)$. Therefore, each machine $M$ is responsible for broadcasting the result regarding at most one vertex. 
Dinitz and Nazari \cite{DinitzN19} have provided a routine $Broadcast(b,x,y)$, that given a message $b$ and two indices $x,y$, broadcasts the message $b$ to all machines $M_x,M_{x+1},\dots,M_{y}$. The properties of this routine is summarized in the following theorem.

\begin{theorem}\label{theo min broadcast}\cite{DinitzN19}
Given an input graph $G' = (V',E')$, the subroutine $Broadcast(b,x,y)$  can be implemented in $O(1/\gamma) $ rounds of MPC using machines with $O(n^\gamma)$ memory, and a total memory of size ${O}(|E'|)$. 
\end{theorem}

For every vertex $u\in V$, the machine $M_{r_u}$ broadcasts a single message to all machines in $M(u)$.

The following lemma summarizes the round complexity of Algorithm \ref{alg sample}. 
\begin{lemma}
    \label{lemma cc select}
    Given a graph $G=(V,E)$, a hopset $H$, a parameter $\ell$ and a sampling probability sequence $\{deg_i \ | \ i\in [0,\ell] \}$, Algorithm \ref{alg sample} can be executed in $O(1/\gamma)$  rounds of MPC using machines with $O(n^\gamma)$ memory, and a total memory of size ${O}(|E|+|H|)$. 
\end{lemma}


\subsubsection{Selecting Edges}\label{sec selecting edges impl}

We now discuss Algorithm \ref{alg connect}. The algorithm works in $\ell+1$ consecutive phases. For every $i\in [0,\ell]$, Algorithm \ref{alg connect} performs only two tasks that require communication. The first task is detecting the $i$-dense vertices and connecting them with their pivots (step \ref{step bf1}). The second task is connecting the $i$-sparse vertices with their close-bunch (step \ref{step bf2}). 
Both tasks are solved using variants of a Bellman-Ford exploration.

Essentially, to assign $i$-pivots to $i$-dense vertices, explorations from all sources in $A_i$ to depth $4\beta$ in $G\cup H$ are executed. Since each vertex needs to be detected by at most one exploration, this is equivalent to computing a single Bellman-Ford exploration from a dummy vertex $s$ connected to all vertices in $A_{i+1}$ with $0$-weight edges. 

The primary challenge lies in connecting $i$-sparse vertices with their close bunches, as it entails executing multiple Bellman-Ford explorations with limited depth and hops, in parallel. While Dinitz and Nazari \cite{DinitzN19} addressed this issue, their work assumes $\rho \leq \gamma$. In our implementation, we extend their approach to cases where $\rho > \gamma$.

\paragraph{Detecting $i$-dense vertices. }
To detect the $i$-dense vertices, we execute a Bellman-Ford exploration from all vertices in $A_{i+1}$. The exploration is limited to $4\beta$-hops and $\apdi$-distance, and executed in the graph $G\cup H$. We remark that while here we have a set of sources, each vertex is required to learn only the closest source to it. Therefore, this exploration is computationally equivalent to performing a single Bellman-Ford exploration, which is limited to $(4\beta)$-hop and $\apdi$-distance, in $G\cup H$ from a dummy vertex $s$ which is connected to all vertices in $A_{i+1}$ with $0$-weight edges.
Dinitz and Nazari \cite{DinitzN19} have provided an MPC implementation for the restricted Bellman-Ford exploration. Their result is summarized in Theorem \ref{theo restricted bellman-ford} below. 

\begin{theorem}
    \label{theo restricted bellman-ford}\cite{DinitzN19}
    Given a graph $G= (V,E)$, a set of edges $H$, a parameter $h$ and a source node $s\in V$ the restricted Bellman-Ford algorithm computes distances $d^{(h)}_{G\cup H}(s,v)$ for all $v\in V$ in $O(\frac{h}{\gamma})$ rounds of MPC using machines with $O(n^\gamma)$ memory, and a total memory of size $O(|E|+|H|)$. 
\end{theorem}

Note that we limit the depth of the exploration. This is computationally equivalent to discarding all distance estimates that exceed our limit. Thus, detecting $i$-dense vertices and connecting them with their pivots can be done in $O(\frac{\beta}{\gamma})$  rounds of MPC using machines with $O(n^\gamma)$ memory, and a total memory of size $O(|E|+|H|)$. 

For every vertex $v\in \aismai$, let $w\in A_{i+1}$ be the vertex that has detected $v$, if such a vertex exists. The machine $M_{r_u}$ writes to its memory that the edge $(v,w)$ was added to $Q$, with weight $\dghb(u,v)$. We note that each machine $M$ writes to its memory at most one edge per vertex that is stored on $M$. Therefore, this does not violate memory constraints.

\paragraph{Connecting $i$-sparse vertices.}

As discussed above, the second task is more complicated from a computational point of view. This is because each vertex needs to acquire information regarding \textit{all} $i$-sparse vertices that are within distance $\frac{1}{2}\apdi$ and $4\beta$ hops from it. 
For this aim, for every index $i\in [0,\ell]$, we execute a Bellman-Ford exploration in the graph $G\cup H$. These explorations are limited to $(4\beta)$-hops and $(\frac{1}{2}\apdi)$-distance. We use the term $S_i$-explorations to describe all explorations that are executed from vertices in $A_i\setminus A_{i+1}$ in step \ref{step bf2} of Algorithm \ref{alg connect}. 

Observe that some vertices may be traversed by several explorations.
However, the expected number of explorations that traverse each vertex $v\in V$ is small. Intuitively, this happens because areas with many $A_i$-vertices are likely to have an $A_{i+1}$-vertex. Thus, all $A_i$-vertices in these areas are $i$-dense. In the following lemma, we provide an upper bound on the number of $S_i$-explorations that traverse each vertex $u\in V$. 
\begin{lemma}
    \label{lemma traverse}
    Let $i\in [0,\ell]$. For every vertex $u\in V$, for any constant $c>1$, with probability at least $1-n^{-c}$, the number of $S_i$-explorations that traverse $u$ is at most $deg_i\cdot c\cdot {\ln n}$.
\end{lemma}

\begin{proof}
    Let $i\in [0,\ell]$, and let $u\in V$. Let $X$ be the set of vertices $v\in A_i$ such that $\dghb(u,v)\leq \apdi$, and denote $|X| = x$. 
    
    \textbf{Case 1:} at least one vertex  $v\in X$ is sampled to $A_{i+1}$. In this case, an exploration was executed from the vertex $v$ to depth $4\beta$ hops and $\apdi$ distance in the graph $G\cup H$. Therefore, 
    all vertices in $X$ were discovered by an exploration that has originated in a vertex from $A_{i+1}$ in step \ref{step bf1} of Algorithm \ref{alg connect}. It follows that all vertices in $X$ are $i$-dense. Thus, in step \ref{step bf2} of Algorithm \ref{alg connect}, none of the vertices in $X$ initiate an exploration. As a result, the vertex $u$ is not traversed by any one of the $S_i$-explorations.
    
    \textbf{Case 2:} no vertex in $X$ is sampled to $A_{i+1}$. In this case, exactly $x$ explorations traverse the vertex $u$. The probability that no vertex of $X$ is sampled to $A_{i+1}$ is $(1-1/deg_i)^x$. Thus, the expected number of $S_i $-explorations that traverse $u$ is $x\cdot (1-1/deg_i)^\ell \leq deg_i$. 
    
    Observe that for $x\geq deg_i\cdot c\cdot {\ln n}$, the probability that no vertex of $X$ is sampled to $A_{i+1}$ is at most 
    \begin{equation*}
        (1-1/deg_i)^{deg_i\cdot c\cdot {\ln n}} \leq e^{-c{\ln n}} \leq n^{-c}.
    \end{equation*}    
    It follows that with probability at least $1-n^{-c}$, the number of $S_i$-explorations that traverse $u$ is at most $deg_i\cdot c\cdot {\ln n}$.
\end{proof}

By Lemma \ref{lemma traverse}, we conclude that w.h.p.,  each vertex is traversed by at most $\addedmemoryparam$ $S_i$-explorations. Therefore, we allocate $\mu = \tilde{O}(n^\rho)$ memory for each edge $(u,v)\in E\cup H$.

For the multiple-sources Bellman-Ford algorithm, we create $\mu = \tilde{O}(n^\rho)$ copies for every edge $(u,v)\in E\cup H$ and sort the edges by the left endpoint. A mechanism is devised to deliver $\mu$ messages along each edge in $O(1)$ rounds and handle distance estimates efficiently. This mechanism involves computing a tuple for every copy of an edge, containing information about the machines storing edges connected to the left and right endpoints. This allows seamless collaboration between machines to process messages and update distance estimates.

In the following lemma, we discuss the round complexity of computing $\addedmemory$ copies for each edge $(u,v)\in E\cup H$.

\begin{lemma}
    \label{lemma copy edges}
    Given a set of edges $E'$, there is a subroutine that for every edge $(u,v)\in E'$, creates and stores $\addedmemory$ copies of $(u,v)$, such that the set $M(u)$ is a contiguous set of machines. Furthermore, for every edge $(u,v)\in E'$ all of its edge-copies are stored on a contiguous set of machines. 
    The subroutine can be executed in $O(1/\gamma)$ time using machines with $O(n^\gamma)$ memory, and a total memory of size $\tilde{O}(|E'|\cdot n^\rho)$. 
\end{lemma}
\begin{proof}
    First, we redistribute the edges on the set of machines such that there is a gap of $\mu-1$ between every pair of consecutive edges. Then, we fill the gap following each edge $(u,v)$ with $\addedmemory-1$ copies of it.

    To distribute the edges, we first compute the new location of each edge $(u,v)$. Let $(u,v)$ be the $j$th edge stored on the $i$th machine, i.e., it has global index $i\cdot n^\gamma +j$. This edge is relocated to global index $l = ((i-1)n^\gamma +(j-1))\cdot\mu$. This index is the $l\% n^\gamma $ index in the machine $M_{\lfloor \frac{l}{n^\gamma} \rfloor}$. Note that this requires each machine to send at most $n^\gamma$ messages, as each machine initially stores at most $n^\gamma$ edges. In addition, observe that we allocate at most $n^\gamma$ edges to every machine. Thus, this step can be done in $O(1)$ rounds.

     Next, we fill the gaps. Specifically, for every edge $(u,v)$ stored now in global index $l$, we make copies $\{ ((u,v),j) \ | \ j\in [0,\addedmemory-1]\}$ and store them on global indices $[l,l+\addedmemory-1]$. Denote by $M_k$ the machine that stores the edge $(u,v)$. Observe that global index $l+\addedmemory-1$ is located in machine number $k' = \lceil \frac{l+\mu-1}{n^\gamma}\rceil$. The machine $M_k$ broadcasts the message $\langle (u,v), l \rangle $ to all machines with indices in $[k,k']$. Observe that each machine has at most one edge $e$ such that it needs to broadcast messages to consecutive machines regarding $e$. In addition, each machine receives a broadcast regarding at most one edge $e$. Thus, by Theorem \ref{theo min broadcast}, this broadcast can be done in $O(1/\gamma)$ time using machines with $O(n^\gamma)$ memory, and a total of $O(|E'|\cdot n^\rho)$ memory.

     Let $M$ be a machine that receives a broadcast $\langle (u,v), l \rangle $. Since $M$ knows its index, using the global index $l$ it can infer which copies of $(u,v)$ are to be stored on it, and write these copies to its memory. Consequently, machines $[M_k,M_{k'}]$ now store the copies $((u,v),x)$ for $x\in [0,\mu-1]$.

    Observe that now, for every vertex $u\in V$, the set of machines $M(u)$  is a  contiguous set of machines, and also, all copies of an edge $(u,v)$ are stored on a contiguous set of machines.

    
\end{proof}

We execute the edge-copying subroutine on the edges of $E\cup H$, which requires $\tilde{O}((|E|+|H|) n^\rho) $ total memory.

For every edge $(u,v)$ and for every index $j\in [0,\mu-1]$, denote by $M_{((u,v),j)}$ the machine that stores the $j$th copy of the edge $(u,v)$. Each machine  $M_{((u,v),j)}$ now locally updates the tuple associated with the $j$th copy of $(u,v)$. Specifically, we compute the tuple $((u,v),j,\widehat{deg}(u),\widehat{deg}(v), \widehat{r}_u,\widehat{r}_v,\widehat{i}_u(v),\widehat{i}_v(u))$, where 
\begin{equation*}
    \begin{array}{lclclclcl}
         \widehat{deg}(u) &=& deg(u)\cdot \addedmemory ,& 
         \widehat{deg}(v) &=& deg(v)\cdot \addedmemory ,\\
         \widehat{r}_u & = & r_u +\frac{\mu}{n^\gamma},&
         \widehat{r}_v & = & r_v +\frac{\mu}{n^\gamma},\\
         \widehat{i}_u(v) & = & i_u(v)\cdot \addedmemory + j ,&
         \widehat{i}_v(u) & = & i_v(u)\cdot \addedmemory + j.
    \end{array}
\end{equation*}

We are now ready to to discuss the multiple sources restricted Bellman-Ford algorithm.  The input for the algorithm is a graph $G= (V,E)$, a hopset $H$, a hop parameter $h = 4\beta$ and a distance threshold parameter $\frac{1}{2}\apdi$. The set of sources $S_i$ is the set of $i$-sparse vertices. For the correctness of our analysis, it is crucial that Lemma \ref{lemma traverse} holds. Essentially, at the beginning of each round, each vertex $u$ stores distance estimates to at most $\mu$ distinct sources. These estimates are sent from the machines in $M(u)$ to all machines that store neighbors of $u$. Then, the machines in $M(u)$ process the messages they had received to update their distance estimates.
In the following lemma, we provide implementation details and analyze the round complexity of a single step of the $S_i$-explorations. 

\begin{lemma}
    \label{lemma single round msBF}
    For $h'\in [0,h]$, w.h.p., step $h'$ of all $S_i$-explorations can be executed in $O(1/\gamma)$ rounds of MPC using machines with $O(n^\gamma)$ memory and $\tilde{O}((|E|+|H|)n^\rho)$ total memory.  
\end{lemma}
\begin{proof}
    Step $h'$ consists of three parts.
    \begin{enumerate}
        \item \textbf{Broadcast.} In this part, each vertex broadcasts its distance estimates to its neighbors in $G\cup H$. 
        \item \textbf{Compute Distance Estimates.} In this part, each vertex infers the current distance estimates according to messages it has received in the previous part. 
        \item \textbf{Rearrange Data.} In the final part, the data is rearranged so that we will be able to send it in the next round.
    \end{enumerate}
    
    Let $h'\in [0,h]$, and let $u$ be a vertex that has been detected by $x$ sources $s_0,s_1,\dots,s_{x-i}$ from $A_i$ until round $h'$. We assume that when round  $h'$ begins, for every edge $(u,v)\in E\cup H$ and every index $j\in [0,x-1]$, the machine $M_{((u,v),j)}$ contains the distance estimate $\langle (s_j,u), \mathpzc{d}_{s_j}(u)\rangle$. We show that this assumption is valid when we discuss part \textit{rearrange data}.
    
    \paragraph{Part $1$: Broadcast.}
    Round $h'$ begins by sending distance estimates. Specifically, for every edge $(u,v)\in E\cup H$ and every index $j\in [0,x-1]$, the machine $M_{((u,v),j)}$ sends the distance estimate $\langle (s_j,u), \mathpzc{d}_{s_j}(u)\rangle$ to machine $M_{((v,u),j)}$, which has index $\widehat{r}_v+ \lfloor \frac{\widehat{i}_v}{n^\gamma}\rfloor$. Now, the machine $M_{((v,u),j)}$ saves $\langle (s_j,v), \mathpzc{d}_{s_j}(v)  = \mathpzc{d}_{s_j}(u) +\omega(u,v) \rangle$ to its memory. If $\mathpzc{d}_{s_j}(v) > \apdi$, indicating that the computed distance exceeds the specified threshold, the distance estimate is discarded immediately. We note that the $j$th copy of an edge $(u,v)$ communicates exclusively with the $j$th copy of the edge $(v,u)$, and vice versa. Therefore, this step adheres to the I/O limitation for each machine.
    
    \paragraph{Part 2: Compute Distance Estimates.}
    Observe that at this point, it is possible that the machines in $M(v)$  hold multiple distance estimations corresponding to the same source. Therefore, we sort the distance estimates stored on $M(v)$ according to the source. By \cite{GoodrichSZ11}, this can be done in $O(1/\gamma)$ time. Now, each machine  $M_y \in M(v)$ that holds a distance estimate $\langle (s,v), \mathpzc{d}_{s}(v)\rangle$ checks whether this is the minimal distance estimate between $s$ and $v$. Note that the distance estimates may be distributed on several machines and that each machine may hold estimates to several sources. Let $\langle (s,v), \mathpzc{d}_{s}(v)\rangle$ be the smallest distance estimate held by a machine $M_y$ for a source $s$. As the estimates are sorted by distance, the machine $M_{y}$ needs to check if the machine $M_{y-1}$, if exists, holds a smaller distance estimate to the source $s$. Any distance estimate found to be suboptimal is then discarded. 
    Now, for every sources $s\in A_i$, the machines in $M(v)$ hold at most one distance estimate $\langle (s,v), \mathpzc{d}_{s}(v)\rangle$.

    \paragraph{Part 3: Rearrange Data.}
    It is left to rearrange the information on the machines in $M(v)$ such that if the vertex $v$ has been detected by sources $s_0,s_1,\dots,s_{x'-1}$, 
    then for every edge $(u,v)\in E\cup H$ and every index $j\in [0,x'-1]$, the machine $M_{((v,u),j)}$ contains the distance estimate $\langle (s_j,v), \mathpzc{d}_{s_j}(v)\rangle$. For this aim, recall that we allow each machine to hold at most $n^\gamma$ distance estimates and also at most $n^\gamma$ edge-tuples. Also, recall that the edge-tuples are sorted on $M(v)$ according to the the right endpoint, and ties are broken by the index of the copy $j$. Let $(v,u^*)$ be the minimal edge in this sort, i.e., $u^*$ is the minimal ID neighbor of $v$.

    Next, we sort all distance estimates $\langle (s,v), \mathpzc{d}_{s}(v)\rangle$ stored on $M(v)$. By Lemma \ref{lemma traverse}, there are at most $\addedmemory$ such distance estimates stored on $M(v)$. Now we know that the number of tuples $((v,u^*),j,\dots)$ held by each machine $M\in M(v)$ is at least the number of distance estimates held by $M$. Each tuple $((v,u^*),j,\dots)$ is now assigned with at most one distance estimate $\langle (s_j,v), \mathpzc{d}_{s_j}(v)\rangle$, where $u^*$ is the minimal ID neighbor of $v$. It is the responsibility of the $j$th copy of the edge $(v,u^*)$ to distribute the distance estimate $\langle (s_j,v), \mathpzc{d}_{s_j}(v)\rangle$ to the $j$th copies of all edges. For this aim, we add the estimate $\langle (s_j,v), \mathpzc{d}_{s_j}(v)\rangle$ to the tuple $((v,u^*),j,\dots)$, and sort all tuples, by the left endpoint. Ties are broken by the index $j$, and then by the right endpoint. 
    
    Observe that now, for every vertex $v\in V$ and for every $j\in [0,\addedmemory-1]$, the $j$th tuples of all edges incident to $v$ are stored on a contiguous set of machines. Furthermore, for every $j$, the first tuple with index $j$ is the tuple of the edge $(v,u^{*})$.    
    Next, the machine $M_{((v,u^{*}),j)}$ broadcasts the distance estimate held by the tuple  $((v,u^{*}),j,\dots)$ to all machines that store $j$-th copies of edges $(v,\cdot)$. Observe that for every edge $e$, all of its copies are stored \textit{in bulks}, on a contiguous set of machines. Thus, every machine participated in a broadcast concerning edge-copies of at most two edges. 
    Hence, by Theorem \ref{theo min broadcast} this requires $O(1/\gamma)$ time. Each machine $M_{((v,u),j)}$ that receives a distance estimate $\langle (s_j,v), \mathpzc{d}_{s_j}(v)\rangle$, adds this estimate to the tuple   $((v,u),j,\dots)$. Finally, the tuples are sorted again, by the left endpoint. Ties are broken by the right endpoint, and then by the index of the copy. 
\end{proof}

The following theorem summarizes the properties of the multiple-sources restricted Bellman-Ford exploration.

\begin{theorem}
    \label{theo multi source BF}
    Let $G= (V,E)$ be an unweighted, undirected graph, $H$ be a hopset, $h$, $\delta_i$ be parameters, and  $S_i\subseteq V$ be a set of sources,
    such that each vertex $v\in V$ has at most $\mu$ sources $s\in S_i$ with $d^{(h)}_{G\cup H}(s,v)\leq \apdi$. 
    
    The multiple-sources restricted Bellman-Ford algorithm computes distances $d^{(h)}_{G\cup H}(s,v)$ for every pair $(s,v)\in S\times V$ with $d^{(h)}_{G\cup H}(s,v)\leq \apdi$ in $O(h/\gamma)$ rounds of MPC using machines with $O(n^\gamma)$ memory and $\tilde{O}((|E|+|H|)n^\rho)$ total memory.
\end{theorem}
Observe that when the multiple sources Bellman-Ford exploration terminates, for every vertex $v\in V$, for every source $s\in S_i$ such that $\dghb(s,v)\leq \apdi$, there is exactly one index $j$ such that the $j$th copy of the edge  $(v,u^*)$ stored $\dghb(s,v)$. The machine $M_{((v,u^*),j)}$ now adds the edge $(s,v)$ to the set $Q$ with weight $\dghb(s,v)$. 

We aim at storing all edges of $Q$ on a set of output machines. 
Once the algorithm finishes connecting $i$-sparse vertices, the edges stored on all machines are sorted, such that the edges added to the set $Q$ are stored on a separate set of machines. By \cite{GoodrichSZ11}    
this can be accomplished in $O(1/\gamma)$ rounds of MPC using machines with $O(n^\gamma)$ memory, and a total memory of $\tilde{O}((|E|+|H|)n^\rho+|Q|)$.

We are now ready to analyze the round complexity of Algorithm \ref{alg connect}.

\begin{lemma}\label{lemma alg connect rt}
   W.h.p., Algorithms \ref{alg sample} and \ref{alg connect} can be executed in $O(\beta/\gamma)$ 
   rounds of MPC using machines with $O(n^\gamma)$ memory and $\tilde{O}((|E|+|H|)n^\rho+ |Q|)$ total memory.
\end{lemma}

\begin{proof}
First, note that by Lemma \ref{lemma cc select}, Algorithm \ref{alg sample} can be executed in $O(1/\gamma)$  rounds of MPC using machines with $O(n^\gamma)$ memory, and a total memory of size ${O}(|E|+|H|)$.

We now discuss Algorithm \ref{alg connect}. 
    For every $i\in [0,\ell]$, detecting the $i$-dense vertices and connecting them with their pivots requires executing a Bellman-Ford exploration from all vertices in $A_i$, limited to $4\beta$ hops and $\apdi$ depth. By Theorem \ref{theo restricted bellman-ford}, this can be done in $O(\beta/\gamma)$ rounds of MPC using machines with $O(n^\gamma)$ memory and ${O}(|E|+|H|)$ total memory.

    Connecting the $i$-sparse vertices with their close bunches requires creating $\addedmemory$ copies of each edge $(u,v)\in E\cup H$, and executing Bellman-Ford explorations from all $i$-sparse vertices, in parallel. These explorations are limited to $4\beta$-hops and $\frac{1}{2}\apdi$ depth. Observe that by Lemma \ref{lemma traverse}, w.h.p., each vertex $u\in V$ has at most $\addedmemory$ $i$-sparse vertices $v$ such that $\dghb(u,v)\leq \frac{1}{2}\apdi$. Recall that $\addedmemory = \tilde{O}(n^\rho)$. Thus, by Lemma \ref{lemma copy edges} and Theorem \ref{theo multi source BF}, copying edges and performing explorations can be done in $O(\beta/\gamma)$ rounds of MPC using machines with $O(n^\gamma)$ memory and $\tilde{O}((|E|+|H|)n^\rho)$ total memory.

    Sorting the edges such that the edges of the set $Q$ are stored on a set of output machines can be accomplished in $O(1/\gamma)$ rounds of MPC using machines with $O(n^\gamma)$ memory, and a total memory of $\tilde{O}((|E|+|H|)n^\rho+|Q|)$ \cite{GoodrichSZ11}.

    To summarize, w.h.p., Algorithm \ref{alg connect} can be executed in $O(\beta/h)$ rounds of MPC using machines with $O(n^\gamma)$ memory and $\tilde{O}((|E|+|H|)n^\rho+|Q|)$ total memory. 
\end{proof}



\subsection{Analysis of the Size}\label{sec gen size}

In this section, we analyze the size of the edge set $Q$. Our strategy is to first provide upper bounds on the expected size of each set $A_i$. Then, we provide an upper bound on the expected number of edges added to the set $Q$ for every vertex $u\in A_i\setminus A_{i+1}$. This strategy is standard for algorithms that use the approach of Elkin and Peleg \cite{ElkinP01} or of Thorup and Zwick \cite{ThorupZ06} for constructing hopsets, emulators, and spanners. We remark that algorithms employing the Elkin and Peleg approach have a different edge addition mechanisms, while algorithms inspired by Thorup and Zwick usually use a different sampling probability sequence, leading to variations in the details of their size analyses compared to our approach.

In the following two lemmas, we provide an upper bound on the expected size of the set $A_i$ for all $i\in [0,\ell]$. Recall that for every $i\in [0,i_0 =\izo]$, we set $deg_i = \degi$, and for every $i\in [i_0+1,\ell]$ we set $deg_i = n^\rho$.
\begin{lemma}
    \label{lemma size whp 1}
    Let $i\in [0,i_0+1]$. The expected size of $A_i$ is $n^{1-\frac{2^{i}-1}{\kappa}}$. 

    Moreover, for every $i\in [0,i_0+1]$, we have that w.h.p. the size of $A_i$ is at most $n^{1-\frac{2^{i}-1}{\kappa}}\cdot{\log n}$.
\end{lemma}
\begin{proof}
    The proof is by induction on the index $i$.
    For $i= 0$, we have $A_0 = V $, thus $|A_0| = n$ and the claim holds. 

    Assume the claim holds for $i\leq i_0$, and prove it holds for $i+1$. By the induction hypothesis, $E[|A_i|] = n^{1-\frac{2^{i}-1}{\kappa}}$. A vertex $v\in A_i$ is sampled to $A_{i+1}$ with probability $ deg_i = \degi$. Thus, the expected number of vertices that belong to $A_{i+1}$ is 
    \begin{equation*}
    \begin{array}{lclclclclclc}
        E[|A_{i+1}|] & = & 
        E[|A_i|]/deg_i & = &
        n^{1-\frac{2^{i}-1}{\kappa}}/ \degi 
        & = & n^{1-\frac{2^{i+1}-1}{\kappa}}. 
    \end{array}
    \end{equation*}

    For the second assertion of the lemma, recall that $\rho \leq 1/2$. Observe that the expected size of $A_i$ is at least $1$, since  
    \begin{equation*}
        \begin{array}{lclclclclclclclclc}
             \frac{2^{i}-1}{\kappa}& \leq  &
             \frac{2^{i_0+1}-1}{\kappa}& \leq  &\\
             \frac{2^{\izo+1}-1}{\kappa}& \leq  &
             \frac{2\kappa\rho-1}{\kappa}& \leq  &
             2\rho
             & \leq &1  .
        \end{array}
    \end{equation*}
    
    Let $c>2$ be a constant. From a Chernoff bound, for $\delta = c{\ln n}$ we have 
    \begin{equation*}
    \begin{array}{lclclclclclclclc}
        Pr[ |A_i| > (1+c{\ln n}) E[|A_i|]] 
        &\leq & e^{-\frac{(c{\ln n})^2 E[|A_i|]}{2+c{\ln n}}}\\
        
        &\leq& e^{-\frac{c{\ln n}}{2}}  & = &n^{-\frac{c}{2}}.          
    \end{array}
    \end{equation*}

\end{proof}

\begin{lemma}
    \label{lemma size whp 2}
    Let $i\in [i_0+1,\ell]$. The expected size of $|A_i|$ is $n^{1 + \frac{1}{\kappa} - ( i-i_0 )\rho}$. 

    Moreover, for every $i\in [i_0+1,\ell]$, we have that w.h.p. the size of $A_i$ is at most $n^{1 + \frac{1}{\kappa} - ( i-i_0 )\rho}\cdot {\log n}$.
\end{lemma}
\begin{proof}
    
    The proof is by induction on the index $i$.
    For $i= i_0+1$, by Lemma \ref{lemma size whp 1} we have that the expected size of $A_i$ is 
    \begin{equation*}
      \begin{array}{lclclclclclclclcl}
           n^{1-\frac{2^{i_0+1}-1}{\kappa}} 
           &\leq& n^{1-\frac{2^{\izo+1}-1}{\kappa}} 
           &\leq& n^{1-\frac{\kappa\rho-1}{\kappa}}\\
           &\leq& n^{1 + \frac{1}{\kappa} - \rho}
           &\leq& n^{1 + \frac{1}{\kappa} - ( i-i_0 )\rho}.
      \end{array}  
    \end{equation*}
    
    Moreover,  the size of $A_i$ is at most $n^{1 + \frac{1}{\kappa} - ( i-i_0 )\rho}\cdot {\log n}$, w.h.p. 

    Assume the claim holds for $i\in  [i_0+1,\ell-1]$, and prove it holds for $i+1$. By the induction hypothesis, 
    $E[|A_i|] = n^{1 + \frac{1}{\kappa} - ( i-i_0 )\rho}$. A vertex $v\in A_i$ is sampled to $A_{i+1}$ with probability $ deg_i = n^\rho$. Thus, the expected number of vertices that belong to $A_{i+1}$ is 
    \begin{equation*}
    \begin{array}{lclclclclclc}
        E[|A_{i+1}|] & = & 
        E[|A_i|]/deg_i & = &
        n^{1 + \frac{1}{\kappa} - ( i-i_0 )\rho}/ n^\rho 
        \\ &&&= & n^{1 + \frac{1}{\kappa} - ( i+1-i_0 )\rho}. 
    \end{array}
    \end{equation*}

    We now prove the second assertion of the lemma. Observe that the expected size of $|A_i|$ is at least $1$, since  
    \begin{equation*}
    \begin{array}{lclclclclclc}
        1 + \frac{1}{\kappa} - ( i+1-i_0 )\rho
        &\geq& 1 + \frac{1}{\kappa} - ( \ell-i_0 )\rho
        \\&\geq& 1 + \frac{1}{\kappa} - ( \lceil \frac{\kappa+1}{\kappa\rho}\rceil  -1)\rho
        \\&\geq& 1 + \frac{1}{\kappa} - ( \frac{\kappa+1}{\kappa\rho}  )\rho
        \\&\geq& 
        1 + \frac{1}{\kappa} - \frac{\kappa+1}{\kappa} 
        
        &\geq& 0
    \end{array}
    \end{equation*}

    Let $c>2$ be a constant.  From a Chernoff bound, for $\delta = c{\ln n}$ we have 
    \begin{equation*}
    \begin{array}{lclclclclclclclc}
        Pr[ |A_i| > (1+c{\ln n}) E[|A_i|]] 
        &\leq & e^{-\frac{(c{\ln n})^2 E[|A_i|]}{2+c{\ln n}}} 
        \\
         &\leq& e^{-\frac{c{\ln n}}{2}}  & = & n^{-\frac{c}{2}}.          
    \end{array}
    \end{equation*}

\end{proof}

Observe that Lemma \ref{lemma size whp 2} implies that for  $i= \ell$ we have
\begin{equation}
    \label{eq Aell size}
    \begin{array}{lclclclclc}
         E[|A_{\ell}|] 
        & = & n^{1 + \frac{1}{\kappa} - ( \ell-i_0 )\rho} 
        \\& \leq  & n^{1 + \frac{1}{\kappa} - (\lceil \frac{\kappa+1}{\kappa\rho}\rceil  -1 )\rho} 
         \\& \leq  & n^{1 + \frac{1}{\kappa} -  \frac{\kappa+1}{\kappa}  +\rho} 
        &= & n^\rho,
    \end{array}
\end{equation}

and also for any constant $c>2$  we have
\begin{equation}
    \label{eq Aell size whp}
    \begin{array}{lclclclclc}
         Pr[ |A_\ell| > (1+c{\ln n}) n^\rho] 
        &\leq & n^{-\frac{c}{2}}.     
    \end{array}
\end{equation}

Next, we provide an upper bound on the expected size of $B(u)$ for  $i$-sparse vertices $u\in \aismai$.
Recall that for every vertex $u\in \aismai$, the bunch of $u$ is defined to be $B(u) = \{v\in A_i\  |\ \dghb(u,v)< \apdi \} $.
\begin{lemma}
    \label{lemma isprse bunch}
    Let $i\in [0,\ell]$, and let $u\in \aismai$ be an $i$-sparse vertex. Then, the expected size of $B(u)$ is at most $deg_i$, Moreover, w.h.p. the size of $B(u)$ is $O(deg_i\cdot {\log n})$.
\end{lemma}
\begin{proof}
      First, consider the case where $i= \ell$.  Let $u\in A_\ell$. Note that $B(u) \subseteq A_\ell$. By eq. \ref{eq Aell size} we have that  $E[|A_\ell|]\leq n^\rho= deg_\ell$. In addition, by eq. \ref{eq Aell size whp}, we have  $Pr[|A_\ell| > (1+c{\ln n}) deg_\ell] \leq  n^{-c/2}$ for any constant $c>2$.  Thus the claim holds for $i=\ell$.   

    Consider now the case where $i\in [0,\ell-1]$. Let $u\in A_{i}\setminus A_{i+1}$. 
    Since $u$ is $i$-sparse, we have $B(u)\cap A_{i+1} = \emptyset$.  
    Let $v$ be the vertex in $A_{i+1}$ with minimal $\dghb(u,v)$. The expected number of vertices $u'\in A_i$ with $\dghb(u,u')\leq \dghb(u,v)$ is $deg_i$, since each vertex of $A_i$ is sampled to $A_{i+1}$ with probability $1/deg_i$. Observe that for every vertex $u'\in B(u)$, we have $\dghb(u,u')\leq \dghb(u,v)$. Hence, if $B(u)\cap A_{i+1} = \emptyset$, then the expected size of $B(u)$ is at most $deg_i$. 

    For the second assertion of the lemma, from a Chernoff bound, we get that 
    $Pr[|B(u)| > (1+\delta)deg_i ] \leq e^{-\frac{\delta^2 deg_i}{2+\delta}} $ for any $\delta >0$. Recall that $deg_i \geq \nfrac \geq  1$. Choosing $\delta = c{\ln n }$ shows that 
    \begin{equation*}
    \begin{array}{lclclclc}
         Pr[|B(u)| > \jdef] & \leq & 
         e^{-\frac{\delta^2 deg_i}{2+\delta}} 
         \\&\leq & e^{-\frac{\jdef}{2}} 
         \\& \leq & n^{-\frac{c\cdot deg_i}{2}} \leq n^{-c/2}.
    \end{array}
    \end{equation*}    
\end{proof}

 The following lemma provides an upper bound on the expected number of edges added to the set $Q$ by vertices of $\aismai$.

\begin{lemma}\label{lemma vx contr size}
    Let $i\in [0,\ell]$. For every $u\in A_i\setminus A_{i+1}$, the expected number of edges added to the set $Q$ for the vertex $u$ is at most $deg_i$. 
    Moreover, w.h.p. the number of edges added to the set $Q$ for the vertex $u$ is  $O(deg_i\cdot {\log n})$.  
\end{lemma}
\begin{proof}
    Let $u\in \aismai$. The analysis splits into two cases. 

    \textbf{Case 1:} The vertex $u$ is $i$-dense. In this case, the algorithm adds to the set $Q$ only a single edge from the vertex $u$ to its pivot and the claim holds. 

    \textbf{Case 2:} The vertex $u$ is $i$-sparse. In this case, the algorithm adds to the set $Q$ edges from $u$ to all vertices in its close-bunch, i.e., all vertices in $CB(u) = \{v\in A_i\  |\ \dghb(u,v)< \apdi/2 \} $. Observe that $CB(u)\subseteq B(u)$. By Lemma \ref{lemma isprse bunch} we have that the expected size of $B(u)$ is at most $deg_i$. In addition, w.h.p., the size of $B(u)$ is $O(deg_i\cdot {\log n})$. 
\end{proof}

We are now ready to provide an upper bound on the size of the set $Q$. 
\begin{lemma}\label{lemma fin size}
   The expected size of the set $Q$ is  $O( \nfrac\cdot {\log (\kappa\rho)})$. 

   Moreover, the size of $Q$ is $O( \nfrac\cdot {\log ^2 n}\cdot {\log (\kappa\rho)})$ w.h.p.
\end{lemma}
\begin{proof}
    Let $i\in [0,i_0]$. 
    By Lemmas \ref{lemma size whp 1} and \ref{lemma vx contr size}, the expected number of edges added to $Q$ by vertices of $A_i\setminus A_{i+1}$ is at most     
    \begin{equation*}
        \begin{array}{lclclclc}     
           E[|A_i|]\cdot deg_i & = & n^{1-\frac{2^{i}-1}{\kappa}}\cdot \degi & = & \nfrac. 
            \end{array}
    \end{equation*}

    Let $i\in [i_0+1,\ell]$.  By Lemmas \ref{lemma size whp 2} and \ref{lemma vx contr size}, the expected number of edges added to $Q$ by vertices of $A_i\setminus A_{i+1}$ is at most     
    \begin{equation*}
        \begin{array}{lclclclc}     
           E[|A_i|]\cdot deg_i & = & 
           n^{1 + \frac{1}{\kappa} - ( i-i_0 )\rho} \cdot n^\rho 
           & = & n^{1 + \frac{1}{\kappa} - ( i-i_0-1 )\rho}. 
            \end{array}
    \end{equation*}

     Recall that $i_0= \izo$. To summarize, the expected size of $Q$ is at most 
     \begin{equation*}
        \begin{array}{lclclclc}     
           E[|A|] & =  & (i_0+1)\nfrac+ \sum\limits_{j= i_0+1}^{\ell}
            n^{1 + \frac{1}{\kappa} - ( j-i_0-1 )\rho} \\& = & (\izo+2)\nfrac, 
            \end{array}
    \end{equation*}
Where the last inequality holds since

     \begin{equation*}
        \begin{array}{lclclclc}     
            n^{1 + \frac{1}{\kappa}}\cdot \sum\limits_{j= i_0+1}^{\ell}
            n^{- ( j-i_0-1 )\rho} & \leq  & n^{1 + \frac{1}{\kappa}} .
            \end{array}
    \end{equation*}

     For the second assertion of the lemma, by Lemmas  \ref{lemma size whp 1} and \ref{lemma vx contr size},  the number of edges added to $Q$ by vertices of $A_i\setminus A_{i+1}$ is w.h.p. at most 
     \begin{equation*}
        \begin{array}{lclclclc}     
           E[|A_i|] \cdot  deg_i \cdot O({\log^2 n})  = &\\ 
            \begin{cases}
                O(\nfrac\cdot {\log ^2 n}),  & i\in [0,i_0] \\
                 O(n^{1 + \frac{1}{\kappa} - ( i-i_0-1 )\rho}\cdot {\log ^2 n}) & i\in [i_0+1,\ell]
\end{cases}
            \end{array}
    \end{equation*}

     It follows that the size of $Q$ is w.h.p. at most 
     \begin{equation*}
        \begin{array}{lclclclc}     
            &&O(i_0\cdot\nfrac\cdot {\log ^2 n}+ \sum\limits_{j= i_0+1}^{\ell}
            n^{1 + \frac{1}{\kappa} - ( j-i_0-1 )\rho}\cdot {\log^2 n}) \\
            & = & O( \nfrac\cdot {\log ^2 n}\cdot {\log (\kappa\rho)}). 
            \end{array}
    \end{equation*}

\end{proof}



\subsection{Analysis of Stretch}
\label{sec gen stretch} 
Here, we analyze the stretch of set $Q$. Unlike previous constructions, which focus solely on near-exact hopsets or near-additive emulators, our analysis demonstrates the dual functionality of $Q$ as both a near-exact hopset and a near-exact emulator.
Generally speaking, we show here that for pairs of vertices $u,v\in V$ with $d_G(u,v)\leq \eps{\ell}$, the set $Q$ contains a path with a small number of edges, that provides a good approximation for the distance in $G$ between $u,v$. This approach is also often used when using the superclustering-and-interconnection approach \cite{ElkinP01} or the Thourup-Zwich approach \cite{ThorupZ06} for constructing near-additive spanners or emulators and near-exact hopsets.  However, there are some technical differences between our analysis and the analysis of previous constructions.  

Most importantly,  previous constructions were used exclusively to build near-additive spanners, near-additive emulators, or near-exact hopsets. Thus, in the analysis of previous work, some of the properties of the output set are omitted. To the best of our knowledge, this is the first work to provide the general properties of the output, which essentially satisfy the stretch requirement in the context of both near-additive emulators and near-exact hopsets. This makes our argument a bit more elaborate than previous arguments, but allows us to provide a unified view on the construction. 

Next, we provide some definitions which are necessary for the analysis of the stretch. Let $v\in V$. Recall that a vertex $v\in A_i\setminus A_{i+1}$ is said to be $i$-dense if it is connected with its pivot. We say that $v$ is $0$-clustered and denote $p_0(v)=v$. For $i>0$, we say that $v$ is $i$-clustered if $v$ is $(i-1)$-clustered, and $p_{i-1}(v)$ is $(i-1)$-dense. In this case, we define $p_i(v)$ to be the pivot of the vertex $p_{i-1}(v)$, i.e., $p_i(v) = p(p_{i-1}(v))$.
Observe that each $i$-clustered vertex is also $i'$-clustered for all $0\leq i'\leq i$.

Recall that $R_0 = 0$ and  $R_{i} = (1+\epsilon_H)\delta_{i-1}+ R_{i-1}$  for $i\in [0,\ell]$. We now show that for every $i$-clustered vertex $v\in V$, the set $Q$ contains a $v-p_i(v)$ path of weight up to $R_i$ and at most $i$ edges. 

\begin{lemma}\label{lemma i clustered ri hop}
    Let $i\in [0,\ell]$, and let $v$ be an $i$-clustered vertex. Then, $$d^{(i)}_{Q}(v,p_i(v))\leq R_i.$$
\end{lemma}
\begin{proof}
    The proof is by induction on the index $i$. For $i=0$, by definition, we have $p_0(v)=0$, and so the claim holds. Let $i>0$. Assume that the claim holds for $i-1$ and prove it also holds for $i$.

  Let $v$ be an $i$-clustered vertex. Let $u= p_{i-1}(v)$. By the induction hypothesis, we have that
   \begin{equation}\label{eq ri-1 hop}
    d^{(i-1)}_{Q}(v,u)\leq R_{i-1}.   
   \end{equation}
    
   Note that the vertex $u$ has added the edge $\{u,p_i(v)\}$ to the set  $Q$ with weight $\dghb(u,p_i(v))$. Also, note that 
   $\dghb(u,p_i(v))\leq (1+\epsilon_H)\cdot \delta_{i-1}$. 
   Together with eq. \ref{eq ri-1 hop}, we conclude that 
    \begin{equation*}\label{eq v ci(v) hop}
    d^{(i)}_{Q}(v,p_i(v)) \leq d^{(i-1)}_Q(v,u)  + d^{(1)}_Q(u,p_i(v) \leq 
    R_{i-1}+ (1+\epsilon_H)\cdot \delta_{i-1} = R_{i}.   
   \end{equation*}
\end{proof}

We now provide an explicit upper bound on $R_i$. Recall that $\delta_i = \deltai$ for all $i\in [0,\ell]$.  Assume that $\epsilon_H\leq 1/10$.

\begin{lemma}
    \label{lemma bound ri}
    For every $i\in [0,\ell]$, we have $R_i <  (1+\epsilon_H)\alpha\cdot \sum_{j=0}^{i-1}\eps{j}\cdot 4^{i-1-j}$. 
\end{lemma}

\begin{proof}
    The proof is by induction on the index $i$. For $i=0$, recall that $R_0=0$, and so the claim holds. 
    Let $i\geq 0$. Assume that the claim holds for $i$ and prove it for $i+1$.
    Recall that $\epsilon_H \leq 1/10$. By definition, we have

    \begin{equation*}
        \begin{array}{lclclclclclcl}
    R_{i+1} & =& (1+\epsilon_H)\delta_{i}+ R_{i} \\& =& 
         (1+\epsilon_H)\left(\alpha\eps{i}+2R_i\right)+ R_{i} 
    \\& \leq & (1+\epsilon_H)\alpha\eps{i}+ 4R_{i} .
        \end{array}
    \end{equation*}
    Together with the induction hypothesis, we have
    \begin{equation*}
        \begin{array}{lclclclclclcl}
             R_{i+1} & \leq  &   (1+\epsilon_H)\alpha\eps{i}+ 4R_{i}\\
             & \leq &(1+\epsilon_H)\alpha\eps{i}+ 4 \cdot (1+\epsilon_H)\alpha\cdot \sum_{j=0}^{i-1}\eps{j}\cdot 4^{i-1-j} \\
             & < &  (1+\epsilon_H)\alpha\cdot \sum_{j=0}^{i}\eps{j}\cdot 4^{i-j}.
        \end{array}
    \end{equation*}
\end{proof}

\begin{lemma}
    \label{lemma exp ri}
    Let $\epsilon< 1/10$. Then, $R_i \leq \frac{22}{10}\alpha\eps{i-1}$.
\end{lemma}

\begin{proof}
Recall that $\epsilon_H\leq 1/10$. By Lemma \ref{lemma bound ri}, we have 

\begin{equation*}
        \begin{array}{lclclclclclcl}
             R_{i}  & \leq  & 4^{i-1}\cdot (1+\epsilon_H)\alpha\cdot \sum_{j=0}^{{i-1}} \eps{j} \left(\frac{1}{4}\right)^{j} \\
              & <  & 4^{i-1}\cdot\frac{11}{10}\alpha\cdot\left[\frac{\left( \frac{1}{4\epsilon}\right)^{i} -1}{\frac{1}{4\epsilon} -1}  \right ]\\
               & \leq  &  4^{i-1}\cdot\frac{11}{10}\alpha\cdot
               \left( \frac{1}{4\epsilon}\right)^{i}    \cdot \frac{4\epsilon}{1-4\epsilon}  \\
               & \leq  & \frac{11}{10}\alpha\cdot
               \left( \frac{1}{\epsilon}\right)^{i-1}   \cdot 2  & = & \frac{22}{10}\alpha\eps{i-1} .
        \end{array}
    \end{equation*}

\end{proof}

In the following lemma, we show that an exploration from a vertex $u\in V$ to depth $\apdi$ and $4\beta$ hops in the graph $G\cup H$ detects all vertices $v$ with $d_G(u,v)\leq \delta_i$.

\begin{lemma}\label{lemma zeta}
	Let $i\in [0,\ell]$, and let  $u,v\in V$ be a pair of vertices such that $d_G(u,v)\leq \delta_i$. 
	Then, 
	$$\dghb(u,v) \leq (1+\epsilon_H)d_G(u,v).$$
\end{lemma}

\begin{proof} 
    Recall that $\epsilon\leq 1/10$. By Lemma \ref{lemma exp ri} we have
    \begin{equation*}
        \begin{array}{lclclclcl}
            \delta_i &=& \deltai &< & \alpha\cdot \epsi+ 9\alpha\eps{i-1} \\& \leq & 2\alpha\epsi.
        \end{array}
    \end{equation*}

	Let $\pi(u,v)$ be a shortest $u-v$ path in $G$.
    We divide the path $\pi(u,v)$ into four parts $\pi(u=u_0,u_1),\pi(u_1,u_2),\pi(u_2,u_3),\pi(u_3,u_4 =v)$, each of size at most $\frac{1}{2}\alpha\epsi $. Since $H$ is a $(1+\epsilon_H,\beta)$-hopset for distances up to $\frac{1}{2}\alpha\eps{\ell}$, for every index $j\in [0,3] $  we have 
         $d^{(\beta)}_{G\cup H}(u_j,u_{j+1}) \leq  (1+\epsilon_H)d_G(u_j,u_{j+1})$.
 Thus,  
	$$d^{(4\beta)}_{G\cup H}(u,v) \leq (1+\epsilon_H )d_G(u,v).$$
\end{proof}

Next, we show that any $i$-sparse vertex $u$ is connected with all vertices $v$ such that $d_G(u,v)\leq \frac{1}{2}\delta_i$.

\begin{lemma}\label{lemma i sparse neighb}
    Let $u\in A_i\setminus A_{i+1}$ be an $i$-sparse vertex, and let $v\in A_i$ such that $d_G(u,v)\leq \frac{1}{2}\delta_i$. Then, $d^{(1)}_{Q}(u,v) \leq  (1+\epsilon_H)d_G(u,v)$.
\end{lemma}
\begin{proof}
    The distance between $u$ and $v$ is at most $\delta_\ell/2$. Recall that the hopset $H$ is a $(1+\epsilon_H,\beta)$-hopset for distances up to $\alpha\eps{\ell}/2$. Hence, by Lemma \ref{lemma zeta}, we have $\dghb(u,v)\leq (1+\epsilon_H)d_G(u,v) \leq (1+\epsilon_H)\delta_i/2$. It follows that $v\in CB(u)$.  
    
    Moreover, the vertex $u$ is $i$-sparse, i.e., 
    it has  no vertex $v\in A_{i+1}$ with $\dghb(u,v)\leq \apdi$.
    Recall that $i$-sparse vertices are connected with their close-bunch. Thus, the edge $\{u,v\}$ was added to the set $Q$ with weight $\dghb(u,v')$, which is at most $(1+\epsilon_H)d_G(u,v)$.
\end{proof}

Finally, we provide an upper bound on the stretch of the set $Q$ for pairs of vertices $u,v\in V$ with $d_G(u,v)\leq \frac{1}{2}\alpha\epsi $.

\begin{lemma}
   \label{lemma stretch hop}
   Let $u,v$ be a pair of vertices with $d_G(u,v)\leq \frac{1}{2}\alpha\epsi$ such that all vertices of $\pi(u,v)$ are at most $ i$-clustered, except for maybe one vertex that is $(i+1)$-clustered. Then, 
   $$d^{(3(1/\epsilon+3)^i)}_{Q} (u,v) \leq   (1+\epsilon_H+28\epsilon\cdot i)d_G(u,v) + 14\alpha\eps{i-1} .$$
\end{lemma}

\begin{proof}
    Let $\pi(u,v)$ be a shortest $u-v$ path between them in the graph $G$. 

    The proof is by induction on the index $i$. For $i=0$, we have that all vertices on $\pi(u,v)$ are at most $0$-clustered, except for maybe one vertex that is $1$-clustered. Thus, the vertex $u$ or the vertex $v$ is at most $0$-clustered. Assume w.l.o.g. that $u$ is at most $0$ clustered. 
    By lemma \ref{lemma i sparse neighb} we have that $u$ has added the edge $\{ u, v\}$ to the set $Q$ with weight at most $\dghb(u,v)\leq (1+\epsilon_H)d_G(u,v)$. In other words, 
    $d^{(1)}_Q(u,v)\leq (1+\epsilon_H)d_G(u,v)$ 
    and the claim holds.

    Let $i>0$. Assume that the claim holds for $i-1$ and prove it holds for $i$. 
    Observe that all vertices of $\pi(u,v)$ are at most $i$-clustered, except for maybe one vertex that is $(i+1)$-clustered.
    Let $u',v'$ be the first and last $i$-clustered vertices in $\pi(u,v)$, respectively. We divide the path $\pi(u,v)$ into three segments; from $u$ to $u'$, from $u'$ to $v'$ and from $v'$ to $v$. Intuitively, in order to provide an upper bound on the error of the first and third segments, we use the induction hypothesis. For the second segment, we use the vertices $p_i(u')$ and $p_i(v')$. See Figure \ref{fig uv path induc} for an illustration.

    \begin{figure}
        \centering
        \includegraphics[scale = 0.1]{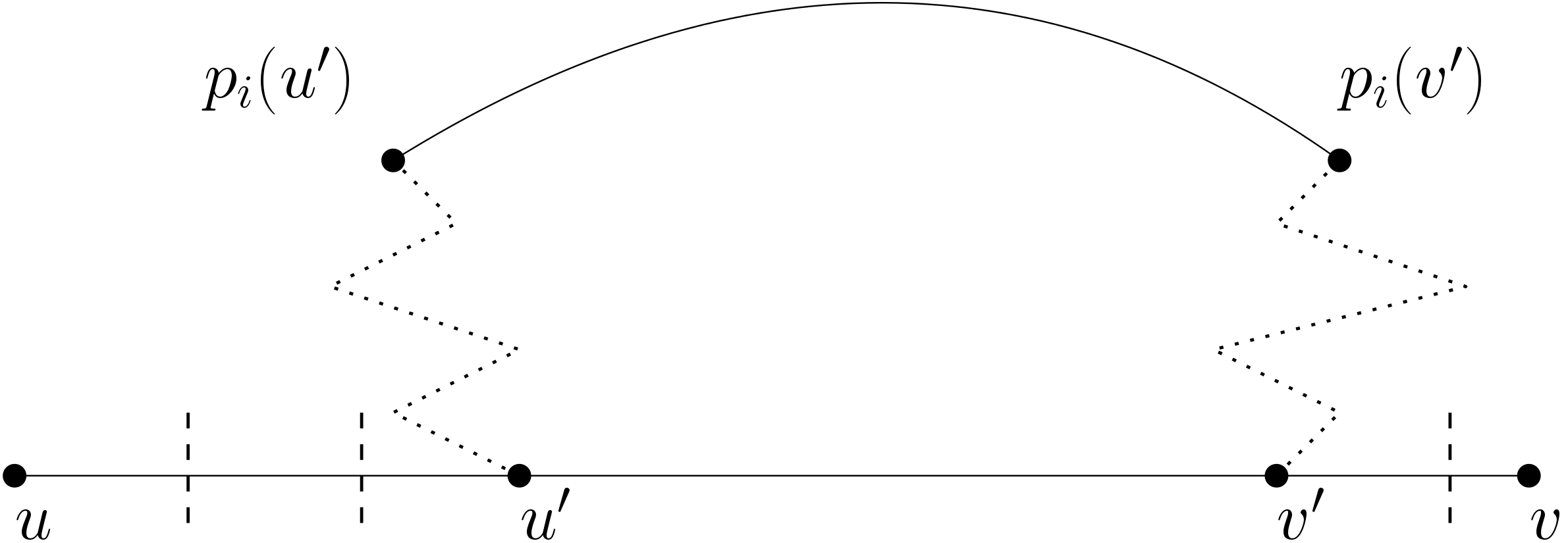}
        \caption{The path from $u$ to $v$ using the edges of the set $Q$. The solid straight line depicts the path $\pi(u,v)$ in the graph $G$. The dashed lines represent the partition of the subpaths $\pi(u,u')$ and $\pi(v',v)$ into segments of length at most $\frac{1}{2} \alpha\eps{i-1}$. The dotted lines represent the paths of at most $i$ edges in $Q$ from $u$ (and $v$) to $p_i(u')$ (and $p_i(v')$). The curved edge represents the edge $\{p_i(u'),p_i(v')\}$ in $Q$. }
        \label{fig uv path induc}
    \end{figure}
    
    Consider first the subpaths $\pi(u,u'),\pi(v',v)$ of $\pi(u,v)$ from $u$ to $u'$ and from $v'$ to $v$. In both these paths, all vertices are at most $(i-1)$-clustered, except for  one vertex that is $i$-clustered. We divide the path $\pi(u,u')$  into segments, each of length exactly $\frac{1}{2}\alpha\eps{i-1}$, except for the last segment, which may be shorter. The induction hypothesis is applicable to each one of these segments. For the sake of brevity, denote  $h' =\left(\left\lceil \frac{d_G(u,u')}{\frac{1}{2}\alpha\eps{i-1}}\right\rceil\cdot 3(1/\epsilon+3)^{i-1}\right)$ and $h'' = \left(\left\lceil \frac{d_G(v',v)}{\frac{1}{2}\alpha\eps{i-1}}\right\rceil\cdot 3(1/\epsilon+3)^{i-1}\right)$. Therefore, we have

    \begin{equation}
        \label{eq uu' induc}
        \begin{array}{lclclclc}
        d^{(h')}_{Q} (u,u')\\ 
        \leq (1+\epsilon_H+28\epsilon({i-1}))d_G(u,u') + \left\lceil \frac{d_G(u,u')}{\frac{1}{2}\alpha\eps{i-1}}\right\rceil\cdot 14\alpha\eps{i-2}\\
        
        \leq (1+\epsilon_H+28\epsilon({i-1})+\frac{14\alpha\eps{i-2}}{\frac{1}{2}\alpha\eps{i-1}}
        )d_G(u,u') +
         14\alpha\eps{i-2}
\\
        \leq (1+\epsilon_H+28\epsilon\cdot i
        )d_G(u,u') +
         14\alpha\eps{i-2}
        .
        
        \end{array}
    \end{equation}
    Similarly, we also have

    \begin{equation}
        \label{eq v'v induc}
        d^{(h'')}_{Q} (v',v) \leq (1+\epsilon_H+28\epsilon\cdot i
        )d_G(v',v) +
         14\alpha\eps{i-2}
    \end{equation}

    We now discuss the subpath of $\pi(u,v)$ from $u'$ to $v'$.
    Recall that $u',v'$ are both $i$-clustered. Also, recall that by  Lemma \ref{lemma i clustered ri hop} we have  
    \begin{equation}\label{eq u' ci(u') hop}
        \begin{array}{lclclclc}
        d^{(i)}_Q(u',p_i(u'))& \leq & R_i, & and &
        d^{(i)}_{Q}(v',p_i(v'))& \leq & R_i.
        \end{array}
    \end{equation}

    Next, we show that the set $Q$ contains the edge $\{p_i(u'),p_i(v')\}$. Recall that all vertices on $\pi(u,v)$ are at most $i$-clustered, except for one vertex, which may be $(i+1)$-clustered. Thus, at least one of the vertices $u',v'$ is at most $i$-clustered. W.l.o.g., assume that $u'$ is at most $i$-clustered. It follows that the vertex $p_i(u')$ is $i$-sparse. Hence, it has added to the set $Q$ edges to all of its close-bunch $CB(p_i(u'))$.
    It is left to show that $p_i(v')\in CB(p_i(u'))$. Recall that $d_G(u',v')\leq d_G(u,v) \leq \frac{\alpha}{2}\epsi$. Thus, by the triangle inequality and by Lemma \ref{lemma i clustered ri hop}, we have 
     \begin{equation*}
     \begin{array}{lclclcl}
          
         d_G(p_i(u'),p_i(v')) &\leq& d_G(p_i(u'),u')+d_G(u',v')+d_G(v',p_i(v')) \\&\leq& \frac{\alpha}{2}\epsi+2R_i \\&=& \frac{1}{2}\delta_i.
     \end{array}
     \end{equation*}

    Hence, $p_i(v')\in CB(p_i(u'))$.
    By Lemma \ref{lemma i sparse neighb}, we have 
    \begin{equation}
        \label{eq pi to pi}
        \begin{array}{lclclcl}
             
        d^{(1)}_{Q} (p_i(u'),p_i(v')) &\leq& (1+\epsilon_H)d_G(p_i(u'),p_i(v')) \\& \leq& 
        (1+\epsilon_H)(d_G(u',v')+2R_i).
        \end{array} 
    \end{equation}

    Recall that $\epsilon_H<1/10$. 
    By equations \ref{eq u' ci(u') hop} and \ref{eq pi to pi}
    \begin{equation}
        \label{eq u'v'}
        \begin{array}{lclclclc}
        d^{(2i+1)}_{Q}(u',v') &\leq& 
        d^{(i)}_{Q}(u',p_i(u'))+
        d^{(1)}_{Q}(p_i(u'),p_i(v'))+\\
        &&\quad
        d^{(i)}_{Q}(p_i(v'),v')\\
        &\leq &
        (1+\epsilon_H)(d_G(u',v')+2R_i) +2R_i \\
        &\leq &(1+\epsilon_H)d_G(u',v') + 5R_i.  
        \end{array}
    \end{equation}
Recall that by Lemma \ref{lemma exp ri} we have $R_i\leq 2.2\alpha\eps{i-1}$. Also, recall that $\epsilon \leq 1/10$. 
   By equations \ref{eq uu' induc}, \ref{eq v'v induc} and \ref{eq u'v'} we have

    \begin{equation}\label{eq u v hop}
        \begin{array}{lclclclc}
           & d_{Q}^{(h'+h''+2i+1)}(u,v)
            
            \\
            \leq & 
            d^{(h')}_{Q} (u,u') +

            d^{(2i+1)}_{Q}(u',v')+
            
            d^{(h'')}_{Q} (v',v)
            \\
            
            \leq &
            (1+\epsilon_H+28\epsilon\cdot i
        )d_G(u,u') +
         14\alpha\eps{i-2}+

            (1+\epsilon_H)d_G(u',v') \\
            & \qquad+ 5R_i+
            
            (1+\epsilon_H+28\epsilon\cdot i
        )d_G(v',v) +
         14\alpha\eps{i-2}
\\
              \leq &
            (1+\epsilon_H+28\epsilon\cdot i
        )d_G(u,v) +
         28\alpha\eps{i-2}+

             11\alpha\eps{i-1}
             
             \\
              \leq &
            (1+\epsilon_H+28\epsilon\cdot i
        )d_G(u,v) +

             14\alpha\eps{i-1}.
        \end{array}
    \end{equation}

It is left to show that the number of hops $ h'+ h''+ 2i + 1 $ is at most $ 3(1/\epsilon+3)^{i}$. 
Recall that $d_G(u,v)\leq\frac{\alpha}{2} \epsi$. Also recall that  $h' =\left\lceil \frac{d_G(v',v)}{\frac{\alpha}{2}\eps{i-1}}\right\rceil\cdot 3(1/\epsilon+3)^{i-1}$ and $h'' = \left\lceil \frac{d_G(v',v)}{\frac{\alpha}{2}\eps{i-1}}\right\rceil\cdot 3(1/\epsilon+3)^{i-1}$. Thus, 
\begin{equation}
    \label{eq bound h}
    \begin{array}{lclclclclclclcl}
        &h'+ h''+ 2i + 1  \\= &\left\lceil \frac{d_G(u,u')}{\frac{\alpha}{2}\eps{i-1}}\right\rceil\cdot 3(\frac{1}{\epsilon}+3)^{i-1} + \left\lceil \frac{d_G(v',v)}{\frac{\alpha}{2}\eps{i-1}}\right\rceil\cdot 3(\frac{1}{\epsilon}+3)^{i-1} 
    + 2i + 1\\
    \leq &\left( \frac{d_G(u,v)}{\frac{\alpha}{2}\eps{i-1}}+2 \right)\cdot 3(\frac{1}{\epsilon}+3)^{i-1} + 2i + 1   
    \\
    \leq &\left( \frac{1}{\epsilon}+2 \right)\cdot 3(\frac{1}{\epsilon}+3)^{i-1} + 2i + 1   
    \\
    =& 3(\frac{1}{\epsilon}+3)^{i} - 3(\frac{1}{\epsilon}+3)^{i-1} + 2i + 1  \\
    \leq& 3(\frac{1}{\epsilon}+3)^{i}.

    \end{array}
\end{equation}
Where the last inequality holds since $1/\epsilon> 0$, and also $i\geq 1$.

To summarize, by equations  \ref{eq u v hop} and \ref{eq bound h} we have

\begin{equation}
             d_{Q}^{(3(1/\epsilon+3)^i)}(u,v)
            \leq  (1+\epsilon_H+28\epsilon\cdot i)d_G(u,v) + 14\alpha\eps{i-1} . 
\end{equation}

\end{proof}

Observe that Lemma \ref{lemma stretch hop} provides an upper bound on the error that the set $Q$ guarantees for pairs of vertices with distance up to $\frac{\alpha}{2}\epsi$ from one another. 
In the following lemma, we use Lemma \ref{lemma stretch hop} to analyze the error that the set  $Q$ for 
for pairs of vertices with a distance of up to $\alpha\epsi$ from one another.

\begin{lemma}
\label{lemma stretch gen final}
   Let $u,v\in V$ be a pair of vertices such that $d_G(u,v)\leq \alpha\eps{\ell}$. Then, 
   $$d_G(u,v) \leq d_Q(u,v)\leq d^{(6(1/\epsilon+3)^\ell)}_{Q} (u,v) \leq   (1+\epsilon_H+28\epsilon\ell)d_G(u,v) + 28\alpha\eps{\ell-1} .$$
\end{lemma}

\begin{proof} 
    First, note that whenever an edge $\{x,y\}$ was added to $Q$ with some weight $d$, there is an $x-y$ path in $G\cup H$ with weight $d$. Thus, for every pair of vertices $u,v\in V$ we have 
    \begin{equation}
        \label{eq not shorter}
        d_G(u,v)\leq d_Q(u,v).
    \end{equation}
    In addition, we trivially have $d_Q(u,v)\leq d^{(x)}_Q(u,v)$ for every $x$. Thus the two left-hand-side inequalities hold. 
    
    We now prove the right-hand-side inequality. Let $\pi(u,v)$ be a shortest path in $G$ between $u$ and $v$. We divide the path $\pi(u,v)$ into two sub-paths, $\pi(u,w),\pi(w,v)$, each of length at most $\frac{\alpha}{2}\eps{\ell}$. Observe that all vertices in $V$ are at most $\ell$-clustered. Thus, by Lemma \ref{lemma stretch hop} we have: 
    \begin{equation}
    \begin{array}{lclclc}
         
        d^{(6(1/\epsilon+3)^\ell)}_Q(u,v) & \leq & d^{(3(1/\epsilon+3)^\ell)}_Q(u,w)+d^{(3(1/\epsilon+3)^\ell)}_Q(w,v)
         \\
         &\leq & (1+\epsilon_H+28\epsilon\ell)d_G(u,v) + 28\alpha\eps{\ell-1}.
    \end{array}
    \end{equation}
\end{proof}

\paragraph{Summary}

 By Lemmas  \ref{lemma cc select}, \ref{lemma alg connect rt}, \ref{lemma fin size} and \ref{lemma stretch gen final}, we derive the following corollary. 

 \begin{restatable}{corollary}{CorQ}
     \label{coro Q}
    Let $G=(V,E)$ be an unweighted, undirected graph on $n$ vertices, and let $\epsilon <1/10$, $\alpha>0$, $\rho \in[ 1/{\log {\log n}}, 1/2]$ and $\kappa \geq 1/\rho$ be parameters. Let $\ell = \emuell$.  Given a $(1+\epsilon_H,\beta)$-hopset $H$ for distances up to 
    $\frac{\alpha}{2}\cdot \eps{\ell}$ in the graph $G$, where $\epsilon_H\leq 1/10$, 
    w.h.p., our algorithm computes a set of edges $Q$ of size $O(\nfrac\cdot {\log (\kappa\rho)}\cdot {\log ^2 n})$ such that for every pair of vertices $u,v\in V$  
    with $d_G(u,v)\leq \alpha\eps{\ell}$ we have 
    \begin{equation*}\begin{array}{llllclclclclc}
        d_G(u,v)&\leq d_Q(u,v) 
        \qquad \leq\\
        d^{(6(\frac{1}{\epsilon}+3)^{\ell})}_{Q} (u,v)
        
        &\leq (1+\epsilon_H +28\epsilon \ell ) d_G(u,v) +28\alpha\cdot \eps{\ell-1}  .
    \end{array}
    \end{equation*}
     The algorithm can be executed in $O(\beta/\gamma)$ 
    rounds of MPC using machines with $O(n^\gamma)$ memory and $\tilde{O}((|E|+|H|)n^\rho+|Q|)$ total memory. 
 \end{restatable}

        

\section{Near-Exact Hopsets}\label{sec hop}
    
In this section, we show how the general framework from Section \ref{sec general} can be used for the construction of limited-scale near-exact hopsets. Let $G=(V,E)$ be an unweighted, undirected graph on $n$-vertices and let $\rho \in [1/{\log {\log n}},1/2]$, $\epsilon \leq 1/10$ and $t> 1$ be input parameters. Intuitively, the parameter $\rho$ governs the memory usage of the algorithm as well as the size of the hopset, and the parameter $\epsilon$ governs the multiplicative error of the hopset.
Our algorithm constructs a $(1+\epsilon_\mathcal{H},\beta_\mathcal{H})$-hopset $\mathcal{H}$  for the scale $[0,t]$. 
The parameter $\epsilon_\mathcal{H}$ is a function of $\epsilon,\rho,t$ given by $\epsilon_\mathcal{H} = O(\epsilon{\log t}/\rho)$.
The hopbound parameter $\beta_\mathcal{H}$ is a function of $n,t,\epsilon_\mathcal{H},\rho$, and is given by  
\begin{equation}
    \label{eq def beta}
    \beta_\mathcal{H} = 6\left(\frac{1}{\epsilon}+3\right)^{1/\rho} =  O\left( \frac{\log t}{\epsilon_\mathcal{H}\rho}\right)^{1/\rho}. 
\end{equation}

Recall that the general framework takes a graph, along with parameters $\epsilon, \alpha, \rho, \kappa$, and a near-exact hopset for a suitable distance scale as input. In Section \ref{sec stretch hop}, we demonstrate that with a suitable choice of $\alpha$, the general framework can be employed to generate a near-exact hopset for the scale $[2^k,2^{k+1}]$. Consequently, by using the general framework iteratively for $O({\log t})$ consecutive times, one can construct a hopset for distances up to $t$. This results in a hopset with an expected size of $O(n^{1+ 1/\rho} \cdot {\log t})$.

However, it is unnecessary to re-sample a new hierarchy for each iteration of the general framework. Instead, one can sample the hierarchy just once using Algorithm \ref{alg sample}, and then execute Algorithm \ref{alg connect} for $O({\log t})$ consecutive times, utilizing the same hierarchy. This approach is favorable as it allows us to eliminate the ${\log t}$ term from the hopset size.

We now discuss the distance scales for which we build our hopset. Note that for distances up to $\beta_\mathcal{H}$, the empty set is a $(1,\beta_\mathcal{H})$-hopset, since for each pair of vertices $u,v\in V$  with distance $d_G(u,v) \leq \beta_{\mathcal{H}}$, the graph $G$ contains $u-v$ a path of at most $\beta_{\mathcal{H}}$ edges. Denote $k_0 = \lfloor {\log \beta_\mathcal{H}}\rfloor$ and define $H_{k_0-1} = \emptyset$.  Also, recall that we are interested in computing a near-exact hopset only for pairs of vertices with distance at most $t$. Denote $\lambda ={\lceil {\log t}\rceil -1}$. For every $k\in \krange$, we build a near-exact hop set for the scale $(2^k,2^{k+1}]$.

Next, we specify the input parameters for Algorithms \ref{alg sample} and for every one of the $O({\log t})$ executions of \ref{alg connect}. 
Algorithm \ref{alg sample} accepts parameters $\epsilon,\rho,\kappa$. The parameters $\epsilon,\rho$ are given to us as input. For the parameter $\kappa$, we set $\kappa = 1/\rho$. The intuition behind this assignment is that $\kappa$ governs the number of edges in the output of Algorithm \ref{alg connect}, and $\rho$ governs the memory usage of the algorithm.  In our case, the restrictions on the size and memory usage coalesce, as we can allow the hopset to be as large as the memory we are allowed to use. The output of the algorithm is a hierarchy $V = A_0 \supseteq A_1 \supseteq \dots\supseteq A_\ell \supseteq A_{\ell+1} = \emptyset$, where $\ell = \emuell = \lceil 1/\rho\rceil$.

For every $k\in \krange$, execution $k$ of Algorithm \ref{alg connect} accepts the same parameters $\epsilon,\rho,\kappa$ as Algorithm \ref{alg sample}, as well as the sampled hierarchy $\mathcal{A}$, a parameter $\alpha > 0$ and a $(1+\epsilon_H,\beta)$-hopset $H$ for distances up to $\frac{1}{2}\alpha\eps{\ell}$. For the parameter $\alpha$, we set $\alpha = 2^{k+1}\cdot \epsilon^\ell$. 
For the hopset $H$, we provide the union of outputs of iterations $k_0,\dots,k-1$ of Algorithm \ref{alg connect}. Formally, for every $k\in \krange$, let $H_k$ be the output of execution $k$ of Algorithm \ref{alg connect}. The hopset $H$ provided to execution $k$ of Algorithm \ref{alg connect} is the set $H = \bigcup_{j \in [k_0,k]} H_j$. 
In Section \ref{sec stretch hop} we show that $H_k$ is a near-exact hopset for the scale $(\alpha \eps{\ell}/2 = 2^k , \alpha\eps{\ell} = 2^{k+1}]$. Thus, $H$ is a near-exact hopset for the scale $[0,2^{k+1}]$, and the hopset $H$ can be used as the input hopset for iteration $k+1$ of Algorithm \ref{alg connect}.

Finally, let $\mathcal{H} = \bigcup_{k=k_0}^{\lambda}H_k$. This completes the description of the construction of $\mathcal{H}$.


\subsection{Analysis of the Size}\label{sec size hop}
To analyze the size of the hopset, we essentially show that for every vertex $v\in V$, its close-bunch in the $j$th execution of Algorithm \ref{alg connect} is a subset of its close-bunch in the $j+1$ execution of the algorithm, for all $j\in [0,{\log t}-1]$. As a result, we essentially show that there is a small set of vertices in the vicinity of $v$ that an execution of Algorithm \ref{alg connect} might connect with $u$. 
Consequently, we show that the upper bound given in Section \ref{sec gen size}, for the size of a single-scale hopset $H_k$ is also applicable to the entire set $\mch$, even though $\mch$ is composed of $O({\log t})$ single-scale hopsets. In particular, we prove the following lemma. 

\begin{restatable}{lemma}{hopSize}
        \label{lemma size hop 1}
    Let $u\in V$. The expected number of edges added to the set $\mathcal{H}$ for the vertex $u$ by all iterations of Algorithm \ref{alg connect} is, w.h.p., $O(n^\rho{\log ^2n})$.
\end{restatable}

\begin{proof}
First, consider the case where $i= \ell$.  Let $u\in A_\ell$. Note that throughout all executions of Algorithm \ref{alg connect}, the vertex $u$ adds edges only to vertices of $A_\ell$. Recall that by eq. \ref{eq Aell size whp}, w.h.p., $|A_\ell| = O(n^\rho\cdot {\log n})$.  Thus the claim holds for $i = \ell$.

Consider now the case where $i\in [0,\ell-1]$. Let $u\in A_{i}\setminus A_{i+1}$. 
Let $B_k(u),CB_k(u)$ be the bunch and the close-bunch of $u$ in execution $k$ of Algorithm \ref{alg connect}, respectively. Formally,
\begin{equation*}
    \begin{array}{lclclclclc}
         B_k(u) &=&
         \{v\in A_{i} \ | \  \dghb(u,v)\leq \apdi  \}& and\\
         CB_k(u) &=& \{v\in A_{i} \ | \  \dghb(u,v)\leq \frac{1}{2}\apdi  \}, 
    \end{array}
\end{equation*}
where $H = \bigcup_{k'\in [k_0,k-1]} H_{k'}$. 

Let $k$ be the index of the first iteration in which $u$ was considered $i$-dense, i.e., $B_k(u)\cap A_{i+1} \neq \emptyset$, and $B_{k'}(u)\cap A_{i+1} = \emptyset$ for all $k'\leq k$. Note that $v\in B_{k'}(u)$ for all $k'\in [k,\lambda]$. Thus, for every $k'\in [k,\lambda]$, the vertex $u$ adds to the set $H_{k'}$ only one edge. 

We now discuss the contribution of $u$ to the sets $H_{k'}$ for $k'\in [k_0,k-1]$. During iterations $k_0,k_{0+1},\dots,k-1$, the vertex $u$ was considered $i$-sparse, and was connected with its close-bunch. Observe that $CB_{k_0}(u)\subseteq CB_{k_0+1}(u)\subseteq \dots \subseteq CB_{\lambda}(u) $.
Thus, throughout all iterations $k'\in [k_0,k-1]$ of the algorithm, the vertex $u$ is only connected with vertices in $CB_{k-1}(u)$. Observe that $CB_{k-1}(u)\subseteq B_{k-1}(u)$. By Lemma \ref{lemma isprse bunch}, w.h.p., the size of $B_{k-1}(u)$ is  $O(n^\rho\cdot{\log n})$.

It follows that, w.h.p., the contribution of $u$ to the set $\mathcal{H}$ is at most 

\begin{equation*}
        \begin{array}{lclclclclc}
             n^\rho\cdot {\log n} + \lambda &\leq & O(n^\rho{\log n}).
        \end{array}
    \end{equation*}
\end{proof}

We remark that if a vertex  $u$ has selected some vertex $v$ to be its pivot in iteration $k$ of Algorithm \ref{alg connect}, then $v$ is a valid pivot for all consecutive iterations of Algorithm \ref{alg connect}. Forcing $u$ to repeatedly select $v$ as its pivot slightly decreases the number of edges in the hopset $H$, but it is asymptotically meaningless, and it complicates the construction. 

By Lemma \ref{lemma size hop 1}, we derive the following corollary. 

\begin{restatable}{corollary}{CorHop}
        \label{coro size hop}
    The size of $\mathcal{H}$ is, w.h.p.,  $O(n^{1+\rho}{\log^2 n})$.
\end{restatable}

\subsection{Analysis of the Stretch}\label{sec stretch hop}

In this section, we analyze the stretch of the hopset $\mathcal{H} = \bigcup_{k\in\krange} H_{k} $.
Recall that every iteration $k\in \krange$  of Algorithm \ref{alg connect} receives as input a near-exact hopset. Denote by $\epsilon_{k-1}$ the multiplicative error of the input hopset for iteration $k$. 
Also, recall that for every $k\in \krange$, the set $H_k$ is the output of execution $k$ of Algorithm \ref{alg connect}.

Recall that by Corollary \ref{coro Q},  w.h.p.,  the set $H_k$ guarantees
\begin{equation}\label{eq q stretch guar}
        \begin{array}{lclclclclclcl}
        d^{(\beta_\mathcal{H})}_{G \cup H_k} (u,v) \leq  (1+\epsilon_{k-1} +28\epsilon (\hopell) ) d_G(u,v) +28\alpha\cdot \eps{\ell-1} ,
\end{array}
\end{equation}
for pairs of vertices $u,v$ with $d_G(u,v)\leq \alpha\eps{\ell} = 2^{k+1}$.
We provide an assignment to $\epsilon_k$ such that eq. \ref{eq q stretch guar} implies that the hopset $H_k$ is indeed a $(1+\epsilon_{k},\beta_\mch)$-hopset for the scale $(2^k,2^{k+1}]$. In particular, recall that  $\beta_\mathcal{H} = 6(\frac{1}{\epsilon}+3)^\ell$. For every $k\in \krange$, set $\epsilon_k = \frac{1}{\rho}70\epsilon\ell\cdot k$. In the following lemma, we show that this assignment for $\epsilon_k$ is satisfactory. 
 
\begin{restatable}{lemma}{HopStretch}
        \label{lemma H hopset}
    Let $k\in [k_0-1,\lambda]$ and ${H} = \bigcup_{k'=k_0}^{k} H_{k'} $. 
    For every pair of vertices $u,v\in V$ with $d_G(u,v)\leq 2^{k+1}$, w.h.p., we have
    \begin{equation*}
    \begin{array}{lclclclc}
         d_G(u,v)&\leq& d^{(\beta_\mathcal{H})}_{G\cup H}(u,v) &\leq &(1+ \epsilon_k)d_G(u,v).
    \end{array}    
    \end{equation*}
\end{restatable}
\begin{proof}
    The proof is by induction on the index $k$. For $k= k_0-1$, we have  $H = \emptyset $. In addition, $2^{k+1} \leq \beta_\mathcal{H}$ since $k_0 = \lfloor {\log \beta_\mathcal{H}}\rfloor$. Indeed, for every pair of vertices $u,v\in V$ with $d_G(u,v)\leq 2^k$ we have 
    \begin{equation*}
    \begin{array}{lclclc}     
        d_G(u,v)& = &  d^{(\beta_\mathcal{H})}_{G\cup H}(u,v) .
    \end{array}
    \end{equation*}
    
   Assume that the claim holds for $k' \in [k_0-1, k-1]$, and prove it holds for $k\in \krange$. 
   Recall that execution $k$ of Algorithm \ref{alg connect} requires a hopset $H$ for the scale $ [0,\frac{1}{2}\alpha\eps{\ell}]$ as input. Also recall that $\alpha = 2^{k+1}\cdot \epsilon^\ell$. Thus, the algorithm expects a hopset for the scale $ [0,2^{k}]$. 
    By the induction hypothesis, we have that the set ${H} = \bigcup_{k'=k_0-1}^{k-1} H_{k'} $ is a $(1+\epsilon_{k-1},\beta_\mathcal{H})$-hopset for distances up to $2^{k}$. 
    Thus, it can be used as the hopset input for execution $k$ of Algorithm \ref{alg connect}.

    Consider now a pair of vertices $u,v$ with $d_G(u,v)\in (2^{k},2^{k+1}]$. Recall that $\rho<1/2$. Also, recall that when the hopset $H_k$ was
    constructed, the parameter $\alpha$ was set to  $  \epsilon^{\ell}2^{k+1}$. Together with Corollary \ref{coro Q}, we have that, w.h.p.,  the set $H_k$ guarantees

    \begin{equation*}
        \begin{array}{lclclclclclcl}
             
        d^{(\beta_\mathcal{H})}_{G \cup H_k} (u,v) &\leq & (1+\epsilon_{k-1} +28\epsilon (\hopell) ) d_G(u,v) \\ && \quad+28\alpha\cdot \eps{\ell-1} 
        \\
        &\leq & 
        (1+\epsilon_{k-1} +28\epsilon (\frac{1}{\rho}+1) ) d_G(u,v) \\ && \quad+28\cdot 2^{k+1}\epsilon^\ell\cdot \eps{\ell-1} 
        \\
        
        &\leq & 
        (1+\epsilon_{k-1} +\frac{28\epsilon +28\epsilon\rho}{\rho} ) d_G(u,v) +28\cdot\epsilon\cdot 2^{k+1}\\
        
        &\leq & 
        (1+\epsilon_{k-1} +\frac{42\epsilon}{\rho} ) d_G(u,v) +28\cdot\epsilon\cdot 2^{k+1} . 
        
        \end{array}
    \end{equation*}
    Recall  
    that $\epsilon_{k'} =  \frac{70\epsilon\cdot k'}{\rho}$ for every $k'\in \krange$, and $\rho<1/2$. Since $d_G(u,v) > 2^k$, we have 
    \begin{equation*}
        \begin{array}{lclclclclclcl}
        
        d^{(\beta_\mathcal{H})}_{G \cup H_k} (u,v) 

        &\leq & (1+\epsilon_{k-1} + \frac{42\epsilon}{\rho}  ) d_G(u,v) +28\epsilon\cdot 2 \cdot 2^{k}\\

        &< & (1+\epsilon_{k-1} + \frac{42\epsilon}{\rho}  ) d_G(u,v) +28\epsilon\cdot \frac{1}{\rho}\cdot d_G(u,v)
         
        \\
        &\leq & (1+\frac{70\epsilon(k-1)}{\rho} +\frac{70\epsilon }{\rho} ) d_G(u,v) \\ & = &  (1+\epsilon_{k} ) d_G(u,v).
        \end{array}
    \end{equation*}
        
    To summarize,  for every pair of vertices $u,v$ with $d_G(u,v)\leq 2^{k+1}$, we have $$d_G(u,v) \leq d^{(\beta_\mathcal{H})}_{G\cup H} (u,v) \leq (1+\epsilon_{k} ) d_G(u,v).$$
\end{proof}




Observe that Lemma \ref{lemma H hopset} implies that $H = \bigcup_{j \in [0,k]} H_j$ can indeed be used as the hopset input for the $k+1$ iteration of the algorithm.

Recall that $\mathcal{H} = \bigcup_{k=k_0}^{\lambda}H_{k}$.
 By Lemma \ref{lemma H hopset}, we derive the following corollary.

\begin{corollary}
    \label{coro hop st}
    For every pair of vertices $u,v\in V$ with distance $d_G(u,v) \leq 2^\lambda$, w.h.p., we have 
\begin{equation*}
        \begin{array}{lclclclclc}
    d_G(u,v)&\leq& d^{(\beta_\mathcal{H})}_{G\cup \mathcal{H}}(u,v) 
    &\leq& \left(1+\frac{70\epsilon\cdot \lambda}{\rho}\right)d_G(u,v).
    \end{array}
    \end{equation*}
\end{corollary}

\paragraph{Rescaling}\label{sec rescale hop}
In this section we rescale $\epsilon$ to get a $(1+\epsilon_\mathcal{H},\beta_\mathcal{H})$-hopset. 
Recall that $\lambda ={\lceil {\log t}\rceil -1}$. Denote $\epsilon_\mathcal{H} = \frac{70\epsilon {\log t}}{\rho}$.
The condition $\epsilon<1/10$ now translates to $\frac{\epsilon_\mathcal{H}\rho}{70{\log t}} \leq 1/10$, which holds trivially for $\epsilon_\mathcal{H},\rho<1$ and $t>1$. We replace it with the stronger condition $\epsilon_\mathcal{H}<1$. 
The hopbound $\beta_\mathcal{H}$  now translates to $\beta_\mathcal{H} 
  = 6(
  \frac{70{\log t}}{\epsilon_\mathcal{H}\rho}+3)^{\frac{1}{\rho} +1}$

Note that by Corollary \ref{coro Q}, for every $k\in \krange$, the memory required to construct each hopset $H_k$ is 
$\tilde{O}((|E|+|\bigcup_{k'\in [k_0,k-1]}H_{k'}|+|H_k|)n^\rho)$. Also, note that $ \bigcup_{k'\in [k_0,k-1]}H_{k'} \cup H_k \subseteq \mch$.  By \ref{coro size hop}, w.h.p., we have $|\mathcal{H}| = \tilde{O}(n^{1+\rho})$. Thus, we can construct the entire hopset $\mch$ using 
$\tilde{O}((|E|\cdot +n^{1+\rho})n^\rho)$  total memory.

By Corollaries \ref{coro Q} \ref{coro size hop} and \ref{coro hop st}, we derive the following corollary, which summarizes the properties of our hopset. 


\begin{restatable}{corollary}{CorFinalHop}
        \label{coro hop}
    Let $G=(V,E)$ be an unweighted, undirected graph on $n$-vertices and let $\rho \in [1/{\log {\log n}},1/2],\epsilon_\mathcal{H}<1 $ and $t>1$ be input parameters.
    Our algorithm computes. w.h.p., a $(1+\epsilon_\mathcal{H},\beta_\mathcal{H})$-hopset of size  $\tilde{O}(n^{1+\rho})$ for distances up to $t$ in the graph $G$, where 
    \begin{equation*}
        \begin{array}{lclclclclc}
               \beta_\mathcal{H} &=& O\left(\frac{{\log t}}{\epsilon_\mathcal{H}\rho}\right)^{\frac{1}{\rho} +1} . 
        \end{array}
    \end{equation*}
    The algorithm terminates within $O(\frac{\beta_\mathcal{H}{\log t}}{\rho\gamma})$ communication rounds when 
    using machines with $O(n^\gamma)$ memory, and $\tilde{O}((|E|+n^{1+\rho}) n^{\rho})$ total memory.
\end{restatable}



Denote by ${\log ^{(x)} n}$ the $x$-times iterated logarithm of $n$, i.e., ${\log ^{(2)}n}$ $  = {\log {\log n}}$. 
Recall that we use our hopset to build near-exact emulators. Therefore, we have a special interest in the case where 
$ t   = \tval $, 
for $\epsilon_{M} \leq 1$ and $\rho > 1/ {\log^{(2)} n}$.
Note that in this case, $\beta_\mch$ is at most  
\begin{equation*}
    \begin{array}{lclclclclc}
        O\left(\frac{ {\log^{(2)} n} \cdot { \log \left( \frac{{\log^{(2)} n}}{\epsilon_{M}} \right) }}{\epsilon_\mathcal{H}\rho}\right)^{\frac{1}{\rho} +1}&  =& 
        
        O\left(\frac{ {\log^{(2)} n} 
        \cdot ({\log^{(3)} n} -{\log \epsilon_{M}})}
        {\epsilon_\mathcal{H}\rho}\right)^{\frac{1}{\rho} +1},
    \end{array}
\end{equation*}
The hopset properties are summarized in the following theorem.

\begin{theorem}
    \label{theo hopset}
        Let $G=(V,E)$ be an unweighted, undirected graph on $n$-vertices and let 
        $\rho \in[ 1/{{\log^{(2)} n}}, 1/2]$, $\epsilon_\mathcal{H}<1 $ and $\epsilon_{M}\leq 1$ be input parameters.
        Our algorithm computes, w.h.p.,  a $(1+\epsilon_\mathcal{H},\beta_\mathcal{H})$-hopset of size $\tilde{O}(n^{1+\rho})$ for distances up to $t $   in the graph $G$, where $t =\tval$ and 
\begin{equation*}
    \begin{array}{lclclclclc}
         \beta_\mathcal{H} &=& 
        O\left(\frac{ {\log^{(2)} n} 
        \cdot ({\log^{(3)} n} -{\log \epsilon_{M}})}
        {\epsilon_\mathcal{H}\rho}\right)^{\frac{1}{\rho} +1}.
    \end{array}
\end{equation*}
    In sublinear MPC when using machines with $O(n^\gamma)$ memory, and $\tilde{O}((|E|+n^{1+\rho}) n^{\rho})$ total memory, the construction time is
    \begin{equation*}
    \begin{array}{lclclclclc}
        O\left( \frac{1}{\gamma}\left( \frac{ 
        {\log^{(2)} n}\cdot ({\log^{(3)}n} -{\log {\epsilon_{M}}})
        }
        {\epsilon_\mch\rho}\right)^{\frac{1}{\rho} +2}\right). 
    \end{array}
\end{equation*}

\end{theorem}

\section{Near-Additive Emulators}\label{sec emu}
    
In this section, we show how the general framework from Section \ref{sec general} can be used for the construction of near-additive emulators with $\tilde{O}(n)$ size. 
We are given an unweighted, undirected graph $G=(V,E)$  on $n$ vertices, parameters $\epsilon <1/10$ and $\rho \in[ 1/{\log {\log n}}, 1/2]$. We set $\kappa = {\log n}$ and $\alpha  = 1$.

In addition, we build $(1+\epsilon_\mch,\beta_\mch)$-hopset $\mch$ for distances up to $t = \tval$ in the graph $G$, using the algorithm devised in Section \ref{sec hop}. For this aim, we provide the algorithm from Section \ref{sec hop} with the parameters $\epsilon_{M} = \epsilon_\mch = 28\epsilon (\ell+1)$ and $\rho$. 
In Section \ref{sec rescale emu} we show that $\mch$ is a satisfactory hopset.

The algorithm in Section \ref{sec general} takes the parameters $\epsilon, \kappa, \rho, \alpha$, and the hopset $\mch$ as inputs. Let $M$ be the output of the algorithm.

Corollary \ref{coro Q} provides a straightforward upper bound on both the size of the output $M$ as well as the round complexity of the algorithm. It is left to show that the set $M$ is indeed a near-additive emulator.

\paragraph{Analysis of the Stretch}
By Corollary \ref{coro Q},  for $\ell = \emuell  $, for every pair of vertices $u,v\in V$  with $d_G(u,v)\leq \eps{\ell} $ we have  
    \begin{equation}\label{eq emu st}
    \begin{array}{lclclclclc}
        d_G(u,v)\leq d_{M}(u,v) 
        \leq (1+\epsilon_H +28\epsilon \ell ) d_G(u,v) +28\cdot \eps{\ell-1}  .
    \end{array}
    \end{equation}
Initially, this relationship might appear unsatisfactory because it doesn't explicitly establish an upper limit on the stretch for distant pairs of vertices. However, we will demonstrate that it is indeed sufficient. Essentially, when considering a pair of vertices $u,v$, one can break down a shortest path between them into segments of precisely $\epsilon \ell$ in length, with the exception of the last segment which may be shorter. Equation \ref{eq emu st} is applicable to each of these segments. For each segment, excluding the last one, the additive term $28\cdot \epsilon (\ell-1)$ represents only an $\epsilon'$ fraction of the segment's length, where $\epsilon'$ is appropriately defined. Refer to Figure \ref{fig emu} for a visual representation of this concept. The formal proof of this argument is presented in the following lemma.

\begin{restatable}{lemma}{EmuStretch}\label{lemma st emu}
        For every pair of vertices $u,v\in V$, we have 
    \begin{equation*}
        \begin{array}{lclclclclclc}
             
        d_G(u,v)&\leq &  d_{M}(u,v) &\leq& \left(1+2\epsilon_{M}   \right) \cdot  {d_G(u,v)} +28 \eps{\ell-1} .
        \end{array}
    \end{equation*}
\end{restatable}

             

\begin{figure}
        \centering
        \includegraphics[scale = 0.15]{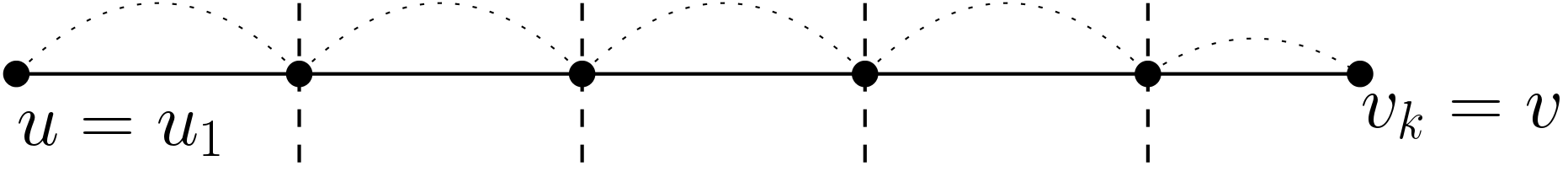}
        \caption{The $u-v$ path in $Q$. The solid line depicts the shortest path $\pi(u,v)$ in $G$. Dashed lines depicts the division of $\pi(u,v)$ into segments. Dotted lines depict the path we have between the endpoints of each segment by Corollary \ref{coro Q}.   }
        \label{fig emu}
    \end{figure}

\begin{proof}
    Let $u,v \in V$, and let $\pi(u,v)$ be a shortest path between them in the graph $G$. Let $r = \left\lceil \frac{d_G(u,v)}{\eps{\ell}}\right\rceil $. We divide the path $\pi(u,v)$ into $r$ segments $S_1, S_2, \dots, S_{r}$, each of length exactly $\eps{\ell}$, except for the last segment which may be shorter. For every segment $S_i$, denote by $u_i$ its first vertex and $v_i$ its last vertex. Observe that for all $i\in [1,r-1]$ we have $v_i = u_{i+1}$. In addition, we also have $u= u_1$ and $v = v_{r}$.  See Figure \ref{fig emu} for an illustration.

    By Corollary \ref{coro Q}, we have that for every index $i\in [1,r]$, it holds that 
       \begin{equation*}
        \begin{array}{lclclclclclc}
             d_G(u_i,v_i)&\leq& d_{M}(u_i,v_i) \\&\leq& (1+\epsilon_\mch +28\epsilon \ell ) d_G(u_i,v_i) +28\cdot \eps{\ell-1}  .
        \end{array}
    \end{equation*}
   In particular, we have $d_G(u_i,v_i)  \leq d_{M} (u_i,v_i)$ for all $i\in [1,r]$, and thus 
    $d_G(u,v)\leq d_{M}(u,v)$.

    Moreover, 
    \begin{equation*}
        \begin{array}{lclclclclclc}
             d_{M}(u,v) &\leq & \sum_{i=1}^{r} d_{M}(u_i,v_i) \\
             &\leq & \sum_{i=1}^{r} (1+\epsilon_H +28\epsilon \ell ) d_G(u_i,v_i) +28 \eps{\ell-1} \\

             &\leq &  (1+\epsilon_H +28\epsilon \ell ) d_G(u,v) + 
             \left\lceil \frac{d_G(u,v)}{\eps{\ell}} \right\rceil \cdot 28 \eps{\ell-1} \\
             
             &= &  (1+\epsilon_H +28\epsilon \ell ) d_G(u,v) + 
              28\epsilon \cdot d_G(u,v)
              \\ && \quad+  28 \eps{\ell-1} \\

             & = &  (1+\epsilon_H +28\epsilon (\ell+1) ) d_G(u,v) +  28 \eps{\ell-1}. 
                 
        \end{array}
    \end{equation*}
    
    Recall that $\epsilon_{M}= \epsilon_\mch = 28\epsilon (\ell+1)$. It follows that

    \begin{equation*}
        \begin{array}{lclclclclclc}
             d_{M}(u,v) &\leq &  (1+2\epsilon_{M}) d_G(u,v) +  28 \eps{\ell-1}.                  
        \end{array}
    \end{equation*}
\end{proof}

\paragraph{Rescaling} \label{sec rescale emu}

In this section we rescale $\epsilon$ to obtain a $(1+\epsilon_{M},\beta_{M})$-emulator. 
Set $\epsilon_{M} =28\epsilon (\ell+1)$. The condition $\epsilon< 1/10$ is translated to 
${\epsilon_{M}}/{28(\ell+1)}\leq 1/10$. We replace it with the stronger condition $\epsilon_{M}<1/3$.  
The additive term $28 \eps{\ell-1}$ now translates to $28\left(\frac{28 (\ell+1)}{\epsilon_{M}}\right)^{\ell-1}$.
Recall that $\ell = \emuell $, $\kappa = {\log n}$ and  $\rho \geq 1/\kappa$. Then, $\ell \leq {\log (\rho{\log n})} + 1/\rho$. Denote 
\begin{equation*}
\begin{array}{lclcl}
    \beta_{M}  &=& \betaemuko.
\end{array}
\end{equation*}

Next, we show that $\mch$ is indeed a satisfactory hopset. Note that our construction requires a hopset for distances up to $t'$, where 
\begin{equation*}
        \begin{array}{lclclclclc}
       t' =  \eps{\emuell} \leq  
        \left(  \frac{28(\emuellbk+1)}{\epsilon_{M}}\right)^{\emuellbk}.
        \end{array}
    \end{equation*}
    Recall that $\mch$ is a  hopset for distances up to $t$.  Observe that 
\begin{equation*}
        \begin{array}{lclclclclc}
             
t'& = &\left(  \frac{28(\emuellbk+1)}{\epsilon_{M}}\right)^{\emuellbk} \\
&\leq & \tval & = & t .

        \end{array}
    \end{equation*}

Thus, $\mch$ is a satisfactory hopset for the construction. 

\paragraph{Round complexity}
The round complexity and the memory usage of the algorithm are dominated by those needed to compute the hopset $\mch$.

\subsection{Round complexity}\label{sec append round complexity emu}
In this section, we analyze the round complexity of the construction of $\mch$. 
By Theorem \ref{theo hopset}, the construction time of $\mch$ is $$O\left( \frac{1}{\gamma}\left( \frac{ 
        {\log^{(2)} n}\cdot ({\log^{(3)}n} -{\log {\epsilon_{M}}})
        }
        {\epsilon_\mch\rho}\right)^{\frac{1}{\rho} +2}\right) $$ in sublinear MPC when using machines with $O(n^\gamma)$ memory, and $\tilde{O}((|E|+n^{1+\rho}) n^{\rho})$ total memory.

By Corollary \ref{coro Q}, given the hopset $\mch$, the emulator $M$ can be constructed  in $O(\beta_\mch/\gamma)$ rounds of sublinear MPC using machines with $O(n^\gamma)$ memory and $\tilde{O}((|E|+|\mch|+|M|)n^\rho)$ total memory. Observe that $|M| = \tilde{O}(n)$, thus the round complexity and the memory usage are dominated by those required for the construction of $\mch$. 
Recall that $\epsilon_\mch = \epsilon_{M}$.

By Corollary \ref{coro Q} and Lemma \ref{lemma st emu}, we derive the following corollary.

   \begin{restatable}{theorem}{EmuFinal}
    \label{theo emulators}
    Let $G=(V,E)$ be an unweighted, undirected graph on $n$ vertices, and let $\epsilon_{M} <1/3$,  $\rho \in[ 1/{\log {\log n}}, 1/2]$.   
    Our algorithm constructs an emulator $M$, of size $\tilde{O}(n) $, such that, w.h.p., for every pair of vertices $u,v\in V$   we have 
    \begin{equation*}
        \begin{array}{lclc}
    d_G(u,v) &\leq& d_{M} (u,v) \leq (1+2\epsilon_{M}) d_G(u,v) \\ && \quad+\betaemuk.
        \end{array}
    \end{equation*}
    The round complexity is  $O\left( \frac{1}{\gamma}\left( \frac{ 
        {\log^{(2)} n}\cdot ({\log^{(3)}n} -{\log {\epsilon_{M}}})
        }
        {\epsilon_{M}\rho}\right)^{\frac{1}{\rho} +2}\right) $
    in sublinear MPC when using machines with $O(n^\gamma)$ memory, and $\tilde{O}((|E|+n^{1+\rho})\cdot n^{\rho})$ total memory.
\end{restatable}

\section{Applications}\label{sec applications} 

    \subsection{Approximate Single Source Shortest Paths}\label{sec sssp}

In this section,  we devise an algorithm that uses the general framework to compute $(1+\epsilon)$-approximate shortest distances from a single source. Let $G = (V, E)$ be an unweighted, undirected graph on $n$ vertices, and let $\epsilon<1$, $\gamma \leq 1$ and $\rho \in[ 1/{\log {\log n}}, 1/2]$ be parameters. Let $s\in V$ be a source vertex. 

Generally speaking, our algorithm builds a $(1+\epsilon_\mch,\beta_\mch)$ limited-scale hopset $\mathcal{H}$ and a $(1+\epsilon_{M},\beta_{M})$-emulator $M$. 
Then, the hopset and the emulator are used to compute distance estimates for all vertices. For every vertex, the minimal estimate is chosen.



\subsubsection{Near-Exact Hopsets} \label{sec near-exact hopsets msssp}
Set $t = \tval$ and $\epsilon_\mch = \epsilon/3$. We provide the algorithm from Section \ref{sec hop} with the graph $G$, parameters $\epsilon_\mathcal{H},\rho$ and $t$.  By Theorem \ref{theo hopset}, the algorithm from Section \ref{sec hop} computes a $(1+\epsilon_\mathcal{H},\beta_\mathcal{H})$-hopset for distances up to $t$ in the graph $G$. 
\subsubsection{Near-Additive Emulators}\label{sec emulators msssp}
Set $\epsilon_{M} = \epsilon/3$ and $\kappa = {\log n}$. 
Recall that the hopset $\mathcal{H}$  is a 
$(1+\epsilon_\mathcal{H}, \beta_{\mathcal{H}})$-hopset for distances up to $t $ in the graph $G$. 
By Theorem \ref{theo emulators}, given the graph $G$, the parameters 
$\epsilon_{M}, \kappa,\rho$, w.h.p., the algorithm from Section \ref{sec emu} constructs a $(1+\epsilon_{\mathcal{H}}+\epsilon_{M},\beta_{M})$-emulator $M$ of size $\tilde{O}(n)$, where 
\begin{equation*}
    \begin{array}{lclclclclc}
         \beta_{M} &=& \betaemuk.
    \end{array}
\end{equation*}

   
\subsubsection{Computing Distances}\label{sec msssp computing distances}
We are ready to compute approximate distances. Let $s\in V$ be the source vertex. For every vertex $v\in V$, let $\mathpzc{d}(v)$ be the estimate of $v$ for $d_G(s,v)$. Set $\mathpzc{d}(s)=0$. For every $v\in V\setminus \{s\}$ set  $\mathpzc{d}(v) = \infty$.
First, we compute the $(\beta_\mch)$-hops limited distance in $G\cup \mathcal{H}$ from $s$ to all vertices $v\in V$. 
To achieve this, a Bellman-Ford exploration is executed from $s$ to $\beta_\mathcal{H}$-hops in the graph $G\cup \mathcal{H}$. 
By Theorem \ref{theo restricted bellman-ford}, this exploration requires $O(\beta_\mch/\gamma)$ time using machines with ${O}(n^\gamma)$ memory, and a total memory of $\tilde{O}((|E|+|\mch|)n^{\rho})$. For every vertex $v\in V$, a single machine from $M(v)$ sends 
$\langle v, \dghBH(s,v)\rangle$ to the machine $M^*$ with $\tilde{O}(n)$ memory. 


Next, the edges of the emulator $M$ are aggregated into $M^*$ that has $\tilde{O}(n)$ memory. This machine internally computes $d_{M}(s,v)$ for every pair 
$s\in S$ and $v\in V$.  For every $v\in V$, the machine $M^*$ holds both $\dghBH(s,v)$ and $d_{M}(s,v)$. Thus, it can internally compute $\mathpzc{d}(v) = {\min\{\dghBH(s,v),d_{M}(s,v)  \} } $.

In the following lemma, we prove that the algorithm computes $(1+\epsilon)$-approximate distances. 
\begin{lemma}
    \label{lemma msssp error}
    Let $v\in V$. Then, $\mathpzc{d}(v) \leq (1+\epsilon)d_G(v,s)$.
\end{lemma}
\begin{proof}
    First, consider the case where $d_G(v,s) < t$. 
        Recall that $\mathcal{H}$ is a $(1+\epsilon_{\mathcal{H}},\beta_{\mathcal{H}})$-hopset for distances up to $t$, and that $\epsilon_\mch<\epsilon$. Thus, 
$\mathpzc{d}(v)  \leq  \dghBH(s,v) \leq   (1+\epsilon)d_G(u,v)$

       Next, consider the case where $d_G(v,s)\geq t$. 
        Note that 
        \begin{equation}
            \label{eq emu guar}
            \begin{array}{llllllclclc}
            d_{M}(s,v) &\leq
        (1+\epsilon_{\mathcal{H}} + \epsilon_{M})d_G(u,v) + \beta_{M}, & and \\
            \beta_{M} & =  \betaemuk\\ 
            \leq & 
            
            \tval \cdot \epsilon_{M} 
             &= 
            t\cdot \epsilon_{M}.
            \end{array}
        \end{equation}
        Recall that $\epsilon_{M} = \epsilon_\mch = \epsilon/3$. 
        Hence,
        \begin{equation}
            \label{eq sv approximation}
            \begin{array}{lclclclc}
            d_{M}(s,v) \leq  (1+\epsilon_\mathcal{H}+\epsilon_{M})d_G(s,v) + \beta_{M} \leq 
             (1+\epsilon)d_G(s,v). 
            \end{array}
        \end{equation}    
\end{proof}
\paragraph{Round Complexity}\label{sec append sssp round complexity}

In the following lemma, we analyze the round complexity of the algorithm. 

\begin{lemma}
    \label{lemma sssp computational complexity}
    In heterogeneous MPC, when using a single machine with $\tilde{O}(n)$ memory and additional machines with size $O(n^\gamma)$ and a total memory of $\tilde{O}((|E|+n^{1+\rho}) n^{\rho})$,
    the round complexity of the algorithm is 
    \begin{equation*}
    \begin{array}{lclclclclc}
        O\left( \frac{1}{\gamma}\left( \frac{ 
        {\log^{(2)} n}\cdot ({\log^{(3)}n} -{\log {\epsilon }})
        }
        {\epsilon \rho}\right)^{\frac{1}{\rho} +2}\right). 
    \end{array}
\end{equation*} 
\end{lemma}
\begin{proof}
Recall that $ \epsilon_\mathcal{H} = \epsilon_{M} = \epsilon /3$, and also  
\begin{equation*}
    \begin{array}{lclclclclc}
        \beta_\mathcal{H} &=& 
        O\left(\frac{ {\log^{(2)} n} 
        \cdot ({\log^{(3)} n} -{\log \epsilon})}
        {\epsilon\rho}\right)^{\frac{1}{\rho} +1}.
    \end{array}
\end{equation*}
    By Theorems \ref{theo hopset} and \ref{theo emulators}, w.h.p., in sublinear MPC, when using machines with $O(n^\gamma)$ memory and $\tilde{O}((|E|+n^{1+\rho}) n^{\rho})$ total memory, the hopset $\mch$ and the emulator $M$ can be constructed in time 
    \begin{equation*}
        \begin{array}{lclclclclc}
            O\left( \frac{1}{\gamma}\left( \frac{ 
            {\log^{(2)} n}\cdot ({\log^{(3)}n} -{\log {\epsilon}})
            }
            {\epsilon\rho}\right)^{\frac{1}{\rho} +2}\right). 
        \end{array}
    \end{equation*}
    By Theorem \ref{theo hopset}, the size of the hopset $\mch$ is $\tilde{O}(n^{1+\rho})$. Together with Theorem \ref{theo restricted bellman-ford}, we have that the Bellman-Ford exploration to depth $\beta_\mathcal{H}$ from the vertex $s$ in the graph $G\cup H$ can be computed in $O(\beta_\mathcal{H}/\gamma)$ rounds, when using machines with $O(n^\gamma)$ memory, and $\tilde{O}(|E|+n^{1+\rho})$ total memory. 
    
    Collecting the distance estimates computed by the Bellman-Ford explorations to a single machine $M^*$ with $\tilde{O}(n)$ memory can be done in a single round. 
    By Theorem \ref{theo emulators}, the size of the emulator $M$ is $\tilde{O}(n)$. Hence, collecting the edges of $M$ into $M^*$ also requires a single round. 
    
    The distance estimates in the emulator $M$, as well as the final estimates $\mathpzc{d}(v)={\min\{ \dghBH(s,v), d_{M}(s,v)\} } $ for every vertex $v\in V$ are computed locally by $M^*$.

\end{proof}

\paragraph{Summary}

The following theorem summarizes the algorithm. 

\begin{theorem}\label{theo sssp}
    Let $G=(V,E)$ be an unweighted, undirected graph on $n$ vertices, and let $\epsilon <1$ and $\rho \in[ 1/{\log {\log n}}, 1/2]$ be parameters.  Also, let $s\in V$ be a source vertex.
    For every $v\in V$, w.h.p., our randomized algorithm computes $(1+\epsilon)$-approximate shortest paths from $s$ in
    $ O\left( \frac{1}{\gamma}\left( \frac{ 
        {\log^{(2)} n}\cdot ({\log^{(3)}n} -{\log {\epsilon }})
        }
        {\epsilon \rho}\right)^{\frac{1}{\rho} +2}\right) $
     rounds of heterogeneous MPC,
      when using a single machine with $\tilde{O}(n)$ memory and additional machines with $O(n^\gamma)$ memory and a total memory of size $\tilde{O}((|E|+n^{1+\rho})\cdot n^{\rho})$.
\end{theorem}

\subsection{Approximate Multiple Source Shortest Paths}\label{sec multi distances}

In this section, we devise an algorithm that uses the general framework to compute $(1+\epsilon)$-approximate shortest distances from a set $S$ of sources in unweighted, undirected graphs. Let $G = (V, E)$ be an unweighted, undirected graph on $n$ vertices, and let $\epsilon<1$, $\gamma \leq 1$ and $\rho\in [1/{\log{\log n} } ,1/2]$ be parameters. Let $S\subset V$ be a set of sources. 

We build on the algorithm given in Section \ref{sec sssp}. 
Recall that we use a single machine with $\tilde{O}(n)$ memory and additional machines with size $O(n^\gamma)$ and a total memory of $\tilde{O}((|E|+n^{1+\rho})n^{\rho})$. In addition, recall that the size of the hopset $\mch$ is w.h.p., $\tilde{O}(n^{1+\rho})$. Thus, by Theorem \ref{theo multi source BF} we can compute up to $O(n^\rho)$ limited $\beta_\mch$ Bellman-Ford explorations, in parallel, in $O(\beta_\mch/\gamma)$ time.  However, the machine $M^*$ has only $\tilde{O}(n)$ memory, hence, it can use the emulator $M$ to compute and store distances only for $\poly({\log n})$ sources. 

Recall that $\rho >1/{\log {\log n}}$, thus, $n^\rho = \Omega(\poly{\log n})$ whenever $n>1$. Therefore, computing distances from $\poly({\log n})$ sources can be done using the same resources as in Section \ref{sec sssp}. Computing distances from $O(n^\rho)$ sources in $\poly({\log {\log n}})$ time can be achieved by employing additional machines with near-linear memory. 

In Sections \ref{sec polylogn sources} and \ref{sec nrho sources}, we discuss the adaptations to the algorithm provided in Section \ref{sec sssp} that provide approximate distances in $\poly({\log {\log n}})$ time from $\poly({\log n})$ and from $O(n^\rho)$ sources, respectively. 

\subsubsection{Computing Distances From $\poly({\log n})$ Sources}\label{sec polylogn sources}
Let $S\subset V$ be a set of vertices of size $\poly({\log n})$. The algorithm follows the algorithm given in Section \ref{sec sssp} to compute the hopset $\mch$ and the emulator $M$. Then, the algorithm executes Bellman-Ford explorations in $G\cup \mch$ from all sources in $S$, in parallel. Each exploration is executed to at most $\beta_\mch$ hops. Recall that when the exploration terminates, for every pair $(s,v)\in S\times V$, there is exactly one machine in $M(v)$ which stores $\dghBH(s,v)$. This machine sends $\dghBH(s,v)$ to the machine $M^*$. Since there are $\poly({\log n})$ sources in $S$, the machine $M^*$ receives $\tilde{O}(n)$ messages, which it stores locally. 

As in Section \ref{sec sssp}, the emulator $M$ is aggregated to the machine $M^*$, which internally computes Bellman-Ford explorations from all sources in $S$. When this process terminates, the machine $M^*$ stores $d_{M}(s,v)$ and $\dghBH(s,v)$ for every pair $(s,v)\in S\times V$. Thus, it can locally compute $\mathpzc{d}_s(v) = {\min \{d_{M}(s,v)\dghBH(s,v)\}}$ for every pair $(s,v)\in S\times V$. 

By the same argument given in Section \ref{sec sssp}, for every pair $(s,v)\in S\times V$, we have that $\mathpzc{d}_s(v)$ is now a $(1+\epsilon)$-approximation to $d_G(s,v)$. In addition, note that the running time of the algorithm is dominated by the time required to build the hopset $\mathcal{H}$ and the emulator $M$. Thus, the round complexity of the algorithm matches that of Theorem \ref{theo sssp}. 
The following theorem summarizes the properties of the algorithm.

\begin{theorem}\label{theo sssp log sources}
    Let $G=(V,E)$ be an unweighted, undirected graph on $n$ vertices, and let $\epsilon <1$ and $\rho \in[ 1/{\log {\log n}}, 1/2]$ be parameters.  Also, let $ S\subset V$ be a set of sources with size $\poly({\log n})$. 
    For every pair $(s,v)\in S\times V$, w.h.p., our algorithm computes $(1+\epsilon)$-approximate shortest paths in
    $ O\left( \frac{1}{\gamma}\left( \frac{ 
        {\log^{(2)} n}\cdot ({\log^{(3)}n} -{\log {\epsilon }})
        }
        {\epsilon \rho}\right)^{\frac{1}{\rho} +2}\right) $
      rounds of heterogeneous extra memory MPC,
      when using a single machine with $\tilde{O}(n)$ memory and additional machines with size $O(n^\gamma)$ and a total memory of size  $\tilde{O}((|E|+n^{1+\rho})n^{\rho})$. 
\end{theorem}

\subsubsection{Computing Distances From $O(n^\rho)$ Sources}\label{sec nrho sources}
In order to compute distances from a set $S$ of $O(n^\rho)$ sources, we use $O(n^\rho)$ machines with $\tilde{O}(n)$ memory. The rest of the machines have memory $O(n^\gamma)$, and the total memory remains  $\tilde{O}((|E|+n^{1+\rho})n^{\rho})$. We build on the algorithm provided in Section \ref{sec polylogn sources}, with two differences. First, every near-linear memory machine is assigned a source vertex $s\in S$. For every vertex $v\in V$, the distance estimate $\dghBH(s,v)$ will be sent to the machine assigned to $s$. Second, the entire emulator $M$ is also sent to this machine. Now, each near-linear machine can compute distance estimates from its source in the emulator $M$, and infer near-exact distance approximations from the source.  

Formally, we assume that all machines are indexed, and that the low-index machines are near-linear. We assign the sources in $S$ a unique ID in the range $[0,|S|-1]$. For this aim, we sort the list of sources on the machines, such that each machine has at most $O(n^\gamma)$ sources stored on it. Since each machine knows its index, it can infer the indices corresponding to the sources it stores, and therefore it can assign them with appropriate IDs. 
For every $j\in [0,|S|-1]$, machine $M_j$ is assigned to the source $s_j$. 
This process is computed before the parallel Bellman-Ford exploration is executed. 

When the parallel Bellman-Ford exploration terminates, each machine that stores a distance estimate $\dghBH(s_j,v)$ knows the unique index $j$ of the source $s_j$. As in Section \ref{sec polylogn sources}, for every pair $(s_j,v)\in S\times V$ there is only one machine $M$ assigned with the task of sending the distance estimate $\dghBH(s_j,v)$ to a near-linear machine. Now, the machine $M$ sends the distance estimate  $\dghBH(s_j,v)$ to the machine $M_j$. Observe that the number of distance estimates sent from each machine is bounded by the size of its memory. In addition, each machine that receives messages in this step has near-linear memory, and it receives $n$ distance estimates. Therefore, the I/O limitation is not violated. 

Next, the emulator $M$ is computed. Recall that when the computation of $M$ terminates, it is stored on a set of output machines. We use the copying procedure discussed in Section \ref{sec edge selec gen}, to create $|S|$ copies of $M$. Then, for every $j\in [0,|S|-1]$, the set of machines that store the $j$th copy of $M$ send it to machine $M_j$. Note that in this process, the number of messages sent by each machine is bounded by its memory. In addition, all machines that receive input have near-linear memory. Since the size of $M$ is $\tilde{O}(n)$, this also does not violate I/O limitations. 

Finally, for every $j\in [0,|S|-1]$, the machine $M_j$ computes distances from $s_j$ to all $v\in V$. In addition, the machine $M_j$ also stores the distance $\dghBH(s_j,v)$ for every $v\in V$. Thus, it can locally compute $\mathpzc{d}_{s_j}(v) = {\min \{ \dghBH(s_j,v), d_{M}(s_j,v) \}} $ for every $v\in V$. This completes the description of the algorithm. 

Note that the algorithm has the same round complexity as the algorithm given in Section \ref{sec polylogn sources}. However, this algorithm uses $O(n^\rho)$ near-linear machines. Therefore, it works in near-linear MPC, rather than in the heterogeneous model.

\begin{theorem}\label{theo sssp poly sources}
    Let $G=(V,E)$ be an unweighted, undirected graph on $n$ vertices, and let $\epsilon <1$ and $\rho \in[ 1/{\log {\log n}}, 1/2]$ be parameters.  Also, let $ S\subset V$ be a set of sources with size $O(n^\rho)$. 
    For every pair $(s,v)\in S\times V$, w.h.p., our algorithm computes $(1+\epsilon)$-approximate shortest paths in
    $ O\left( \frac{1}{\gamma}\left( \frac{ 
        {\log^{(2)} n}\cdot ({\log^{(3)}n} -{\log {\epsilon }})
        }
        {\epsilon \rho}\right)^{\frac{1}{\rho} +2}\right) $
      rounds of near-linear extra memory MPC,
      when using a $O(n^\rho)$ machines with $\tilde{O}(n)$ memory, additional machines with size $O(n^\gamma)$ and a total memory of  $\tilde{O}((|E|+n^{1+\rho})n^{\rho})$.
\end{theorem}

\section{APSP Constant Query Time}\label{sec APSP}
    
This section focuses on the computation of All-Pairs Shortest Paths (APSP). Given an unweighted, undirected graph $G= (V,E)$ on $n$-vertices and parameters $\epsilon \leq 1/2$, $\rho \in [1/{\log {\log n}},1/2]$ and $k \leq 1/\rho$, our randomized algorithm computes a distance oracle of size $\tilde{O}(n^{1+1/k})$, such that for any pair of vertices $u,v$ we can retrieve w.h.p. a $(1+\epsilon)(2k-1)$ approximation. The preprocessing requires
 $ O\left( \frac{1}{\gamma}\left( \frac{ 
        {\log^{(2)} n}\cdot ({\log^{(3)}n} -{\log {\epsilon }})
        }
        {\epsilon \rho}\right)^{\frac{1}{\rho} +2}\right) $
     rounds  and the query time is $O(1)$, in heterogeneous extra memory MPC, when using a single machine with $\tilde{O}(n)$ memory and additional machines with $O(n^\gamma)$ memory and a total memory of size $\tilde{O}((|E|+n^{1+\rho})\cdot n^{1/k})$.

This section begins with an overview of the motivation for our solution to the problem. Memory limitation implies that computing and storing all distances in the graph is impractical. Therefore, we aim for a compact data structure. This structure, when queried with a pair $(u,v)$, provides an approximation for the distance $d_G(u,v)$ in constant time.

Observe that the algorithm provided in Section \ref{sec sssp} can be used to solve this problem. In particular, building the near-exact hopset and near-additive emulator as in Section \ref{sec sssp} provides two solutions for the problem:
\begin{enumerate}
    \item A $(1+\epsilon)$-approximation with $O(1)$ preprocessing and  query time $\poly({\log {\log n}})$.
    \item A $(1+\epsilon,\beta_M)$-approximation with $\poly({\log {\log n}})$ preprocessing time and $O(1)$ query time, where $$\beta_{M} = \left(\frac{{\log {\log n}}}{\epsilon}\right)^{O({\log {\log n}})}.$$ 
\end{enumerate}
Given a distance query $(u,v)$, the first solution is obtained by computing all distances from $u$ in the emulator $M$ (which is done locally on a near-linear machine), and computing a restricted Bellman-Ford exploration from $u$ in the graph $G\cup H$. We can run the whole algorithm, including the construction of the emulator, in the query time, resulting in $\poly({\log {\log n}})$ query time. The second solution is obtained by computing the distance estimate only in the emulator $M$, and returning it. Here we compute the emulator in advance in $\poly({\log {\log n}})$ preprocessing time, and then the query time is just $O(1)$. 

Note that the above solutions either require $\poly({\log {\log n}})$ query time or have a large additive term in the approximation. Our next goal is to obtain a \emph{constant} approximation in \emph{constant} query time.
Recall that for distant pairs of vertices, the emulator $M$ guarantees a near-exact approximation. Therefore, the algorithm given in Section \ref{sec sssp} provides a satisfactory solution for distant pairs of vertices. Close pairs of vertices require a different approach.   

To deal with close pairs we use distance sketches. Essentially, a distance sketch is a data structure that, for every vertex $v\in V$, contains distances in $G$ to some carefully selected vertices $w\in V$, also referred to as the \textit{bunch} of $v$. 
Given a query $(u,v)$, the bunches of the vertices $u$ and $v$ are used to compute an approximation to the distance $d_G(u,v)$. To allow constant query time, we limit the size of the bunch of each vertex $v\in V$. 

Thorup and Zwick have provided an algorithm that computes distance sketches in the centralized model of computation. Their distance sketches provide $(2k-1)$-distance approximation and have query time of $O(k)$, but their construction requires that some vertices learn their distance from vertices that are far from them. This requirement is impractical when aiming at $\poly({\log {\log n}})$ time in the MPC model.

To overcome this hurdle, Dinitz and Nazary \cite{DinitzN19} have used near-exact hopsets. Generally speaking, they begin by computing a full-scale near-exact hopset $H$ in $\poly({\log n})$ time. Then, they use the hopset $H$ to compute their distance sketch. Their distance sketch provide $(1+\epsilon)(2k-1)$-distance approximation, for some appropriate $\epsilon < 1$.

Here, we cannot afford to compute a full-scale hopset in polylogarithmic time. However, we can use a limited-scale hopset to compute partial but satisfactory distance sketches. We introduce the concept of limited-scale distance sketches, where a $d$-limited-scale distance sketch offers a $(1+\epsilon)(2k-1)$-distance approximation exclusively for vertex pairs $u,v$ with $d_G(u,v)\leq d$.  Through the computation of an appropriate limited-scale distance sketch and a near-additive emulator, we can ensure $\poly({\log {\log n}})$ preprocessing time and $O(1)$ query time for every query $(u,v)\in V\times V$.

The rest of this section is organized as follows. Section \ref{sec APSP TZ ds} provides a brief overview of the algorithm of Thorup and Zwick for constructing distance sketches.  Section \ref{sec limited scale DS} provides the limited-scale distance sketches algorithm. 

\subsection{Thorup-Zwick Distance Sketches}\label{sec APSP TZ ds}
The input for the algorithm is a graph $G=(V,E)$ on $n$ vertices, and a parameter $k \geq 2$. 
The algorithm is composed of a preprocessing phase, which computes for every vertex $v\in V$ a table of approximate distances to a selected set of vertices. Then, given a query $u,v$, the algorithm uses a query algorithm that uses the tables of $u$ and $v$ to compute an estimate of the distance $d_G(u,v)$.

\subsubsection{Preprocessing}
The preprocessing step of the algorithm begins by sampling hierarchy $\mathcal{A} = A_0\supseteq A_1 \supseteq\dots\supseteq A_k $ where $A_0 = V$ and for every $i\in [0,k-2]$, each vertex of $A_i$ is sampled to $A_{i+1}$ with probability $n^{-1/k}$. Also,  define $A_{k} = \emptyset$. Then, each vertex $v\in V$ select $k$-pivots for itself, where the $i$-th pivot $p_i(v)$ is the closest $A_i$ vertex to $v$. 

Next, for every level $i\in [0,k-1]$, each vertex $v$ computes and stores the distance to $p_{i+1}(v)$ and to all $A_i$ vertices closer to $v$ than $p_{i+1}(v)$. Define the \textit{bunch} of $v$ to be the set $B(v)  = \bigcup_{i= 0}^{k-1} \{ w\in \aismai \ | \ d_G(w,v) < d_G(A_{i+1},v)\}$. The vertex $v$ stores the identities and distance $d_G(v,w)$ for all vertices $w\in B(v)$.  
This completes the description of the preprocessing step of the algorithm. Its pseudo code is given in Algorithm \ref{alg preprocessing TZ}.

\begin{algorithm}
   \caption{Distance Sketch Preprocessing}
    \begin{algorithmic}[1]
    \Statex \textbf{Input:} A graph $G= (V,E)$, a parameter $k\geq 2$.
    \Statex \textbf{Output:} A distance sketch $D$.
    \State $A_0 = V$, $A_k = \emptyset$
    \For {$i\in [1,k-1]$}
        \State each vertex of $A_i$ is sampled to $A_{i+1}$ with probability $n^{\frac{1}{k}}$
    \EndFor 
    \For {every $v\in V$}
        \For {$i\in [0,k-1]$}
            \State let $d_G(A_i,v) = 
            {\min \{d_G(w,v) \ | \ w \in A_i \} }$
                \State let $p_i(v) $ be a vertex in $A_i$ s.t. 
                $d_G(p_i(v),v) = d_G(A_i,v) $
        \EndFor
        \State $B(v) = \bigcup_{i=0}^{k-1}\{ w\in \aismai \ | \ d_G(w,v) < d_G(A_{i+1},v)    \} $
    \EndFor \\
    \Return for every $v\in V$, identity and distance $d_G(w,v)$ for every $w\in B(v)$
    \end{algorithmic}
       \label{alg preprocessing TZ}
\end{algorithm}

\subsubsection{Query Answering}
Generally speaking, given a query $(u,v)$, the query algorithm looks for the minimal index $i$ such that the $p_i(u)\in B(v)$ or $p_i(v)\in B(u)$. Once this index is found, the algorithm has essentially found a vertex $w\in B(v)\cap B(u)$. It returns $d_G(u,w)+d_G(w,v)$ as an estimate for the distance $d_G(u,v)$. The pseudo code of the query algorithm is given in Algorithm \ref{alg search TZ}.

\begin{algorithm}
   \caption{Compute Distances}
    \begin{algorithmic}[1]
      \Function{$dist_k$}{$u,v$} 
      \State $u_0\gets u\ ; \ v_0\gets v$
      \State $w_0\gets u\ ; \ i\gets 0$
      \While {$w_i\notin B(v_i)$}
      \State $i\gets i+1$
      \State $(u_i,v_i)\gets (v_{i-1},u_{i-1})$
      \State $w_i = p_i(u_i)$
      \EndWhile\\
        \Return $\dghBH(w_i,u_i)+\dghBH(w_i,v_i)$
    \EndFunction
    \end{algorithmic}
       \label{alg search TZ}
\end{algorithm} 

\subsubsection{Error Analysis Intuition}
The error analysis relies on two observations. The first observation is that $A_{k-1}\subseteq B(x)$ for every vertex $x\in V$. This implies that the $(k-1)$-pivot of $u$ belongs to $B(v)$ and vice versa. Consequently, the while loop terminates when $i = k-1$ or sooner.

The second observation concerns the distance between the vertex $u_i$ and its $i$-pivot $w_i$. Notably, whenever the algorithm assigns $w_i = p_i(u_i)$, it signifies that $w_{i-1}\notin B(v_{i-1})= B(u_i)$. Hence, $w_i$, which is the $i$-pivot of $u_i = v_{i-1}$, is closer to $v_{i-1}$ than $w_{i-1}$. Formally, 

\begin{equation}
    \begin{array}{lclclclc}
        d_G(v_{i-1},w_i) & \leq & d_G(v_{i-1},w_{i-1}).
    \end{array}
\end{equation}

By the triangle inequality, one can show 
\begin{equation*}
    \begin{array}{lclclclc}
        d_G(u_i,w_i) & = & d_G(v_{i-1},w_i)  
        \\& \leq & d_G(v_{i-1},u_{i-1})+d_G(u_{i-1},w_{i-1})
        \\
      & = &d_G(v,u)+d_G(u_{i-1},w_{i-1}).  
    \end{array}
\end{equation*}
This essentially implies that in every iteration of the while loop, the distance between $u_i$ and its pivot $w_i$ increases only by an additive factor of $d_G(u,v)$ w.r.t. the distance $d_G(u_{i-1},w_{i-1})$. Consequently, since $i$ is at most $k-1$, one can show that the algorithm returns a distance estimate which is at most $(2k-1)d_G(u,v)$.
This completes the intuition behind the analysis of the error of the distance sketch.

\subsection{Limited-Scale Distance Sketches}\label{sec limited scale DS}
The input for the algorithm is a graph $G=(V,E)$ on $n$ vertices and a parameter $k \geq 2$. In addition, the algorithm receives as input a  $(1+\epsilon_{\mathcal{H}},\beta_{\mathcal{H}})$-hopset $\mathcal{H}$ for distances up to $t$, with size $\tilde{O}(n^{1+\rho})$.  The output of the algorithm is a $d$-limited-scale distance sketch, for $d= \dlimits$.

\subsubsection{The Challenges in Using a Limited-Scale Hopset}

Using the limited-scale hopset $\mathcal{H}$ for constructing distance sketches gives rise to two challenges. The first challenge is that the hopset $\mathcal{H}$ does not allow us to compute distances between some pairs of vertices in a reasonable time.  In particular, this implies that computing the distance from every vertex $v\in V$ to all vertices in $A_{k-1}$ is prohibitively expensive. However, recall that the distance sketches are used only to approximate distances up to $d= \dlimits$. Therefore, we can modify the preprocessing and the query algorithm to use only the limited-scale hopset and still provide a meaningful stretch guarantee for pairs of vertices with distance up to $d$. For pairs of vertices with distance more than $d$, our algorithm may return \textit{out of range}.

The second challenge in using the limited-scale hopset $\mathcal{H}$ for constructing distance sketches is that the $\beta_{\mathcal{H}}$-hops limited distances in $G\cup \mathcal{H}$ do not adhere to the triangle inequality. This challenge occurs also when using a full-scale hopset. To illustrate, consider vertices 
$u,v,w\in V$ such that $d_G(u,w),d_G(w,v) = x$, and $d_G(w,v) = 2x$. It is possible that $\dghBH(u,w),\dghBH(w,v) = x$ while $\dghBH(u,v) = (1+\epsilon_{\mathcal{H}})2x$. Observe that this is true even when using a full-scale hopset. As a result, we can no longer say that if we entered the while loop in the query algorithm and assigned $w_i\gets p_i(u_i)$, it is because $w_i$ is closer to $u_i$ than $w_{i-1}$. 
This implies that the distance from $u_i$ to its pivot $w_i$ grows significantly more than in the Thorup-Zwick construction. A rescaling of $\epsilon_\mathcal{H}$ is possible, but it forces us to have $ \epsilon_\mathcal{H}  \leq 1/\exp(k) $. This approach is less favorable as the running time of the hopset algorithm depends on $1/\epsilon_{\mathcal{H}}$. 

Alternatively, we can modify the preprocessing algorithm.
Observe that whenever $$\dghBH(u_i,v_i),\dghBH(v_i,w_i)\leq t,$$ we have 
\begin{equation*}
    \begin{array}{lclclc}
        \dgmhb{2}(u_i,w_i) & \leq & \dghBH(u_i,v_i)+\dghBH(v_i,w_i).
    \end{array}
\end{equation*}
This incentivizes us to use $i\beta$ hops when we look for pivots and  $A_{i-1}$ vertices that are closer than the $i$-pivots.  This approach allows us to show that if the query algorithm has assigned a value to $w_i$, then the $u_i-w_i$ distance in the distance sketch is at most $(1+\epsilon_\mathcal{H})d_G(u,v)\cdot i$.

\subsubsection{The Preprocessing Algorithm}

Next, we outline the revised algorithm for computing a $d$-limited distance sketch $\mathcal{D}$ using the limited-scale hopset $\mathcal{H}$. 

The preprocessing algorithm begins by sampling a hierarchy as in Algorithm \ref{alg preprocessing TZ}. Recall that in Algorithm \ref{alg preprocessing TZ}, each vertex finds its $i$-pivot for every $i\in [0,k-1]$. Here, for every index $i\in [0,k-1]$, if there exists a vertex $w\in A_{i}$ such that $\dgmhb{i}(v,w)\leq t$, then $v$ selects its $i$-pivot to be the vertex $w\in A_i$ with minimal $\dgmhb{i}(v,w)$. In addition, for $i\in [1,k-1]$, the bunch of the vertex $v$ is added with $p_i(v)$ and with all vertices $u\in A_{i-1}\setminus A_i$ such that $\dgmhb{i}(u,v)< {\min\{\dgmhb{i}(v, A_i),t\}}$. Intuitively, if $v$ has an $i$-pivot, then $B(v)$ contains $p_i(v)$ and all vertices in $A_{i-1}$ that are closer to $v$ than the pivot (in $i\beta_\mathcal{H}$ hops limited distance in $G\cup \mathcal{H}$). In addition, the vertex $v$ stores $\dgmhb{i}(u,v)$ for all $A_{i-1}\setminus A_i$ vertices in $B(v)$, and $\dgmhb{i}(p_i(v),v)$ to the data sketch $\mathcal{D}$.   If $v$ does not have an $i$-pivot, then $B(v)$ contains all vertices in $A_{i-1}$ that are within distance at most $t$ from $v$ (in $(i\beta_\mathcal{H})$ hops limited distance in $G\cup \mathcal{H}$). In this case, the vertex $v$ stores $\dgmhb{i}(u,v)$ for all $A_{i-1}\setminus A_i$ vertices in $B(v)$.
This completes the description of the preprocessing algorithm. The pseudo code of the algorithm is given in \ref{alg preprocessing TZ MPC}. 

\begin{algorithm}
   \caption{Distance Sketch Preprocessing}
    \begin{algorithmic}[1]
    \Statex \textbf{input:} Graph $G=(V,E)$, a limited-scale $(1+\epsilon_{\mathcal{H}},\beta_{\mathcal{H}})$-hopset $\mathcal{H}$, a parameter $k$ and a parameter $d$. 
    \Statex a distance sketch for distances up to $d$. 
    \State $A_0 = V$, $A_k = \emptyset$
        \For {$i\in [1,k-1]$}
            \State each vertex of $A_i$ is sampled to $A_{i+1}$ with probability $n^{\frac{1}{k}}$
        \EndFor 
        \For {every $v\in V$}
            \State $p_0(v) = v$, $B(v) = \emptyset$ 
            \For {$i\in [1,k-1]$}
                \State let $\dgmhb{i}(A_i,v) = 
                {\min \{\dgmhb{i}(w,v) \ | \ w \in A_i \} }$
                \If { $\dgmhb{i}(A_i,v) \leq t$} 
                    \State let $p_i(v) $ be a vertex in $A_i$ s.t. 
                    $\dgmhb{i}(p_i(v),v) = \dgmhb{2}(A_i,v) $
                    \State $B(v) = B(v) \cup \{ w\in A_{i-1}\setminus A_i \ | \ \dgmhb{i}(w,v) < 
                    \dgmhb{i}(A_{i},v)    \} $
                \Else 
                    \State $B(v) = B(v) \cup \{ w\in A_{i-1}\setminus A_i \ | \ \dgmhb{i}(w,v) < t    \} $
                \EndIf
            \EndFor
        \EndFor\\
        \Return for every $v\in V$, for every pivot $p_i(V)$ its identity and  $\dgmhb{i}(p_i(v),v)$. For every $w\in B(v)\cup A_{i-1}\setminus A_i$ that is not a pivot of $v$, its identity and $\dgmhb{i}(w,v)$.   
    \end{algorithmic}
       \label{alg preprocessing TZ MPC}
\end{algorithm}

\subsubsection{The Query Algorithm} 

The query algorithm requires adaptation, due to the fact that some vertices
may lack $i$-pivots for some values of $i$. In particular, before the algorithm assigns
$w_i$ to be the $i$-pivot of $u_i$, we ask whether $u_i$ has an $i$-pivot. If the answer is no, the algorithm terminates without returning a value. We later show that this
happens only if $d_G(u, v) > d$. The pseudo code of the query algorithm is given
in Algorithm \ref{alg search TZ MPC}. This completes the description of the algorithm for computing $d$-limited-scale distance sketches.

\begin{algorithm}
   \caption{Compute Distances}
    \begin{algorithmic}[1]
      \Function{$dist$}{$u,v$} 
      \State $u_0\gets u\ ; \ v_0\gets v$
      \State $w_0\gets u\ ; \ i\gets 0$
      \While {$w_i\notin B(v_i)$}
          \State $i\gets i+1$
          \State $(u_i,v_i)\gets (v_{i-1},u_{i-1})$
          \If {$u_i$ has an $i$-pivot}
            \State $w_i = p_i(u_i)$
          \Else
            \State \textbf{return} $u,v$ are out of range
          \EndIf
      \EndWhile\\
        \Return $\mathcal{D}(w_i,u_i)+\mathcal{D}(w_i,v_i)$ 
    \EndFunction
    \end{algorithmic}
       \label{alg search TZ MPC}
\end{algorithm}

\subsubsection{Analysis of the Stretch}
We analyze the stretch of Algorithm \ref{alg search TZ MPC} for pairs of vertices $u,v$ with $d_G(u,v)\leq d$.
First, we provide an upper bound on $\dgmhb{i}(u_i,w_i)$, and $d_G(w_i,v_i)$ for every index $i$ such that $w_i$ was assigned a value. 

\begin{lemma}
    \label{lemma ui-wi mpc ds}
    Let $u,v$ be a pair of vertices such that $d_G(u,v)\leq d$. Let $i'$ be the maximal value assigned to $i$ by Algorithm \ref{alg search TZ MPC}. For every $i\in [0,i']$, we have 
    \begin{equation*}
        \begin{array}{clclclclcl}
            (1)& \dgmhb{i}(w_i,u_i) & \leq &
            (1+\epsilon_\mathcal{H})\cdot i \cdot d_G(u,v). 
            \\
            (2)& d_G(w_i,v_i) & \leq & t/(1+\epsilon).
        \end{array}
    \end{equation*} 
\end{lemma}
\begin{proof}
    The proof is by induction on the index $i$. Let $i=0$. For the first assertion of the lemma, both sides of the inequality $(1)$  are equal to $0$ and the claim holds. For the second assertion of the lemma, recall that $w_0 = u_0= u$, $v_0 = v$ and also $d = \dlimits$. Thus, we have 

    \begin{equation*}
        \begin{array}{lclclclclcl}
            d_G(w_0,v_0) & = & d_G(u,v) &\leq & \dlimits & \leq & \frac{t}{ (1+\epsilon) }.
        \end{array}
    \end{equation*}

    Assume the claim holds for $i\in [0,k-2]$ and prove it holds for $i+1$.
    
    \textbf{First assertion:}\\
    The proof consists of two parts. In the first part, we show that the algorithm has assigned a value to $w_{i+1}$, i.e., it did not return \textit{$u,v$ are out of range}. Then, we provide an upper bound on $\dgmhb{(i+1)}(w_{i+1},u_{i+1})$. 

    Since the algorithm has entered the while loop, we know that $w_i\notin B(v_i)$. This implies that at least one of the following conditions hold: 
    \begin{enumerate}
        \item\label{item 1}{$v_i$ has an $(i+1)$-pivot $p_{i+1}(v_i)$ such that $\dgmhb{(i+1)}(p_{i+1}(v_i),v_i) \leq \dgmhb{(i+1)}(w_i,v_i).$ }
        \item\label{item 2} { $\dgmhb{(i+1)}(w_i,v_i) >t$}.
    \end{enumerate}
     By the induction hypothesis, we have $d_G(w_i,v_i)\leq t/(1+\epsilon)$. Since $\mch$ is a $(1+\epsilon_\mch,\beta_\mch)$-hopset for distances up to $t$, we conclude that $\dghBH(w_i,v_i)\leq t$. Hence, condition \ref{item 2} cannot hold. Consequently, we conclude that condition \ref{item 1} holds. Specifically, we have 
     the vertex $u_{i+1} = v_{i}$ has an $(i+1)$-pivot, and so the algorithm has set $w_{i+1} = p_{i+1}(u_{i+1}) = p_{i+1}(v_i)$. Condition \ref{item 1} implies that 
    \begin{equation}\label{eq vi to its pivot}
        \begin{array}{lclclc}
            \dgmhb{(i+1)}(w_{i+1},u_{i+1}) & = & \dgmhb{(i+1)}(p_{i+1}(v_i),v_i) \\&\leq& \dgmhb{(i+1)}(w_i,v_i).
        \end{array}
    \end{equation}

    We now provide an upper bound on $\dgmhb{(i+1)}(w_{i+1},u_{i+1})$. By eq. \ref{eq vi to its pivot}, and the induction hypothesis, we have 

    \begin{equation}\label{eq wi1-ui1}
        \begin{array}{lclclc}
            \dgmhb{(i+1)}(w_{i+1},u_{i+1})  &\leq& \dgmhb{(i+1)}(w_i,v_i)
            \\&\leq& \dgmhb{i}(w_i,u_i) + \dghBH(u_i,v_i)
            \\
            &\leq& (1+\epsilon_\mathcal{H})\cdot i \cdot d_G(u,v)
            \\ && \quad
            + (1+\epsilon_\mch)d_G(u,v)\\
            &\leq& (1+\epsilon_\mathcal{H})\cdot (i+1) \cdot d_G(u,v).
        \end{array}
    \end{equation}

    \textbf{Second Assertion:}\\
    For the second assertion of the lemma, note that the triangle inequality holds for distance in $G$. Also, recall that $d_G(u,v)\leq \dlimits$. Note that $i\leq k-1$. Together eq. \ref{eq wi1-ui1}, we have 
    \begin{equation*}
        \begin{array}{lclclclclcl}
            d_G(w_{i+1},v_{i+1}) & \leq & d_G(w_{i+1},u_{i+1})+d_G(u_{i+1},v_{i+1}) 
            \\&\leq & \dgmhb{(i+1)}(w_{i+1},u_{i+1}) +d_G(u,v)
            \\&\leq & (1+\epsilon_\mathcal{H})\cdot ( k+1) \cdot \dlimits
            \\&\leq & \frac{t}{1+\epsilon_\mathcal{H}}.
        \end{array}
    \end{equation*}

\end{proof}

We are now ready to provide an upper bound on the stretch of the distance sketch. 

\begin{lemma}
    \label{lemma ds final stretch}
    Given a pair of vertices $u,v$ with $d_G(u,v)\leq d$, the algorithm returns a distance estimate that is at most $(1+\epsilon_\mch)(2k-1)d_G(u,v)$.  
\end{lemma}
\begin{proof}

By lemma \ref{lemma ui-wi mpc ds}, we have that Algorithm \ref{alg search TZ MPC} never returns \textit{$u,v$ out of range} for a query $u,v$ with $d_G(u,v)\leq d$. Therefore, the algorithm terminates only when the condition of the while loop is no longer satisfied, indicating that  $ w_{i}\in B(v_{i})$. In addition, Algorithm \ref{alg search TZ MPC} returns $\mathcal{D}(w_i,v_i)+\mathcal{D}(w_i,u_i)$. We now provide upper bounds on these terms. 

Note that $w_i$ is the $i$-pivot of $u_i$. Therefore, $$\mathcal{D}(w_i,u_i)\leq \dgmhb{i}(w_i,u_i).$$ By Lemma \ref{lemma ui-wi mpc ds}, we have 
\begin{equation}\label{eq Dui-wi}
    \begin{array}{lclclclclcl}
        \mathcal{D}(w_i,u_i)& \leq & \dgmhb{i}(w_i,u_i) & \leq &
        (1+\epsilon_\mathcal{H})\cdot i \cdot d_G(u,v).
    \end{array}
\end{equation} 

In addition, we also have 
\begin{equation}\label{eq dg wi-vi}
    \begin{array}{lclclclclcl}
        d_G(w_i,v_i) & \leq & t/(1+\epsilon_\mch).
    \end{array}
\end{equation} 
Since the hopset $\mch$ is a $(1+\epsilon_\mch,\beta_\mch)$-hopset for distances up to $t$, we have $\dghBH(w_i,v_i)\leq (1+\epsilon_\mch)d_G(w_i,v_i)$. 
Observe that since $w_i\in B(v_i)$, we have $\mathcal{D}(w_i,v_i)\leq \dghBH(w_i,v_i)$. Together with eq. \ref{eq dg wi-vi}, we have  
\begin{equation*}
    \begin{array}{lclclclclcl}
        \mathcal{D}(w_i,v_i)& \leq &\dghBH(w_i,v_i) 
        &\leq & 
        (1+\epsilon_\mch)d_G(w_i,v_i).
    \end{array}
\end{equation*} 

Observe that the triangle inequality holds for distances in $G$, and also that the hopset $\mch$ does not shorten distances w.r.t. the graph $G$. Together with eq. \ref{eq Dui-wi}
\begin{equation}\label{eq Dwi-vi}
    \begin{array}{lclclclclcl}
        \mathcal{D}(w_i,v_i)& \leq & 
        (1+\epsilon_\mch)(d_G(w_i,u_i)+d_G(u_i,v_i))
        \\& < & 
        (1+\epsilon_\mch)^2(i+1) d_G(u,v).
    \end{array}
\end{equation} 

Recall that $i\leq k-1$. 
By \ref{eq Dui-wi} and \ref{eq Dwi-vi}, we conclude 
\begin{equation*}
    \begin{array}{lclclclclcl}
       &&\mathcal{D}(w_i,u_i)+\mathcal{D}(w_i,v_i)\\& \leq & 
        (1+\epsilon_\mch) (k-1)  d_G(u,v) + 
        (1+\epsilon_\mch)^2\cdot k\cdot  d_G(u,v)\\
        & \leq & 
        (1+\epsilon_\mch) (2k-1)  d_G(u,v).
    \end{array}
\end{equation*} 
\end{proof}

\subsubsection{Distance Sketch Size}
The analysis of the size of the distance sketch is reminiscent of the analogous analysis of the distance sketches of Thorup and Zwick. The difference between the two analysis is only that in the construction of Thorup-Zwick, all vertices have an $i$-pivot for all levels $i\in [0,k-1]$, whereas in our algorithm, certain vertices may lack an $i$-pivot for certain levels. However, the absence of an $i$-pivot for a vertex $v\in V$ implies that the expected number of vertices $u\in A_{i-1}$ in the bunch of $v$ is no more than $n^{1/k}$. Therefore, the upper bound on the size of the distance sketch given by Thorup and Zwick is also applicable for our construction. For completeness, we analyze the size of our distance sketches.

First, we provide an upper bound on $|A_{k-1}|$. 
\begin{lemma}
    \label{lemma tz ak upper}
    The size of $A_{k-1}$ is, w.h.p., $\tilde{O}(n^{1/k})$.
\end{lemma}
\begin{proof}
    Recall that for every $i\in [0,k-2]$ each vertex $v\in A_i$ is sampled to $A_{i+1}$ with probability $n^{1/k}$. Thus, the expected size of $A_{k}$ is at most $\frac{n}{(n^{1/k})^{k-1}} = n^{1/k} $. 

    Note that $k$ is a constant, thus $n^{1/k}>1$. 
    Let $c>2$ be a constant. From a Chernoff bound, for $\delta = c{\ln n}$ we have 
    \begin{equation*}
    \begin{array}{lclclclclclclclc}
        Pr[ |A_{k-1}| > (1+c{\ln n}) E[|A_{k-1}|]] 
        &\leq & e^{-\frac{(c{\ln n})^2 E[|A_{k-1}|]}{2+c{\ln n}}} 
        \\&\leq& e^{-\frac{c{\ln n}}{2}} \\ & = &n^{-\frac{c}{2}}.          
    \end{array}
    \end{equation*}

    Thus, w.h.p., $|A_{k-1}| = \tilde{O}(n^{1/k})$. 
\end{proof}

Next, we upper bound the number of vertices from the bunch of every vertex $v$ that belong to level $i$ of the hierarchy. 
\begin{lemma}
\label{lemma traverses 2}
    Let $i\in [0,k-1]$. For every vertex $u\in V$, for any constant $c>1$, with probability at least $1-n^{-c}$, there are at most  $n^{1/k}\cdot c\cdot {\ln n}$ vertices $v\in A_i$ such that: 
    \begin{equation*}
        \begin{array}{lclclc}
             \dgmhb{(i+1)}(u,v) \leq  \dgmhb{(i+1)}(p_{i+1}(u),u) &\textit{if u has an (i+1)-pivot}\\
             
             \dgmhb{(i+1)}(u,v) \leq  t &\textit{otherwise}.
        \end{array}
    \end{equation*}
\end{lemma}
\begin{proof}
    First, for $i= k-1$, note that by Lemma \ref{lemma tz ak upper} ,  $|A_{k-1}| = \tilde{O}(n^{1/k})$, thus the claim trivially holds.

     Let $i\in [0,k-2]$, and let $u\in V$. We sort the vertices  $v\in A_i$ according to $\dgmhb{(i+1)}(u,v)$. Let $j$ be the minimal index of a vertex such that $v_j \in A_{i+1}$. Recall that vertices of $A_i$ are sampled to $A_{i+1}$ with probability $n^{1/k}$. Thus, the expected value of $j$ is no more than $n^{1/k}$. 

    Let $c>2$ be a constant. From a Chernoff bound, for $\delta = c{\ln n}$ we have 
    \begin{equation*}
    \begin{array}{lclclclclclclclc}
        Pr[ j > (1+c{\ln n}) E[j]] 
        &\leq & e^{-\frac{(c{\ln n})^2 E[j]}{2+c{\ln n}}} 
        \\&\leq& e^{-\frac{c{\ln n}}{2}}  \\& = &n^{-\frac{c}{2}}.          
    \end{array}
    \end{equation*}
    Let $X'$ be the set of vertices $v\in X$ such that $\dgmhb{(i+1)}(u,v) \leq \dgmhb{(i+1)}(u,v_j)$. With probability at least $n^{-\frac{c}{2}} $  the size of $X'$ is at most $(1+c{\ln n})n^{1/k} $.

Observe that the set of vertices $v\in A_i$ that satiety: 
    \begin{equation*}
        \begin{array}{lclclc}
             \dgmhb{(i+1)}(u,v) \leq \dgmhb{(i+1)}(p_{i+1}(u),u) &\textit{if u has an (i+1)-pivot}\\
             
             \dgmhb{(i+1)}(u,v) \leq  t &\textit{otherwise}
        \end{array}
    \end{equation*}
    is a subset of $X'$, thus the claim holds. 
\end{proof}

By Lemmas \ref{lemma tz ak upper} and  \ref{lemma traverses 2}, we derive the following two lemmas.

 \begin{lemma}\label{lemma ds size single}
     For every $v\in V$ and $i\in [0,k-1]$, w.h.p., we have $|B(v) \cap (\aismai)| \leq  \tilde{O}(n^{1/k})$. 
 \end{lemma}

 \begin{lemma}\label{lemma ds size total}
     For every $v\in V$, w.h.p., we have $|B(v)| \leq  \tilde{O}(kn^{1/k})$. 
 \end{lemma}

\subsubsection{Implementation and Round Complexity}
This section is divided into two segments. First, we discuss the implementation of the preprocessing algorithm \ref{alg preprocessing TZ MPC}. Then, we discuss the round complexity of answering a query using Algorithm \ref{alg search TZ MPC}. 
\paragraph{Preprocessing time.}
The input for the preprocessing algorithm is the graph $G$, a $(1+\epsilon_\mch,\beta_\mch)$-hopset $\mathcal{H}$ of size $O(n^{1+\rho})$ for distances up to $t$ in $G$ and parameters $k \geq 2$ and $\gamma\in (0,1)$.  

The first step in the preprocessing algorithm is sampling the hierarchy $V= A_0\supseteq A_1 \supseteq \dots \supseteq A_{k-1} \supseteq A_k = \emptyset$. Computationally, this is equivalent to the hierarchy sampling procedure of the general framework. Therefore, Lemma \ref{lemma cc select} is also applicable for the sampling procedure of Algorithm \ref{alg preprocessing TZ MPC}, i.e., it can be executed in $O(1/\gamma)$ rounds of MPC using machines of $O(n^\gamma)$ memory and a total memory of size $\tilde{O}(|E|+n^{1+\rho})$.

Next, the preprocessing algorithm moves to compute the $i$-pivots and the bunches of every vertex $v\in V$. First, pivots are computed for all scales. For $i\in [0,k-1]$, a single Bellman-Ford exploration from all sources in $A_i$ is executed to $(i\beta_\mch)$-hops in $G\cup \mch$. Each vertex $v\in V$ that is detected by this exploration, selects the vertex $w\in A_i$ with minimal $\dgmhb{i}(w,v)$ to be its $i$-pivot. Observe that we are interested only in finding the closest $A_i$-vertex for every $v\in V$, thus we do not compute the distances from $v$ to all vertices of $A_i$. As discussed in Section \ref{sec selecting edges impl}, this is computationally equivalent to executing a single restricted-Bellman-Ford exploration from a dummy vertex $s$ which is connected with $0$-weight edges to all vertices in $A_i$. Therefore, by Theorem \ref{theo restricted bellman-ford}, executing $\beta_\mch$-restricted Bellman-Ford explorations from all vertices in $A_i$ can be done in $O(\beta_\mch/\gamma)$ rounds of sublinear MPC, when using machines of size $O(n^\gamma)$ and a total memory of size $\tilde{O}(|E|+n^{1+\rho})$.

For computing the bunches, for every $i\in[0,k-1]$, we execute $(i+1)\beta_\mch$-hops restricted Bellman-Ford explorations from all $A_i$-vertices, in parallel. For this aim, we employ the multiple sources Bellman-Ford algorithm devised in Section \ref{sec selecting edges impl}.
We use the term $A_i$-explorations to describe all explorations that are executed from vertices in $A_i$ to depth $(i+1)\beta_\mch$ in order to compute bunches. 
Observe however that the set $A_i$ may be of size $\Omega(n)$. Thus, we need to provide an upper bound on the number of $A_i$-explorations that traverse each vertex. For this aim, we observe that each vertex $v\in V$ is interested only in distances to vertices $u\in A_i$ such that $\dgmhb{(i+1)}(u,v) < \dgmhb{(i+1)}(p_{i+1}(v),v)$. Therefore, throughout the $A_i$-exploration, each vertex $v\in V$ needs to send messages to its neighbors only regarding vertices $u\in A_i$ only when 
$\dgmhb{(i+1)}(u,v) < \dgmhb{(i+1)}(p_{i+1}(v),v)$.

We slightly modify the multiple sources Bellman-Ford algorithm, such that vertices do not receive messages with distance estimates larger than their individual threshold. 
Essentially,  for each vertex $v\in V$, we set an individual distance threshold $\delta_{i,v} = \dgmhb{(i+1)}(p_{i+1}(v),v)$. Each vertex $v$ informs all its neighbors of its threshold. These vertices check whether the distance estimate they want to send $v$ comply with the distance threshold of $v$ or not.

Formally, we create $\mu = n^{1/k}c{\ln n}$ copies for every edge $(u,v)\in E\cup H$, using the copying procedure devised in Section \ref{sec selecting edges impl}, and compute appropriate tuples to allow communication between the $j$th copy of $(u,v)$ and the $j$th copy of $(v,u)$. 
For each edge-copy $((v,u),j)$ where $(u,v)\in G\cup \mch$ and $j\in [0,\addedmemory -1]$, the machines $M_{((u,v),j)}$ sends this threshold to the machine $M_{((v,u),j)}$. Now, whenever the machine $M_{((u,v),j)}$ 
wishes to send a message to $M_{(v,u),j)}$ regarding a distance estimate from a source $s$ to $u$, it first checks whether this estimate is at most $\delta_{i,v} - \omega(u,v)$, 
where $\omega(u,v)$ is the weight of the edge $(u,v)$. This ensures that each vertex $v\in V$  is traversed only by $A_i$-explorations such that their origin $s$ satisfies $\dgmhb{(i+1)}(s,v) < \dgmhb{(i+1)}(p_{i+1}(v),v)$.

By Lemma \ref{lemma traverses 2}, each vertex $v\in V$ has at most $\tilde{O}(n^{1/k})$ such vertices. Therefore, by Theorem \ref{theo multi source BF}, w.h.p., the $A_i$-explorations can be executed in parallel in $O(i\beta_\mch/\gamma)$ rounds of sublinear MPC using machines with $O(n^\gamma)$ memory and $\tilde{O}((|E|+|\mch|)n^{1/k})$ total memory. Recall that w.h.p., $|\mch| = O(n^{1+\rho})$. 
It follows that the time required to execute the $A_i$-explorations for every $i\in [0,k-1]$ is at most 
\begin{equation*}
    \begin{array}{lclclc}
         \sum_{i\in [0,k-1]} O(i\beta_\mch/\gamma) & = & O(k^2\beta_\mch/\gamma). 
    \end{array}
\end{equation*}
Observe that when  the algorithm terminates, for each vertex $v\in V$ and edge edge $(u,v)$ that the vertex $v$ adds to the distance sketch $\mathcal{D}$, there is a single machine in $M(v)$ that is aware of the edge $ (u,v)$ and its weight. In other words, the distances to $B(v)$ is stored on $M(v)$. 
The following lemma summarizes this discussion. 

\begin{lemma}
    \label{lemma apsp cc}
    The limited-scale distance sketches can be constructed in $O(k^2\beta_\mch/\gamma)$ rounds of sublinear MPC using machines with $O(n^\gamma)$ memory and $\tilde{O}((|E|+n^{1+\rho})n^{1/k})$ total memory. 
\end{lemma}

\paragraph{Query time }
Given a query $(u,v)$, the algorithm executes the Algorithm \ref{alg search TZ MPC}. If we have a machine $M^*$ with $\Omega(kn^{1/k} \poly{ \log(n)})$ memory, then the machines $M(u)$ and the machines $M(v)$ can send  $M^*$ the distances to their bunches that they store in $O(1)$ rounds. The machine $M^*$ can now execute Algorithm \ref{alg search TZ MPC} locally.

Consider the case where we are not allowed to use a machine with $\Omega(kn^{1/k} \poly{\log(n)})$ memory. In this case, for every $i$, the machines in $M(v_i)$ check if $w_i\in B(v_i)$.
For this aim, if a machine in $M(v_i)$ stores $\mathcal{D}(v_i,w_i)$, then it broadcasts this message to all machines in $M(v_i),M(u_i)$. Observe that there is at most one machine that broadcasts $\mathcal{D}(v_i,w_i)$. 
If no machine has sent such a broadcast, then the machines in $M(v_i),M(u_i)$ compute $u_{i+1},v_{i+1} = v_i,u_i$ locally, and the machines $M(u_{i+1})$ find $w_{i+1}$, and send its ID to $M(u_i),M(v_i)$.  Recall that the $(i+1)$-pivot of $u_{i+1}$ is stored on a single machine (if exists). Therefore, by Theorem \ref{theo min broadcast}, each round of the query algorithm can be computed in $O(1/\gamma)$ rounds of sublinear MPC when using machines of size $O(n^\gamma)$  and $\tilde{O}((|E|+n^{1+\rho})n^{1/k})$ total memory.

\paragraph{Summary}
 From Lemmas \ref{lemma ds final stretch}, \ref{lemma ds size total} and \ref{lemma apsp cc} we derive the following corollary. 

\begin{corollary}\label{coro distance sketch prop}
    There is a randomized algorithm that receives a graph $G = (V,E)$ on $n$ vertices, parameters $\epsilon <1$,  $\rho \in [1/{\log {\log n}},1/2]$, $k\leq 1/\rho$, a limited-scale $(1+\epsilon_\mch,\beta_\mch)$-hopset for distances up to $t$, and w.h.p. computes a limited distance sketch in 
    $O(k^2\beta_\mch/\gamma)$
     rounds of sublinear MPC,
      when using machines with $O(n^\gamma)$ memory and a total memory of size $\tilde{O}((|E|+n^{1+\rho})n^{1/k})$.
      For each vertex $v\in V$, the distance sketch retains $\tilde{O}(kn^{1/k})$ distances. 
\end{corollary}

\begin{corollary}
    \label{coro query constant}
    Given a pair of vertices $u,v$ such that $d_G(u,v)\leq \dlimits$, the distance sketch can be used to compute a $(1+\epsilon_\mch)(2k-1)$-approximation for $d_G(u,v)$ in $O(1)$ rounds when using at least one machine with $\Omega(kn^{1/k}\poly{\log(n)})$ memory, and additional machines with $O(n^\gamma)$ memory each, such that the total memory is $\tilde{O}((|E|+|\mch|)n^{1/k})$. 
\end{corollary}

\begin{corollary}
    \label{coro query k}
    Given a pair of vertices $u,v$ such that $d_G(u,v)\leq \dlimits$, the distance sketch can be used to compute a $(1+\epsilon_\mch)(2k-1)$-approximation for $d_G(u,v)$ in $O(k)$ rounds when using  machines with $O(n^\gamma)$ memory and $\tilde{O}((|E|+|\mch|)n^{1/k})$ total memory. 
\end{corollary}

\subsubsection{Computing APSP}

In order to compute a distance oracle for all distance scales, we employ the algorithm devised in Section \ref{sec sssp} to construct a limited-scale near-exact hopset $\mch$ and a near-additive emulator $M$.
Recall that $\mch$ is a $(1+\epsilon_\mch,\beta_\mch)$-hopset for distances up to $t = \tval$. Also, recall that the emulator is stored on a machine $M^*$ with $\tilde{O}(n)$ memory. 

Then, we use the hopset $\mch$ and construct a distance sketch for distances up to $d = \dlimits$. 

Upon a distance query $u,v$, the machine $M^*$ computes $d_{M}(u,v)$. In addition, the machines $M(u)$ and $M(v)$ send $B(u),B(v)$ to the machine $M^*$, which in turn runs Algorithm \ref{alg search TZ MPC} locally. Finally, the machine $M^*$ chooses the minimal estimate as its final answer.

Recall that $M$ is a $(1+\epsilon_\mch+ \epsilon_{M},\beta_{M})$-emulator. Also, recall that 
\begin{equation*}
    \begin{array}{lclclclc}
         \beta_{M} & = & \betaemu\\
         t & = & \tval &\rightarrow\\
         \beta_{M} & \leq & t\cdot \epsilon_{M}
  \end{array}
\end{equation*}
Observe that for pairs of vertices $u,v$ with $d_G(u,v)\geq \dlimits$, we have 

\begin{equation*}
    \begin{array}{lclclclc}
        &&(1+\epsilon_\mch+\epsilon_{M})d_G(u,v) + \beta_{M} \\
        & \leq  & 
        (1+2\epsilon_{M})d_G(u,v) + t\cdot {\epsilon_{M}}\\
        & \leq  & (1+2\epsilon_{M})d_G(u,v) +d_G(u,v)\cdot(1+\epsilon_\mch)^2(k+1) \cdot {\epsilon_{M}}\\
        & =   & (1+2\epsilon_{M}+ (1+\epsilon_\mch)^2\cdot {\epsilon_{M}} (k+1) )d_G(u,v)

    \end{array}
\end{equation*}

Recall that $\epsilon_\mch = \epsilon_{M} = \epsilon/3$. In addition, we require $\epsilon <1/2$.  It follows that $(1+\epsilon_\mch)^2\cdot \epsilon_{M} \leq 3/10 $. Therefore, we have 

\begin{equation*}
    \begin{array}{lclclclc}
        (1+\epsilon_\mch+\epsilon_{M})d_G(u,v) + \beta_{M} & \leq  &
        (1+\frac{1}{3}+ \frac{3}{10} (k+1) )d_G(u,v)
\\
        & \leq  &
        (\frac{49}{30} + \frac{3k}{10} )d_G(u,v)
\\
        & \leq  &
        (2k-1)d_G(u,v).
    \end{array}
\end{equation*}
It follows that the emulator $M$ provides a $(2k-1)$-approximation for distanced at least $d$.

Observe that the running time of the algorithm is dominated by the time required to compute the hopset $\mch$ and the emulator $M$, as in Section \ref{sec sssp}. The total memory is dominated by the $\tilde{O}((|E|+n^{1+\rho})n^{1/k})$ memory required to compute the distance sketch. 
By Theorem \ref{theo sssp}, Corollary \ref{coro distance sketch prop} and \ref{coro query constant}, we derive the following corollary which summarizes the properties of the distance oracle.

\begin{theorem}
    \label{theo distance oracles} 
    Let $G=(V,E)$ be an unweighted, undirected graph on $n$ vertices, and let $\epsilon <1/2$, $\rho \in[ 1/{\log {\log n}}, 1/2]$ and $k \leq 1/\rho$ be parameters. There is a randomized algorithm that w.h.p.  computes a distance oracle of size $\tilde{O}(kn^{1+1/k})$ that provides $(1+\epsilon)(2k-1)$-approximation for all the distances.  The preprocessing algorithm can be executed in  $ O\left( \frac{1}{\gamma}\left( \frac{ 
        {\log^{(2)} n}\cdot ({\log^{(3)}n} -{\log {\epsilon }})
        }
        {\epsilon \rho}\right)^{\frac{1}{\rho} +2}\right) $
     rounds  and the query time is $O(1)$. The algorithm works in heterogeneous MPC, when using a single machine with $\tilde{O}(n)$ memory and additional machines with $O(n^\gamma)$ memory and a total memory of size $\tilde{O}((|E|+n^{1+\rho})\cdot n^{1/k})$.
\end{theorem}

\bibliographystyle{alpha}

\bibliography{cite}

\newcommand{\etalchar}[1]{$^{#1}$}
\begin{thebibliography}{GGK{\etalchar{+}}18}

\bibitem[ABB{\etalchar{+}}19]{assadi2019coresets}
Sepehr Assadi, MohammadHossein Bateni, Aaron Bernstein, Vahab Mirrokni, and Cliff Stein.
\newblock Coresets meet edcs: algorithms for matching and vertex cover on massive graphs.
\newblock In {\em SODA}, pages 1616--1635. SIAM, 2019.

\bibitem[ABP18]{abboud2018hierarchy}
Amir Abboud, Greg Bodwin, and Seth Pettie.
\newblock A hierarchy of lower bounds for sublinear additive spanners.
\newblock {\em SIAM Journal on Computing}, 47(6):2203--2236, 2018.

\bibitem[ANOY14]{andoni2014parallel}
Alexandr Andoni, Aleksandar Nikolov, Krzysztof Onak, and Grigory Yaroslavtsev.
\newblock Parallel algorithms for geometric graph problems.
\newblock In {\em STOC}, pages 574--583. ACM, 2014.

\bibitem[ASZ20]{andoni2020parallel}
Alexandr Andoni, Clifford Stein, and Peilin Zhong.
\newblock Parallel approximate undirected shortest paths via low hop emulators.
\newblock In {\em STOC}, pages 322--335, 2020.

\bibitem[BBD{\etalchar{+}}19]{behnezhad2019massively}
Soheil Behnezhad, Sebastian Brandt, Mahsa Derakhshan, Manuela Fischer, MohammadTaghi Hajiaghayi, Richard~M Karp, and Jara Uitto.
\newblock Massively parallel computation of matching and mis in sparse graphs.
\newblock In {\em PODC}, pages 481--490, 2019.

\bibitem[BDG{\etalchar{+}}21]{DBLP:conf/spaa/BiswasDGMN21}
Amartya~Shankha Biswas, Michal Dory, Mohsen Ghaffari, Slobodan Mitrovic, and Yasamin Nazari.
\newblock Massively parallel algorithms for distance approximation and spanners.
\newblock In {\em SPAA}, pages 118--128, 2021.

\bibitem[BHG{\etalchar{+}}21]{bergamaschi2021new}
Thiago Bergamaschi, Monika Henzinger, Maximilian~Probst Gutenberg, Virginia~Vassilevska Williams, and Nicole Wein.
\newblock New techniques and fine-grained hardness for dynamic near-additive spanners.
\newblock In {\em Proceedings of the 2021 ACM-SIAM Symposium on Discrete Algorithms (SODA)}, pages 1836--1855. SIAM, 2021.

\bibitem[BHH19]{behnezhad2019exponentially}
Soheil Behnezhad, Mohammad~Taghi Hajiaghayi, and David~G Harris.
\newblock Exponentially faster massively parallel maximal matching.
\newblock In {\em FOCS}, pages 1637--1649. IEEE, 2019.

\bibitem[BKS17]{beame2017communication}
Paul Beame, Paraschos Koutris, and Dan Suciu.
\newblock Communication steps for parallel query processing.
\newblock {\em Journal of the ACM (JACM)}, 64(6):40, 2017.

\bibitem[CDP21]{czumaj2021simple}
Artur Czumaj, Peter Davies, and Merav Parter.
\newblock Simple, deterministic, constant-round coloring in congested clique and mpc.
\newblock {\em SIAM journal on computing}, 50(5):1603--1626, 2021.

\bibitem[CFG{\etalchar{+}}19]{chang2019complexity}
Yi-Jun Chang, Manuela Fischer, Mohsen Ghaffari, Jara Uitto, and Yufan Zheng.
\newblock The complexity of ($\delta$+ 1) coloring in congested clique, massively parallel computation, and centralized local computation.
\newblock In {\em PODC}, pages 471--480, 2019.

\bibitem[C{\L}M{\etalchar{+}}18]{czumaj2018round}
Artur Czumaj, Jakub {\L}{\k{a}}cki, Aleksander M{\k{a}}dry, Slobodan Mitrovi{\'c}, Krzysztof Onak, and Piotr Sankowski.
\newblock Round compression for parallel matching algorithms.
\newblock In {\em STOC}, pages 471--484, 2018.

\bibitem[DFKL21]{DBLP:conf/podc/DoryFKL21}
Michal Dory, Orr Fischer, Seri Khoury, and Dean Leitersdorf.
\newblock Constant-round spanners and shortest paths in congested clique and {MPC}.
\newblock In {\em PODC}, pages 223--233, 2021.

\bibitem[DG08]{dean2008mapreduce}
Jeffrey Dean and Sanjay Ghemawat.
\newblock Mapreduce: simplified data processing on large clusters.
\newblock {\em Communications of the ACM}, 51(1):107--113, 2008.

\bibitem[DN19]{DinitzN19}
Michael Dinitz and Yasamin Nazari.
\newblock Massively parallel approximate distance sketches.
\newblock In Pascal Felber, Roy Friedman, Seth Gilbert, and Avery Miller, editors, {\em 23rd International Conference on Principles of Distributed Systems, {OPODIS} 2019, December 17-19, 2019, Neuch{\^{a}}tel, Switzerland}, volume 153 of {\em LIPIcs}, pages 35:1--35:17. Schloss Dagstuhl - Leibniz-Zentrum f{\"{u}}r Informatik, 2019.

\bibitem[DP22]{DBLP:conf/podc/DoryP20}
Michal Dory and Merav Parter.
\newblock Exponentially faster shortest paths in the congested clique.
\newblock {\em J. {ACM}}, pages 29:1--29:42, 2022.

\bibitem[EM19]{DBLP:conf/podc/ElkinM19}
Michael Elkin and Shaked Matar.
\newblock Near-additive spanners in low polynomial deterministic {CONGEST} time.
\newblock In Peter Robinson and Faith Ellen, editors, {\em Proceedings of the 2019 {ACM} Symposium on Principles of Distributed Computing, {PODC} 2019, Toronto, ON, Canada, July 29 - August 2, 2019}, pages 531--540. {ACM}, 2019.

\bibitem[EM21]{DBLP:conf/podc/ElkinM21}
Michael Elkin and Shaked Matar.
\newblock Ultra-sparse near-additive emulators.
\newblock In {\em {PODC} '21: {ACM} Symposium on Principles of Distributed Computing}, pages 235--246. {ACM}, 2021.

\bibitem[EN19]{DBLP:journals/talg/ElkinN19}
Michael Elkin and Ofer Neiman.
\newblock Efficient algorithms for constructing very sparse spanners and emulators.
\newblock {\em {ACM} Trans. Algorithms}, 15(1):4:1--4:29, 2019.

\bibitem[EN20]{DBLP:journals/eatcs/ElkinN20}
Michael Elkin and Ofer Neiman.
\newblock Near-additive spanners and near-exact hopsets, {A} unified view.
\newblock {\em Bull. {EATCS}}, 130, 2020.

\bibitem[EP01]{ElkinP01}
Michael Elkin and David Peleg.
\newblock (1+epsilon, beta)-spanner constructions for general graphs.
\newblock In Jeffrey~Scott Vitter, Paul~G. Spirakis, and Mihalis Yannakakis, editors, {\em Proceedings on 33rd Annual {ACM} Symposium on Theory of Computing, July 6-8, 2001, Heraklion, Crete, Greece}, pages 173--182. {ACM}, 2001.

\bibitem[ET22]{DBLP:conf/approx/ElkinT22}
Michael Elkin and Chhaya Trehan.
\newblock (1+{\(\epsilon\)})-approximate shortest paths in dynamic streams.
\newblock In {\em Approximation, Randomization, and Combinatorial Optimization. Algorithms and Techniques, {APPROX/RANDOM} 2022}, volume 245 of {\em LIPIcs}, pages 51:1--51:23. Schloss Dagstuhl - Leibniz-Zentrum f{\"{u}}r Informatik, 2022.

\bibitem[FHO22]{fischer2022massively}
Orr Fischer, Adi Horowitz, and Rotem Oshman.
\newblock Massively parallel computation in a heterogeneous regime.
\newblock In {\em PODC}, pages 345--355, 2022.

\bibitem[GGK{\etalchar{+}}18]{ghaffari2018improved}
Mohsen Ghaffari, Themis Gouleakis, Christian Konrad, Slobodan Mitrovi{\'c}, and Ronitt Rubinfeld.
\newblock Improved massively parallel computation algorithms for mis, matching, and vertex cover.
\newblock In {\em PODC}, pages 129--138, 2018.

\bibitem[GN20]{ghaffari2020massively}
Mohsen Ghaffari and Krzysztof Nowicki.
\newblock Massively parallel algorithms for minimum cut.
\newblock In {\em PODC}, pages 119--128, 2020.

\bibitem[GSZ11]{GoodrichSZ11}
Michael~T. Goodrich, Nodari Sitchinava, and Qin Zhang.
\newblock Sorting, searching, and simulation in the mapreduce framework.
\newblock In Takao Asano, Shin{-}Ichi Nakano, Yoshio Okamoto, and Osamu Watanabe, editors, {\em Algorithms and Computation - 22nd International Symposium, {ISAAC} 2011, Yokohama, Japan, December 5-8, 2011. Proceedings}, volume 7074 of {\em Lecture Notes in Computer Science}, pages 374--383. Springer, 2011.

\bibitem[GU19]{ghaffari2019sparsifying}
Mohsen Ghaffari and Jara Uitto.
\newblock Sparsifying distributed algorithms with ramifications in massively parallel computation and centralized local computation.
\newblock In {\em SODA}, pages 1636--1653, 2019.

\bibitem[HLSS19]{hajiaghayi2019mapreduce}
MohammadTaghi Hajiaghayi, Silvio Lattanzi, Saeed Seddighin, and Cliff Stein.
\newblock Mapreduce meets fine-grained complexity: Mapreduce algorithms for apsp, matrix multiplication, 3-sum, and beyond.
\newblock {\em arXiv preprint arXiv:1905.01748}, 2019.

\bibitem[IBY{\etalchar{+}}07]{isard2007dryad}
Michael Isard, Mihai Budiu, Yuan Yu, Andrew Birrell, and Dennis Fetterly.
\newblock Dryad: distributed data-parallel programs from sequential building blocks.
\newblock In {\em ACM SIGOPS operating systems review}, pages 59--72. ACM, 2007.

\bibitem[KP22]{DBLP:conf/focs/KoganP22}
Shimon Kogan and Merav Parter.
\newblock Having hope in hops: New spanners, preservers and lower bounds for hopsets.
\newblock In {\em 63rd {IEEE} Annual Symposium on Foundations of Computer Science, {FOCS} 2022, Denver, CO, USA, October 31 - November 3, 2022}, pages 766--777. {IEEE}, 2022.

\bibitem[KSV10]{karloff2010model}
Howard Karloff, Siddharth Suri, and Sergei Vassilvitskii.
\newblock A model of computation for mapreduce.
\newblock In {\em SODA}, pages 938--948. SIAM, 2010.

\bibitem[Li20]{li2020faster}
Jason Li.
\newblock Faster parallel algorithm for approximate shortest path.
\newblock In {\em STOC}, pages 308--321, 2020.

\bibitem[LMSV11]{lattanzi2011filtering}
Silvio Lattanzi, Benjamin Moseley, Siddharth Suri, and Sergei Vassilvitskii.
\newblock Filtering: a method for solving graph problems in mapreduce.
\newblock In {\em SPAA}, pages 85--94, 2011.

\bibitem[Now21]{nowicki2021deterministic}
Krzysztof Nowicki.
\newblock A deterministic algorithm for the mst problem in constant rounds of congested clique.
\newblock In {\em STOC}, pages 1154--1165, 2021.

\bibitem[Pet10]{Pettie10}
Seth Pettie.
\newblock Distributed algorithms for ultrasparse spanners and linear size skeletons.
\newblock {\em Distributed Computing}, 22(3):147--166, 2010.

\bibitem[TZ05]{ThorupZ05}
Mikkel Thorup and Uri Zwick.
\newblock Approximate distance oracles.
\newblock {\em J. {ACM}}, 52(1):1--24, 2005.

\bibitem[TZ06]{ThorupZ06}
Mikkel Thorup and Uri Zwick.
\newblock Spanners and emulators with sublinear distance errors.
\newblock In {\em Proceedings of the Seventeenth Annual {ACM-SIAM} Symposium on Discrete Algorithms, {SODA} 2006, Miami, Florida, USA, January 22-26, 2006}, pages 802--809. {ACM} Press, 2006.

\bibitem[Whi12]{white2012hadoop}
Tom White.
\newblock {\em Hadoop: The definitive guide}.
\newblock " O'Reilly Media, Inc.", 2012.

\bibitem[ZCF{\etalchar{+}}10]{zaharia2010spark}
Matei Zaharia, Mosharaf Chowdhury, Michael~J Franklin, Scott Shenker, and Ion Stoica.
\newblock Spark: Cluster computing with working sets.
\newblock {\em HotCloud}, 10(10-10):95, 2010.

\end{thebibliography}

\end{document}